\documentclass[twocolumn]{IEEEtran}

\usepackage{amsmath,epsfig,amssymb,verbatim,amsopn,cite,multirow}
\usepackage{amsthm}
\usepackage{balance}
\usepackage{multirow}
\usepackage{algorithm} 
\usepackage{algpseudocode} 
\usepackage[usenames,dvipsnames]{color}
\usepackage[all]{xy}  
\usepackage{url}
\usepackage{amsfonts}
\usepackage{amssymb}
\usepackage{epsfig}
\usepackage{epstopdf}
\usepackage{bm}
\usepackage{balance}
\usepackage{graphicx}
\usepackage{subcaption}
\usepackage{footnote}
\usepackage{cancel}
\usepackage{lipsum}
\usepackage{mathtools}
\usepackage{cuted}
\usepackage{multicol}
\usepackage{stfloats}
\usepackage[normalem]{ulem}

\usepackage[nodisplayskipstretch]{setspace}

\newtheorem{Remark}{Remark}

\newtheorem{proposition}{Proposition}

\newcommand{\qa}{{\bf a}}
\newcommand{\qb}{{\bf b}}

\newcommand{\qg}{{\bf g}}
\newcommand{\qh}{{\bf h}}

\newcommand{\qp}{{\bf p}}
\newcommand{\qq}{{\bf q}}

\newcommand{\qs}{{\bf s}}

\newcommand{\qu}{{\bf u}}
\newcommand{\qv}{{\bf v}}
\newcommand{\qw}{{\bf w}}
\newcommand{\qx}{{\bf x}}

\newcommand{\qA}{{\bf A}}
\newcommand{\qB}{{\bf B}}
\newcommand{\qC}{{\bf C}}
\newcommand{\qD}{{\bf D}}

\newcommand{\qG}{{\bf G}}
\newcommand{\qH}{{\bf H}}
\newcommand{\qI}{{\bf I}}

\newcommand{\qM}{{\bf M}}

\newcommand{\qR}{{\bf R}}

\newcommand{\qV}{{\bf V}}

\newcommand{\LZF}{\mathtt{LZF}}

\newcommand{\MRT}{\mathtt{MRT}}

\newcommand{\CZF}{\mathtt{CZF}}

\newcommand{\Ex}{\mathbb{E}}

\newcommand{\Kn}{\mathcal{K_N}}
\newcommand{\Kf}{\mathcal{K_F}}
\newcommand{\US}{\mathcal{U}_s}
\newcommand{\IS}{\mathcal{I}_s}

\newcommand{\Sa}{\mathcal{S}}

\newcommand{\SknMRT}{\bar{\mathcal{S}}_k^{\MRT}}
\newcommand{\SniMRT}{\bar{\mathcal{S}}_i^{\MRT}}
\newcommand{\SjfMRT}{\Tilde{\mathcal{S}}_j^{\MRT}}
\newcommand{\SkfMRT}{\Tilde{\mathcal{S}}_k^{\MRT}}

\newcommand{\SknZF}{\bar{\mathcal{S}}_k^{\CZF}}
\newcommand{\SinZF}{\bar{\mathcal{S}}_i^{\CZF}}
\newcommand{\SjfZF}{\Tilde{\mathcal{S}}_j^{\CZF}}
\newcommand{\SkfZF}{\Tilde{\mathcal{S}}_k^{\CZF}}
\newcommand{\SniZF}{\bar{\mathcal{S}}_i^{\CZF}}

\newcommand{\SknLZF}{\bar{\mathcal{S}}_k^{\LZF}}
\newcommand{\SinLZF}{\bar{\mathcal{S}}_i^{\LZF}}
\newcommand{\SjfLZF}{\Tilde{\mathcal{S}}_j^{\LZF}}
\newcommand{\SkfLZF}{\Tilde{\mathcal{S}}_k^{\LZF}}
\newcommand{\SniLZF}{\bar{\mathcal{S}}_i^{\LZF}}

\newcommand{\BXiN}{\bar{\boldsymbol{\xi}}}
\newcommand{\BXiF}{\tilde{\boldsymbol{\xi}}}

\newcommand{\BXiNMRT}{\bar{\boldsymbol{\xi}}^{\MRT}}
\newcommand{\BXiFMRT}{\tilde{\boldsymbol{\xi}}^{\MRT}}

\newcommand{\BEtaN}{\bar{\boldsymbol{\eta}}}
\newcommand{\BEtaF}{\tilde{\boldsymbol{\eta}}}

\newcommand{\Sin}{\text{sin}}
\newcommand{\Cos}{\text{cos}}

\newcommand{\OO}{\mathcal{O}}

\newcommand{\betakks}{\bar{\Psi}_{kk}^s}
\newcommand{\betakis}{\bar{\Psi}_{ki}^s}

\newcommand{\betakispcj}{\big(\bar{\Psi}_{ki}^{s'}\big)^*}
\newcommand{\Tcks}{\Tilde{\Psi}_{k}^s}

\newcommand{\Taks}{\Tilde{a}_{k}^{s}}
\newcommand{\Talphakjssp}{\Tilde{\Xi}_{kj}^{ss'}}
\newcommand{\Tbkjssp}{\Tilde{b}_{kj}^{ss'}}
\newcommand{\Tdkjssp}{\Tilde{d}_{kj}^{ss'}}
\newcommand{\balphakissp}{\bar{\Xi}_{ki}^{ss'}}

\newcommand{\betasi}{\bar{\eta}_{si}}
\newcommand{\betask}{\bar{\eta}_{sk}}
\newcommand{\betaspi}{\bar{\eta}_{s'i}}
\newcommand{\Tetasj}{\Tilde{\eta}_{sj}}

\newcommand{\Tetask}{\Tilde{\eta}_{sk}}
\newcommand{\Tetaspj}{\Tilde{\eta}_{s'j}}

\newcommand{\betakCZF}{\bar{\eta}_{k}^{\CZF}}
\newcommand{\betaiCZF}{\bar{\eta}_{i}^{\CZF}}
\newcommand{\TetakCZF}{\Tilde{\eta}_k^{\CZF}}
\newcommand{\TetajCZF}{\Tilde{\eta}_j^{\CZF}}

\newcommand{\Txisj}{\Tilde{\xi}_{sj}}
\newcommand{\Txispj}{\Tilde{\xi}_{s'j}}
\newcommand{\bxisi}{\bar{\xi}_{si}}
\newcommand{\bxispi}{\bar{\xi}_{s'i}}

\newcommand{\bxisk}{\bar{\xi}_{sk}}
\newcommand{\Txisk}{\Tilde{\xi}_{sk}}

\newcommand{\Bbpsi}{\bar{\boldsymbol{\psi}}}
\newcommand{\Btpsi}{\tilde{\boldsymbol{\psi}}}

\newcommand{\bpsi}{\bar{\psi}}
\newcommand{\bpsisk}{\bpsi_{sk}}
\newcommand{\bpsisi}{\bpsi_{si}}
\newcommand{\bpsispi}{\bpsi_{s'i}}

\newcommand{\tpsi}{\tilde{\psi}}
\newcommand{\tpsisk}{\tpsi_{sk}}
\newcommand{\tpsisj}{\tpsi_{sj}}
\newcommand{\tpsispj}{\tpsi_{s'j}}

\newcommand{\Ttaukjsspi}{\Tilde{\tau}_{kj}^{ss'}}
\newcommand{\btaukisspi}{\bar{\tau}_{ki}^{ss'}}

\newcommand{\tSEk}{\widetilde{\text{SE}}_{k}}
\newcommand{\bSEk}{\overline{\text{SE}}_{k}}

\newcommand{\tSEkCZF}{\tSEk^{\CZF}}
\newcommand{\tSEkMRT}{\tSEk^{\MRT}}
\newcommand{\tSEkLZF}{\tSEk^{\LZF}}

\newcommand{\bwsiMRT}{\bar{\qw}_{si}^{\MRT}}

\newcommand{\bSEkCZF}{\bSEk^{\CZF}}
\newcommand{\bSEkMRT}{\bSEk^{\MRT}}
\newcommand{\bSEkLZF}{\bSEk^{\LZF}}

\newcommand{\tMRT}{\Tilde{T}^{\MRT}}
\newcommand{\bMRT}{\bar{T}^{\MRT}}
\newcommand{\tCZF}{\Tilde{T}^{\CZF}}
\newcommand{\bCZF}{\bar{T}^{\CZF}}
\newcommand{\tLZF}{\Tilde{T}^{\LZF}}
\newcommand{\bLZF}{\bar{T}^{\LZF}}

\usepackage{tikz}
\newcommand{\tikzxmark}{%
\tikz[scale=0.15] {
    \draw[line width=0.6,line cap=round] (0,0) to [bend left=5] (1,1);
    \draw[line width=0.6,line cap=round] (0.1,0.85) to [bend right=3] (0.8,0.05);
}}

\newcommand{\tSEkth}{\tilde{\mathcal{R}}_{k}}
\newcommand{\bSEkth}{\bar{\mathcal{R}}_{k}}
\newcommand{\bSINRk}{\overline{\text{SINR}}_k}
\newcommand{\tSINRk}{\widetilde{\text{SINR}}_k}
\newcommand{\tSINRkMRT}{\tSINRk^{\MRT}}
\newcommand{\bSINRkMRT}{\bSINRk^{\MRT}}

\newcommand{\tr}{\mathtt{tr}}

\newcommand{\Set}{\mathcal{S}}

\newcommand{\bTeta}{\boldsymbol{\Theta}}

\newcommand{\tBUk}{\widetilde{\mathrm{BU}}_k}
\newcommand{\tDSk}{\widetilde{\mathrm{DS}}_k}

\newcommand{\tUIj}{\widetilde{\mathrm{UI}}_{j}}
\newcommand{\tbUIi}{\widetilde{\mathrm{UI}}_{i}}

\newcommand{\bDSk}{\overline{\mathrm{DS}}_k}
\newcommand{\btUIj}{\overline{\mathrm{UI}}_{j}}
\newcommand{\bUIi}{\overline{\mathrm{UI}}_{i}}

\newcommand{\diag}{\mathrm{diag}}
\newcommand{\trace}{\mathrm{tr}}

\DeclareMathOperator{\rank}{\mathrm{rank}}

\newcommand{\normLt}[1]{{\|{#1}\|}}

\title{\fontsize{0.83cm}{1cm}\selectfont RIS-Assisted XL-MIMO for Near-Field and Far-Field Communications}

\author{Xiaomin Cao, Mohammadali Mohammadi,~\IEEEmembership{Senior Member,~IEEE,}
Hien Quoc Ngo,~\IEEEmembership{Fellow,~IEEE,} \\
Hyundong Shin,~\IEEEmembership{Fellow,~IEEE,} and  Michail Matthaiou,~\IEEEmembership{Fellow,~IEEE}
\thanks{
This work was supported by the U.K. Engineering and Physical Sciences Research
Council (EPSRC) grant (EP/X04047X/2) for TITAN Telecoms Hub. The work of X. Cao was supported in part by the China Scholarship Council. The work of  H. Q. Ngo was supported by the U.K. Research and Innovation Future Leaders Fellowships under Grant MR/X010635/1, and a research grant from the Department for the Economy Northern Ireland under the US-Ireland R\&D Partnership Programme. The work of M. Matthaiou was supported by the European Research Council (ERC) under the European Union’s Horizon 2020 research and innovation programme (grant agreement No. 101001331). An earlier version of this paper was presented in part at the 2024 IEEE WCNC [DOI:  10.1109/WCNC57260.2024.10570871].
The associate editor coordinating the review of this article and approving it for publication was Qingqing Wu. (\textit{Corresponding authors: Michail Matthaiou; Hyundong Shin.})
}
\thanks{X. Cao, M. Mohammadi, H. Q. Ngo, and M. Matthaiou are with the Centre for Wireless Innovation (CWI), Queen's University Belfast, BT3 9DT Belfast, U.K. (email: \{xcao13, m.mohammadi, hien.ngo, m.matthaiou\} @qub.ac.uk.)}

\thanks{  H. Q. Ngo, M. Matthaiou, and H. Shin are affiliated with the Department of Electronic Engineering, Kyung Hee University, Yongin-si, Gyeonggi-do 17104, Republic of Korea (e-mail: hshin@khu.ac.kr).}

}

\allowdisplaybreaks

\usepackage[utf8]{inputenc}

\begin{document}
\bstctlcite{IEEEexample:BSTcontrol}
\maketitle
\begin{abstract}
    We consider a reconfigurable intelligent surface (RIS)-assisted extremely large-scale multiple-input multiple-output (XL-MIMO) downlink system, where an XL-MIMO array serves two groups of single-antennas users, namely near-field users (NFUEs) and far-field users (FFUEs). FFUEs are subject to blockage, and their communication is facilitated through the RIS. We consider three precoding schemes at the XL-MIMO array, namely central zero-forcing (CZF), local zero-forcing (LZF) and maximum ratio transmission (MRT).  Closed-form expressions for the spectral efficiency (SE) of all users are derived for MRT precoding, while statistical-form expressions are obtained for CZF and LZF processing. A heuristic visibility region (VR) selection algorithm is also introduced to help reduce the computational complexity of the precoding scheme. Furthermore, we devise a two-stage phase shifts design and power control algorithm to maximize the sum of weighted minimum SE of two groups of users with CZF, LZF and MRT precoding schemes. The simulation results indicate that, when equal priority is given to NFUEs and FFUEs, the proposed design improves the sum of the weighted minimum SE by $31.9\%$, $37.8\%$, and $119.2\%$ with CZF, LZF, and MRT, respectively, compared to the case with equal power allocation and random phase shifts design. CZF achieves the best performance, while LZF offers comparable results with lower complexity. When prioritizing NFUEs or FFUEs, LZF achieves strong performance for the prioritized group, whereas CZF ensures balanced performance between NFUEs and FFUEs.   
\end{abstract}


\begin{IEEEkeywords}
    Extremely large-scale multiple-input multiple-output (XL-MIMO), near-field communications, power control, reconfigurable intelligent surface (RIS), visibility region (VR).  
\end{IEEEkeywords}

\section{Introduction}
To meet the growing demands for high data rates and massive connectivity beyond fifth generation (5G) and sixth generation (6G) networks, extremely large-scale multiple input multiple output (XL-MIMO), as one of the key enablers, has received significant attention~\cite{Wang:Surveys:2024}. By increasing the number of antennas by at least an order of magnitude — ranging from several hundred to thousands—XL-MIMO represents a major advance over massive MIMO, enabling high spectral efficiency (SE), energy efficiency (EE) and reliable massive connectivity \cite{matthaiou2021road, Wang:Surveys:2024}. Moreover, significant strides have been made in XL-MIMO hardware implementations; for instance, according to \cite{Deng:JSAC:2023}, a reconfigurable holographic surface-aided platform is capable of supporting the real-time transmission of high-definition video. In \cite{Li:TWC:2024}, the authors concluded that modular XL-MIMO can provide better spatial resolution by availing of its larger array aperture, paving the way for practical large-scale deployments. Despite these impressive advances in XL-MIMO, the signal received by some users may be blocked by the presence of large objects remains unsolved.

We address this problem by introducing a reconfigurable intelligent surface (RIS) which is a planar array equipped with a large number of low-cost passive elements, and each element is able to induce a certain phase shift \cite{Wu:2021:TCM}. Recent breakthroughs have been made in RIS hardware advancements. For example, in \cite{Zijian:TCOM:2023}, the authors proposed the concept of active RISs and focused on signal verification through experimental measurements based on a fabricated active RIS element, which revealed that active RISs can achieve a significant sum-rate gain compared to passive RISs. The paper \cite{Mu:TWC:2022} studied a simultaneously transmitting and reflecting RIS-aided downlink communication system. In \cite{Nerini:TWC:2024}, the authors explored novel modeling, architecture designs, and optimization schemes for beyond diagonal RIS based on graph theory. Leveraging the capability of RISs to customize the wireless propagation environment, a promising approach is to integrate RISs with XL-MIMO systems. RIS-assisted XL-MIMO systems are expected to harness the combined benefits of both technologies, especially when direct base station (BS)-user links are obstructed \cite{Tor:J_STSP:2024,Lee:TCOM:2024,Zhang:letters:2024}. Specifically, the authors of \cite{Tor:J_STSP:2024}
considered the presence of a RIS to address the challenge of non-line-of-sight (NLoS) scenarios in XL-MIMO systems, while the authors of \cite{Lee:TCOM:2024} studied the channel estimation problem for extremely large-scale RIS (XL-RIS) assisted multi-user XL-MIMO systems with hybrid beamforming structures.

Although RIS-assisted XL-MIMO systems provide a promising solution to solve the blockage problem, the increased dimensionality of XL-MIMO will notably alter the electromagnetic (EM) characteristics. Specifically, near-field (NF) propagation dominates as the array size becomes comparable to or larger than the user distance. The distinction between NF and far-field (FF) users (FFUEs)  affects the choice of the channel model, which is determined based on the Rayleigh distance~\cite{Cui:CMG:2023}. Users closer than the Rayleigh distance are modeled using NF (spherical wave) channels, whereas others are modeled using FF (planar wave) models. The precoding design in our work is based on the channel matrix and does not explicitly depend on whether a user is in the NF or FF region. Once the channel is defined, the precoding scheme is applied accordingly. 

Motivated by the advantages of XL-MIMO, studies on NF behaviors \cite{Lu:TWC:2022,Björnson:JOP:2020,Dardari:2020:JSAC} highlight the limitations of conventional models, such as the uniform plane wave (UPW) model, advocating for a uniform spherical wave (USW) model \cite{Lu:TWC:2022}. Three fundamental properties should be taken into account when the receiver is in the NF region of the array, which are the distances to the elements, the effective areas of the array, and the losses due to polarization mismatch \cite{Björnson:JOP:2020,Dardari:2020:JSAC}. However, the extra-large aperture causes each user/scatterer to be observable within a limited visibility region (VR) \cite{Carvalho:IWC:2020}. Consequently, channel statistics vary along the array, leading to the non-stationarity of XL-MIMO. Nevertheless, the above works have not accounted for this non-stationarity. The main causes of spatial non-stationarity in NF channels are as follows:
(i) \textit{The inherent nonlinearity} in the signal space of NF line-of-sight (LoS) paths \cite{Zhi:JSAC:2024, Chen:TWC:2024};
(ii) \textit{Potential partial blockage} of NF LoS components \cite{Chen:PIMRC:2024};
(iii) \textit{Spatially varying scattering environments} across different subarrays \cite{Xu:TSP:2024}. In this study, our consideration of spatial non-stationarity is primarily motivated by the nonlinear phase and amplitude variations across the array, as well as the variation in wave incident angles \cite{Lu:TWC:2022}. To effectively capture and mitigate these effects, we propose a practical and lightweight VR selection algorithm that is well suited for addressing the challenges posed by spatial non-stationarity in NF XL-MIMO systems. 

The impact of VR introduced by XL-MIMO arrays was considered in \cite{Xiao:WCNC:2024,Amiri:Globecom:2018,Zhi:JSAC:2024, Xu:TSP:2024, Mari:TVT:2020,Ribe:EUSIPCO:2021,Yang:TVT:2020,Lu:TWC:2024,Li:TWC:2024}. In~\cite{Amiri:Globecom:2018}, drawing inspiration from coded random access, the authors proposed a decentralized receiver with very low complexity, where processing is executed locally per subarray. The authors introduced a VR detection algorithm and utilized the acquired VR information to design a low-complexity symbol detection scheme in an uplink multi-user scenario in \cite{Zhi:JSAC:2024}. The authors of \cite{Mari:TVT:2020} proposed four antenna selection approaches in XL-MIMO systems to maximize the total EE. Two novel precoding schemes were proposed to reduce the complexity, namely mean angle-based zero forcing and tensor zero forcing in \cite{Ribe:EUSIPCO:2021}. In addition, in \cite{Yang:TVT:2020}, the authors developed a greedy joint user and subarray scheduling algorithm based on statistical channel state information (CSI) to maximize the achievable SE in the system. In addition, the channel characteristics vary along with the array, leading to non-stationarity. Despite the non-stationarity problem that can be addressed by adopting VRs, the fair resource allocation in the RIS-assisted XL-MIMO system is still unsolved. To the authors' best knowledge, no existing system tackles the non-stationarity and power allocation problem simultaneously.

To fill this gap, we perform a comprehensive analysis of the RIS-aided XL-MIMO system, considering VR selection, phase shift design, and power control with three precoding designs. Although power allocation algorithms have been extensively studied in massive MIMO \cite{ngo16,Mohammadali:TCOM:2024}, little research has focused on XL-MIMO systems. In \cite{Souza:TVT:2021}, joint antenna selection and power allocation were optimized to maximize the downlink SE under different loading conditions. In \cite{Liu:TWC:2024}, a multi-agent reinforcement learning-based power control algorithm was proposed for cell-free XL-MIMO. 

Due to the large aperture size of XL-MIMO arrays and the use of high-frequency bands in future wireless networks, users are likely to reside in either the NF or FF regions. This calls for the adoption of appropriate channel models tailored to different user locations. Consequently, XL-MIMO systems differ significantly from conventional massive MIMO in terms of both system modeling and channel characteristics. In particular, the presence of NF users and the spatially non-stationary nature of XL-MIMO channels can have a pronounced impact on system design. Under such conditions, the extent to which the performance can be improved through optimized system design, such as RIS phase shifts optimization and adaptive power control, remains uncertain. To address this challenge, we adopt a comprehensive channel model that captures both NF and FF propagation effects. This gap motivates our study, with key contributions summarized as follows:
\begin{itemize}
    \item We consider three linear precoding schemes—maximum ratio transmission (MRT), local zero-forcing (LZF), and central zero-forcing (CZF)—for a RIS-assisted XL-MIMO downlink communication system with NF users (NFUEs) and FFUEs. Rigorous closed-form SE lower bounds are derived for the MRT precoding scheme under VR selection for both NFUEs and FFUEs. Additionally, statistical lower bounds on the SE are obtained for the LZF and CZF schemes with selected VRs. In addition, we also consider the interference within two groups, which is practical and innovative.
    
    \item We formulate a joint optimization problem that integrates VR selection, phase shifts design, and power control for MRT, LZF, and CZF precoding schemes to provide a uniformly good service within each group of users. Our objective is to maximize the sum of the weighted minimum SE of both NFUEs and FFUEs, while ensuring that all users meet their quality-of-service (QoS) requirements.
    \item The highly coupled non-convex optimization problem is transformed into a tractable form with decoupled optimization variables. To achieve this, we propose a heuristic algorithm for VR selection, which selects effective VRs for each user. Then, we devise a two-stage low-complexity algorithm by decoupling the non-convex problem into two subproblems. In the first stage, we propose an algorithm for optimized phase shifts determination, leveraging the penalty method and successive convex approximation (SCA). In the second stage, we introduce a SCA-based power control algorithm.
  
    \item Our numerical results reveal that $1)$ The optimized phase shifts and power allocation coefficients lead to a substantial improvement in the sum of the weighted minimum SE of NFUEs and FFUEs; $2)$ CZF is generally preferable while LZF achieves nearly equivalent performance with low complexity when giving equal priority to the NFUEs and FFUEs; $3)$ When different priorities are assigned to NFUEs and FFUEs, LZF delivers satisfactory performance for the prioritized users, while CZF achieves a balanced performance between NFUEs and FFUEs.
\end{itemize}

\begin{table*}
	\centering
	\caption{\label{tabel:Survey} Contrasting our contributions to the XL-MIMO literature}
	\vspace{-0.6em}
	\small
\begin{tabular}{|p{3cm}|p{1.5cm}|p{0.63cm}|p{0.63cm}|p{0.63cm}|p{0.63cm}|p{0.63cm}|p{0.63cm}
|p{0.63cm}|p{0.63cm}|p{0.63cm}|p{0.63cm}|p{0.63cm}|p{0.63cm}|}
	\hline
        \centering\textbf{Contributions} 
        &\centering \textbf{This paper}
        &\centering\cite{Xiao:WCNC:2024}
        &\centering\cite{Tor:J_STSP:2024}
        &\centering\cite{Lee:TCOM:2024}
        &\centering\cite{Zhang:letters:2024}
        &\centering\cite{Zhi:JSAC:2024}
        &\centering\cite{Xu:TSP:2024}
        &\centering\cite{Amiri:Globecom:2018}
        &\centering\cite{Mari:TVT:2020}
        &\centering\cite{Ribe:EUSIPCO:2021}
        &\centering\cite{Yang:TVT:2020}
        &\centering\cite{Souza:TVT:2021}
        &\centering\cite{Liu:TWC:2024}
        \cr
        \hline

        VR design     
        &\centering\checkmark
        &\centering\checkmark
        & \centering\tikzxmark
        & \centering\tikzxmark
        & \centering\tikzxmark        
        &\centering\checkmark
        &\centering\checkmark     
        &\centering\checkmark
        &\centering\checkmark 
        &\centering\checkmark
        &\centering\checkmark
        &\centering\checkmark
        & \centering\tikzxmark
        \cr
        \hline

      Power allocation         
        &\centering\checkmark
        &\centering\tikzxmark
        & \centering\tikzxmark
        &\centering\tikzxmark
        & \centering\tikzxmark        
        &\centering\tikzxmark
        &\centering\tikzxmark        
        &\centering\tikzxmark
        &\centering\tikzxmark 
        & \centering\tikzxmark
        &\centering\tikzxmark
        &\centering\checkmark
        &\centering\checkmark
        \cr

        \hline
        Mixed NF and FF users         
        &\centering\checkmark 
        &\centering\checkmark
        & \centering\tikzxmark
        &  \centering\tikzxmark
        & \centering\tikzxmark        
        &\centering\tikzxmark
        &\centering\tikzxmark          
        &\centering\tikzxmark 
        &\centering\tikzxmark 
        & \centering\tikzxmark
        &\centering\tikzxmark
        & \centering\tikzxmark 
        &  \centering\tikzxmark
        \cr
        \hline 

        RIS     
        &\centering\checkmark 
        &\centering\checkmark
        &\centering\checkmark
        & \centering\checkmark
        &\centering\checkmark        
        &\centering\tikzxmark
        &\centering\tikzxmark         
        & \centering\tikzxmark
        & \centering\tikzxmark
        & \centering\tikzxmark
        &\centering\tikzxmark
        & \centering\tikzxmark
        & \centering\tikzxmark
        \cr 
        \hline

\end{tabular}
\vspace{-1em}
\label{Contribution}
\end{table*}

A comparison of our contributions with the state of the art in the field of XL-MIMO is tabulated in Table~\ref{tabel:Survey}. 

\textit{Notation:} We use bold lowercase letters (uppercase) to denote vectors (matrices). The superscripts $(\cdot)^T$ and $(\cdot)^H$ stand for the transpose and Hermitian, respectively; $\lfloor x\rfloor$ denotes the largest integer less than or equal to $x$; $[\qA]_{(:,j)}$ denotes the $j$th column of matrix $\qA$; $\mathbf{0}_M$ and $\mathbf{0}_{M\times M}$ denote the all-zero vector of size $M$ and all-zero matrix of size $M\times M$, respectively; $\diag(\qa)$ denotes a diagonal matrix with each diagonal element being the corresponding element in $\qa$ and $\diag(\qA_1,\ldots,\qA_n)$ denotes a block-diagonal matrix with the square matrices $\qA_1,\ldots,\qA_n$ on the diagonal. The trace and Kronecker product operations are denoted by $\tr(\cdot)$ and $\otimes$, respectively; $\| \qA\|_{\ast} = \sum_i\sigma_i(\qA)$ and $\| \qA\|_2 = \sigma_1(\qA)$ are the nuclear norm and spectral norm, respectively, while $\sigma_i(\qA)$ denotes the $i$th largest singular value of matrix $\qA$; $\qu( \qA)$ is the eigenvector of $\qA$ corresponding to the largest eigenvalue. Moreover, $\mathcal{CN}(0,\sigma^2)$ denotes a complex circularly symmetric Gaussian random variable (RV) with variance $\sigma^2$; $\mathbb{E}\{\cdot\}$ denotes the statistical expectation; $\delta_{ss'}$ denotes the Kronecker delta, with $\delta_{ss'}$ 1 if $s=s'$, and 0 otherwise. We use $|\mathcal{A}|$ to denote the cardinality of the set $\mathcal{A}$.

\section{System Model} \label{Sec:model}
\begin{figure}[t]
    \centering
    \includegraphics[width=70mm]{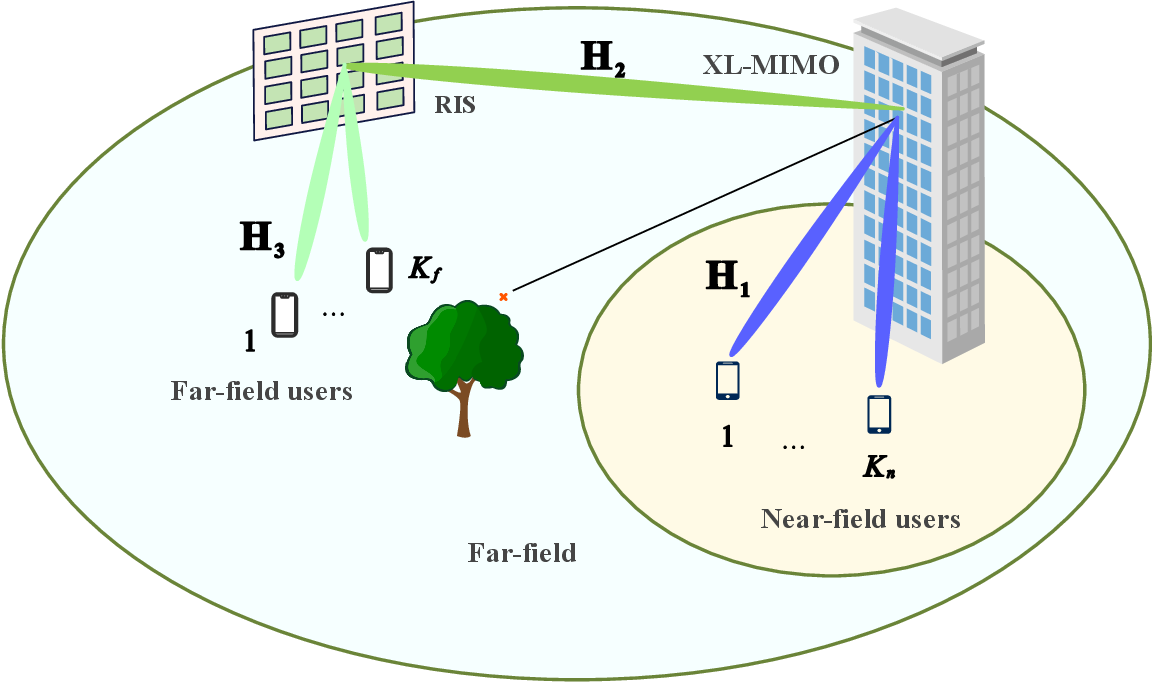}
    \caption{Illustration of RIS-assisted XL-MIMO system.}
    \vspace{-1em}
            \label{fig:systemModel}
\end{figure}

We consider a downlink communication system in which an XL-MIMO array with $M$ antennas serves $K \!=\! K_n \!+\! K_f$ users, where $K_n$ users are in the NF region, referred to as NFUEs, and $K_f$ users are in the FF region, referred to as FFUEs. For simplicity of notation, we define the sets $\Kn\triangleq\{1,\ldots,K_n\}$, $\Kf\triangleq\{1,\ldots,K_f\}$ to represent the sets of NFUEs, FFUEs respectively. Transmission to FFUEs is blocked by obstacles, and then a passive RIS is deployed in the system, whereas paths involving multiple interactions with the RIS—such as signals reflecting off the RIS, then another surface, and back to the RIS—are ignored. The number of reflecting units at the RIS is $N\!=\!N_1 N_2$, where $N_1$ and $N_2$ are the numbers of elements in rows and columns, respectively. For the XL-MIMO array, we consider a uniform planar array (UPA) with $M \!=\! M_x M_y$ antennas, while $M_x$ and $M_y$ denote the number of antennas along the $x$- and $y$-axis, respectively. Furthermore, to address NF spatial non-stationary in XL-MIMO systems, the array is divided into $S$ subarrays, where each subarray contains $M^{\ast}$ antenna elements, that is, $ M^{\ast}=M/S$. We set $M^{\ast} \!=\!M_x^{\ast}M_y^{\ast}$,  where $M_x^{\ast}$ and $M_y^{\ast}$ are the number of antennas along the $x$ and $y$ directions, respectively. We assume all users are equipped with a single antenna and, to quantify the theoretical upper bound on system performance, we consider perfect knowledge of CSI for all involved channels~\cite{Zhi:JSAC:2024,Lu:TWC:2022,Björnson:JOP:2020}.

\subsection {Near-Field Channel Model of XL-MIMO}
We consider a discrete aperture array model with hypothetical isotropic antenna elements. The effective area of each element is given by $\sqrt{A}\times\sqrt{A}$, where $A = e_a\frac{\lambda^2}{4\pi}$ \cite{Lu:TWC:2022}, $\lambda$ is the wavelength and $0 < e_a \leq 1$ denotes the aperture efficiency of each antenna element. Without loss of generality and for simplicity, we assume $e_a=1$. The inter-element distance is $d = \frac{\lambda}{2}$, while the occupation ratio of the array is defined as $\beta=\frac{A}{d^2} \leq 1$, representing the fraction of the total UPA area occupied by the antenna elements. For convenience, we assume that $M_x$ and $M_y$ are odd numbers, and the location of central point of the $(m_x,m_y)$th element is $\qs_{m_x, m_y} = [m_xd,m_yd,0]^T$, where $m_x = 0,\pm1,\ldots,\pm\frac{M_x-1}{2}$, $m_y = 0,\pm1,\ldots,\pm\frac{M_y-1}{2}$, respectively. We assume that $\qu_k=[u_{k,x},u_{k,y},u_{k,z}]^T$ is the location of the $k$th NFUE. Then, the distance between the $k$th NFUE and the center of the $(m_x,m_y)$th element is 
\begin{align}
    r_{m_x,m_y,k} 
    &= \sqrt{(u_{k,x}-m_x d)^2+(u_{k,y}-m_y d)^2+u_{k,z}^2}.
\end{align}
The channel from the $(m_x, m_y)$th element to the $k$th NFUE, which considers free-space path loss, effective aperture, and polarization mismatch, can be expressed as \cite[Eq.~(11)]{Liu:JOP:2023}, \cite{Zhi:JSAC:2024}
\begin{align}    \label{eq:NF channel model}
    &\bar{g}_{m_x,m_y,k} = \sqrt{\eta_{m_x,m_y,k}}e^{-j\frac{2\pi}{\lambda}r_{m_x,m_y,k}} \nonumber \\
    &\hspace{2em}+ \sum\nolimits_{l=1}^{L_k} \Tilde{\beta}_{k,l}\sqrt{\eta_{m_x,m_y,l}} e^{-j\frac{2\pi}{\lambda}d_{m_x,m_y,l}} ,~\forall m_x,m_y,
\end{align}
where $\eta_{m_x,m_y,k}$ is the power gain of the channel, which can be written as \cite[Eq.~(9)]{Zhi:JSAC:2024} 
\begin{equation}
    \eta_{m_x,m_y,k} \!\!{\ \approx \ } \!\!\frac{A}{4\pi}\frac{u_{k,z}\left((u_{k,x} \!-\! m_xd)^2 \!+\! u_{k,z}^2\right)}{\big[\!(u_{k,x} - m_xd)^2 \!+\! (u_{k,y}\!-\! m_yd)^2 \!+\! u_{k,z}^2\!\big]^{\!\frac{5}{2}}},
\end{equation}
and $\eta_{m_x,m_y,l}$ can be obtained as $\eta_{m_x,m_y,k}$ in a similar manner. We assume that the scatterer $l$ is located at $\qp_{k,l} = [p_{k,l,x},p_{k,l,y},p_{k,l,z}]^T$, and $d_{m,l} = \| \qp_{k,l} - \qs_{m_x,m_y}\|$. Moreover, in~\eqref{eq:NF channel model}, $L_k$ is the number of dominant single bounce scatterers visible from NFUE $k$, $\Tilde{\beta}_{k,l}$ includes the impact of the random reflection coefficients of the $l$th scatterer and channel gain from the scatterer $l$ to NFUE $k$. It is assumed that the random phase of $\Tilde{\beta}_{k,l}$ is independent and identically distributed and uniformly distributed in $(-\pi, \pi]$ \cite{{Liu:JOP:2023}}.
    
Note that we consider that the channels between XL-MIMO and NFUEs are LoS channels in this paper, since they have more impact on the received signal of NFUEs.

\begin{table*}[!t]
\small
  \caption{List of main symbols and the physical meanings}
  \label{tab:symbols}
  \centering
  \renewcommand{\arraystretch}{1.15}
  \begin{tabular}{|c|l||c|l|}
    \hline
    $M$                         & Number of antennas at the XL-MIMO
    & $N$                     & Number of elements at
     \\ \hline
    $K_n$         & Number of NFUEs
    & 
    $K_f$     & Number of FFUEs
    \\ \hline 
    $P_n$      & Maximum transmit power for NFUEs
    & $P_f$     & Maximum transmit power for FFUEs      \\ \hline
    
    $\bar{\qH}_1$             & Channel between XL-MIMO and all NFUEs 
    & $\boldsymbol{\Theta}$             & Phase shifts matrix of RIS \\ \hline
    $\bar{\qg}_{k}$         & Channel between XL-MIMO and NFUE $k$
    & $\bar{\qg}_{sk}$         & Channel between $s$th subarray and NFUE $k$ \\ \hline
    $\Tilde{\qG}$         & The channel between RIS ans all $K_f$ FFUEs
    & $\Tilde{\qh}_k$         & Channel between RIS and FFUE $k$\\   \hline
    $\qH_2$         & Channel from XL-MIMO to RIS
    & $\qH_{2,s}$   & Channel from $s$th subarray of XL-MIMO to RIS      \\ \hline
    $\bar{\qw}_{sk}$             & Precoding vector at subarray $s$ to NFUE $k$
    &$\Tilde{\qw}_{sj}$             & Precoding vector at subarray $s$ to FFUE $j$\\ \hline
       
    $\bar{\qD}_{sk}$         & Indicator of subarray $s$'s active antennas for NFUE $k$
    &$\Tilde{\qD}_{sk}$         & Indicator of subarray $s$'s active antennas for FFUE $k$ \\ \hline 
    $\bar{\boldsymbol{\eta}}_{k}$         & Power coefficients allocated to NFUE $k$
    & $\Tilde{\boldsymbol{\eta}}_{j}$         & Power coefficients allocated to FFUE $j$\\ \hline
    $\bar{\eta}_{sk}$         & Power coefficients at subarray $s$  allocated to NFUE $k$
    & $\Tilde{\eta}_{sj}$         & Power coefficients at subarray $s$  for FFUE $j$\\ \hline    
  \end{tabular}
  \vspace{-1em}
\end{table*}

Let $\bar{\qg}_k = [\bar{\qg}_{1k}, \ldots, \bar{\qg}_{Sk}]^T \in \mathbb{C}^{M \times 1}$ denote the overall channel vector between the XL-MIMO array and NFUE $k$, where $\bar{\qg}_{sk} = [\bar{g}_{(s-1)M^{\ast}+1,k}, \ldots, \bar{g}_{s M^{\ast},k}]^T \in \mathbb{C}^{M^{\ast} \times 1}$ represents the channel between the $s$th subarray of the XL-MIMO system and NFUE $k$, while $\bar{g}_{m,k}$, can be obtained from \eqref{eq:NF channel model} by calculating the coordinates $(m_x,m_y)$ of the $m$th element via the following expressions:
\begin{subequations} \label{eq:coordinates}
\begin{align}
    &m_x = -0.5(M_x-1) + \bmod (m-1, M_x), \\
    &m_y = 0.5(M_y-1) - \Big\lfloor \frac{m-1}{M_x} \Big\rfloor.
\end{align}    
\end{subequations}
Moreover, the channel between the XL-MIMO array and all NFUEs is expressed as $\bar{\qH}_1^H \!=\! [\bar{\qg}_1,\ldots,\bar{\qg}_{K_n}]^H \in \mathbb{C}^{K_n \times M}$.  

\vspace{-1em}
\subsection {Far-Field Channel Model of XL-MIMO}
The channels from the XL-MIMO to the RIS\footnote{We consider a RIS that is not extremely large. Therefore, it can be modeled as a single transmitter, and spatial non-stationarities do not arise, unlike in the case of XL-RISs, where such issues must be considered \cite{Lee:TCOM:2024}.} and from the RIS to FFUE $k$ are denoted by $\qH_2 \in \mathbb{C}^{N \times M}$ and $\qh_k^H \in \mathbb{C}^{1 \times N}$, respectively. Furthermore, the cascaded channel from the XL-MIMO to FFUE $k$ through the RIS is expressed as $\Tilde{\qg}_k^H = \qh_k^H\bTeta\qH_2 \in \mathbb{C}^{1\times M}$, where $\bTeta = \diag(e^{j\theta_1},\ldots,e^{j\theta_N})$, with $\theta_n \in [0,2\pi)$ representing the phase shift of $n$th reflecting element, and we set $\boldsymbol{\theta} = [\theta_1,\ldots,\theta_N]$. The cascaded channels of all $K_f$ FFUEs are given by $\Tilde{\qG}^H = [(\Tilde{\qg}_1^H)^T, \ldots, (\Tilde{\qg}_{K_f}^H)^T]^H\in\mathbb{C}^{K_f\times M}$. Since the RIS is often installed on the facade of a tall building, both LoS and NLoS transmission paths would exist in the XL-MIMO-to-RIS link. Thus, we employ the Ricean fading model for the link from the XL-MIMO to the RIS as 
\begin{align}
    \qH_2 = \alpha_2\bar{\qH}_2 + \beta_2\Tilde{\qH}_2,
\end{align}
where $\alpha_2 = \sqrt{\frac{ \iota\zeta}{\iota + 1}}$, $\beta_2 = \sqrt{\frac{\zeta }{\iota + 1}}$, while $\zeta$ is the path-loss coefficient corresponding to the link from XL-MIMO array to the RIS; $\iota$ denotes the corresponding Ricean factor; 
$\bar{\qH}_2$ denotes the LoS channel matrix which can be expressed as
\begin{align}
    \Bar{\qH}_2 &= \qb_N(\varphi^a,\varphi^e)\qa_M^H(\phi^a,\phi^e),
\end{align}  
where $\varphi^a$ and $\varphi^e$ represent the azimuth and elevation angles of arrival (AoAs) at the RIS from the XL-MIMO, while $\phi^{a}$ and $\phi^{e}$ are the corresponding azimuth and elevation angles of departure (AoDs) from the XL-MIMO to the RIS. Furthermore, $\qa_M(\phi^a,\phi^e)$ and $\qb_N(\varphi^a,\varphi^e)$ are the array response vectors from the XL-MIMO to the RIS, where
\begin{subequations}
\begin{align}
    &\qa_M(\phi^a,\phi^e) = [e^{j2\pi m_x \frac{d \Cos(\phi^e) \Sin(\phi^a)}{\lambda}}]_{m_x\in\mathcal{L}(M_x)}^T \nonumber \\ 
    &\hspace{2em} \otimes[e^{j2\pi m_y \frac{d \Sin(\phi^e)}{\lambda}}]_{m_y\in\mathcal{L}(M_y)}^T, \\
   &\qb_N(\varphi^a,\varphi^e) = [e^{j 2\pi n_1 \frac{d_R \Cos(\varphi^e)\Sin(\varphi^a)}{\lambda}}]_{n_1\in\mathcal{L}(N_1)}^T
    \nonumber \\ 
   &\hspace{2em} \otimes[e^{j2\pi n_2\frac{d_R \Sin(\varphi^e)}{\lambda}}]_{n_2\in\mathcal{L}(N_2)}^T,\label{eq:bN} 
\end{align}    
\end{subequations}
where $d_R = \frac{\lambda}{2}$ is the RIS element spacing  and $\mathcal{L}(n) = \{0,1,\ldots,n-1\}$.
Moreover, $\Tilde{\qH}_2$ is the NLoS channel matrix whose components are independent and identically distributed (i.i.d.) complex Gaussian RVs with zero mean and unit variance. To facilitate notation, let $\qH_{2,s}\in \mathbb{C}^{N \times M^{\ast}}$ denote the channel between the $s$th subarray of XL-MIMO and the RIS, which is a submatrix of $\qH_2$, given by $\qH_{2,s}=\qH_2(:, r_s)$, where $r_s$ is the antenna index set corresponding to $s$th subarray, defined as $r_s \triangleq \{(s-1) M^{\ast}+1, sM^{\ast}\}$.

We assume that the links between the RIS and FFUEs are LoS dominant.\footnote{In practice, RISs are intentionally placed on building facades/walls to ensure favorable LoS propagation between the RIS and blocked FFUEs by maximizing reflection gain. We note that the gap between the data rate achieved by a purely LoS channel and that achieved by a LoS-dominant channel is negligible, as shown in \cite{Zhi:2022:JSAC, Mohammadi:TC:2024}.} Thus, the channel between the RIS and FFUE $k$ can be written as \cite{Mohammadi:TC:2024}
\begin{align}
    \qh_k^H = \sqrt{\varsigma_k} \qb_N^H (\vartheta_{k}^a,\vartheta_{k}^e),~\forall k\in\Kf,
\end{align} 
where $\varsigma_k$ is the path-loss coefficient; $\vartheta_{k}^a$ and $\vartheta_{k}^e$ represent the azimuth and elevation AoDs reflected by the RIS to the $k$th FFUE, while the array response vector $\qb_N(\vartheta_{k}^a,\vartheta_{k}^e)$ is given by~\eqref{eq:bN}, where $\varphi^a$ and $\varphi^e$ are replaced with $\vartheta_{k}^a$ and $\vartheta_{k}^e$, respectively.

It is worth mentioning that a LoS or near-LoS connection between the RIS and the XL-MIMO BS is essential to ensure a strong and stable BS–RIS link. This condition allows the maximum number of RIS elements to effectively receive and reflect the incident signals with high efficiency. Additionally, FFUEs should be positioned within the RIS reflection region to avail of the beamforming gain. To enhance system performance, orientation strategies should aim to introduce spatial separation—or angular diversity—along the RIS–user paths. This reduces channel correlation and increases degrees of freedom, enabling more effective multi-user transmission.
\vspace{-1em}
\subsection{Received Signal}
The transmit signal from the $s$th subarray is given by
\begin{align}~\label{eq:xs}
    \qx_s \!\!=\!\! \!\!\sum_{k\in\Kn}\!\!\!\!\sqrt{P_n\betask}\bar{\qD}_{sk}\bar{\qw}_{sk} \Bar{s}_k
    \!+\!\!\sum_{j\in\Kf}\!\!\!\!\sqrt{P_f\Tetasj}\Tilde{\qD}_{sj}\Tilde{\qw}_{sj} \Tilde{s}_j,
\end{align}
where $P_n$ and $P_f$ denote the transmit power for NFUEs and FFUEs; $\bar{\eta}_{sk}$, $\Tetasj$ are the power control coefficients at the $s$th subarray  allocated to the $k$th NFUE and the $j$th FFUE, respectively; $\bar{\qw}_{sk}\in\mathbb{C}^{M^{\ast}\times 1}$ and $\Tilde{\qw}_{sj}\in\mathbb{C}^{M^{\ast}\times 1}$ are the precoding vectors at the $s$th subarray associated with  NFUE $k$ and  FFUE $j$, respectively; $\bar{s}_k$ and $\Tilde{s}_j$ are the  uncorrelated symbols intended for  NFUE $k$ and FFUE $j$, with $\Ex\{|\bar{s}_k|^2\}\!=\!\Ex\{|\Tilde{s}_j|^2\}\!=\!1$. Moreover, for the sake of generalization, we define a diagonal matrix $\qD_{sk}$ with binary entries on its main diagonal to indicate the active antenna elements in subarray $s$ assigned to user $k$. Since we consider subarray-level activation, $\bar{\qD}_{sk} = \qI_{M^\ast} (\Tilde{\qD}_{sj} = \qI_{M^\ast})$ implies that NFUE $k$ (FFUE $j$) is served by  subarray $s$, otherwise, $\bar{\qD}_{sk}= \mathbf{0}_{M^\ast} (\Tilde{\qD}_{sj}= \mathbf{0}_{M^\ast})$.\footnote{To simplify the subsequent analysis, we do not include the signal propagation through sidelobes.}

The transmitted power satisfies the power constraint $\sum\nolimits_{s\in\Set} \Ex\{ \normLt{\qx_s}^2\} \leq P$, where $P$ is the maximum total transmit power at the XL-MIMO array. Using~\eqref{eq:xs}, we can express the power constraint as
\begin{align} \label{eq:simplified power control}
    \sum\nolimits_{s\in\Set}\!\!\Big( P_n {\sum\nolimits_{k\in\Kn}\!\! \bar{\eta}_{sk} \bar{c}_{sk}} + P_f {\sum\nolimits_{j\in\Kf}\!\!\Tetasj \Tilde{c}_{sj}}\Big) \leq P,
\end{align}
where $\bar{c}_{sk} = \tr(\bar{\qD}_{sk}\bar{\qw}_{sk}(\bar{\qD}_{sk}\bar{\qw}_{sk})^H)$, $\Tilde{c}_{sj} = \Ex \{ \tr(\Tilde{\qD}_{sj}\Tilde{\qw}_{sj}(\Tilde{\qD}_{sj}\Tilde{\qw}_{sj})^H)\}$.

The signal received at NFUE $k$ can be expressed as:
\begin{align} \label{eq:receive signa:Near}
    \bar{y}_{k} 
    &=\sum\nolimits_{s\in\Set}\!   
    \Big(\!\!\sqrt{P_n}\! \sum\nolimits_{k\in\Kn}\sqrt{\betask}\bar{\qg}_{sk}^H\bar{\qD}_{sk}\bar{\qw}_{sk} \bar{s}_k \nonumber \\   
    &\hspace{2em}+\!\!\sqrt{P_f}\sum\nolimits_{j\in\Kf}\!\!\sqrt{\Tetasj}\bar{\qg}_{sk}^H\Tilde{\qD}_{sj}\Tilde{\qw}_{sj} \Tilde{s}_j\Big)
    \!+ \!n_k, 
\end{align}
where $\Set$ is the set of all subarrays, $n_k\sim\mathcal{CN}(0,1)$ is the additive white Gaussian noise (AWGN) at the $k$th NFUE. 

For FFUE $k$, the received signal can be expressed as 
\begin{align} \label{eq:receive signa:Far}
    \Tilde{y}_{k} 
    &=\sum\nolimits_{s\in\Set}\!   
    \Big(\!\sqrt{P_f}\sum\nolimits_{k\in\Kf}\!\!\sqrt{\Tilde{\eta}_{sk}}\Tilde{\qg}_{sk}^H\Tilde{\qD}_{sk}\Tilde{\qw}_{sk}\Tilde{s}_k \nonumber \\   
    &\hspace{2em}+\!\!\sqrt{P_n}\sum\nolimits_{i\in\Kn}\!\!\sqrt{\betasi}\Tilde{\qg}_{sk}^H\bar
    {\qD}_{si}\bar{\qw}_{si} \bar{s}_i\Big)
    \!+ \!n_k, 
\end{align}
where $\Tilde{\qg}_{sk}^H = \qh_k^H\bTeta\qH_{2,s}$ denotes the channel between the $s$th subarray and FFUE $k$.

\section{Performance Analysis} \label{Sec:performance analysis}
For NFUE $k$, the received signal in \eqref{eq:receive signa:Near} can be written as
\begin{equation} \label{eq:reformulated:received}
    \bar{y}_{k} = \bDSk  \bar{s}_k + \! \sum\nolimits_{i \in\Kn\setminus k} \! \! \bUIi   \bar{s}_i + \!  \sum\nolimits_{j \in \Kf} \! \! \btUIj   \Tilde{s}_j + n_k,
\end{equation}
where 
\begin{subequations}~\label{eq:FFU:NF}
\begin{align}
    \bDSk &\triangleq \sqrt{P_n} \; \sum\nolimits_{s\in\Set}\sqrt{\betask}\bar{\qg}_{sk}^{H}\Bar{\qD}_{sk}\bar{\qw}_{sk}, ~\label{eq:FFU:NF:DS}\\
    \bUIi &\triangleq \sqrt{P_n} \ {\sum\nolimits_{s\in\Set}\sqrt{\betasi}\bar{\qg}_{sk}^{H}\Bar{\qD}_{si}\bar{\qw}_{si}}, ~\label{eq:FFU:NF:BU}\\
    \btUIj &\triangleq \sqrt{P_f} \ {\sum\nolimits_{s\in\Set}\sqrt{\Tetasj} \bar{\qg}_{sk}^{H}\Tilde{\qD}_{sj}\Tilde{\qw}_{sj}}, 
\end{align}
\end{subequations}
correspond to the desired signal ($\bDSk$), the intra-group interference caused by NFUEs ($\bUIi$) and the inter-group interference caused by FFUEs ($\btUIj$), respectively. Since $\bar{s}_k$ is independent of $\bar{s}_i$ for any $i \neq k$, the first term and the second term of \eqref{eq:reformulated:received} are uncorrelated. Similarly, the third term of \eqref{eq:reformulated:received} is uncorrelated with the first term of \eqref{eq:reformulated:received}. 
Using the arguments from \cite{ngo16} that uncorrelated Gaussian noise represents the worst case, an achievable downlink SE of  NFUE $k$ can be obtained as $\bSEk = \log_2(1+\overline{\text{SINR}}_{k})$, where
\begin{equation} \label{eq:simplified SINR expression of NFUE}
   \overline{\text{SINR}}_{k}= \frac{\vert\bDSk\vert^2}{ \sum\nolimits_{i \in\Kn\setminus k} \!\vert\bUIi\vert^2 \!+\! \sum\nolimits_{j\in\Kf} \!\Ex\big\{\big\vert\btUIj\big\vert^2\big\} \!+\! \sigma_k^2}. 
\end{equation}
Note that the SE of the NFUEs is influenced by the interference caused by signals intended for FFUEs, thereby being a function of the RIS phase shift matrix.
For FFUE $k$, we can rewrite the received signal $\Tilde{y}_{k}$ in \eqref{eq:receive signa:Far} as
\begin{align} \label{eq:received signal expression of FFUE with VR}
    \Tilde{y}_{k} &= \tDSk  \Tilde{s}_k + \tBUk  \Tilde{s}_k + \sum\nolimits_{j \in \Kf\setminus k} \tUIj   \Tilde{s}_j \nonumber \\
    &\hspace{2em}+\sum\nolimits_{i\in \Kn} \tbUIi   \bar{s}_i +  n_k,
\end{align}
where 
\begin{subequations}~\label{eq:FFU:UF}
   \begin{align}
    \tDSk &\triangleq \sqrt{P_f}  \mathbb{E}\left\{\sum\nolimits_{s\in\Set}\sqrt{\Tilde{\eta}_{sk}}\Tilde{\qg}_{sk}^H\Tilde{\qD}_{sk}\Tilde{\qw}_{sk}\right\},~\label{eq:DS} \\
    \tBUk &\triangleq \sqrt{P_f}\Big( \sum\nolimits_{s\in\Set}\sqrt{\Tilde{\eta}_{sk}}\Tilde{\qg}_{sk}^H\Tilde{\qD}_{sk}\Tilde{\qw}_{sk} \nonumber \\
    &\hspace{3em}- \mathbb{E}\left\{\sum\nolimits_{s\in\Set}\sqrt{\Tilde{\eta}_{sk}}\Tilde{\qg}_{sk}^H\Tilde{\qD}_{sk}\Tilde{\qw}_{sk}\right\} \Big),~\label{eq:BU} \\
    \tUIj &\triangleq \sqrt{P_f} \ {\sum\nolimits_{s\in\Set}\sqrt{\Tetasj}\Tilde{\qg}_{sk}^H\Tilde{\qD}_{sj}\Tilde{\qw}_{sj}},~\label{eq:UIj} \\
    \tbUIi &\triangleq \sqrt{P_n} \ {\sum\nolimits_{s\in\Set}\sqrt{\betasi}\Tilde{\qg}_{sk}^H\bar{\qD}_{si}\bar{\qw}_{si}}, ~\label{eq:ULI}
   \end{align} 
\end{subequations}
correspond to the desired signal ($\tDSk$), the beamforming gain uncertainty ($\tBUk$), the intra-group interference caused by the other FFUEs ($\tUIj$) and the inter-group interference caused by NFUEs ($\tbUIi$), respectively. Since $\Tilde{s}_k$ is independent of $\tDSk$ and $\tBUk$, the first term and the second term of \eqref{eq:received signal expression of FFUE with VR} are uncorrelated. Similarly, the third and fourth terms are uncorrelated with the first term of \eqref{eq:received signal expression of FFUE with VR}. According to \cite{ngo16}, the achievable SE at FFUE $k$ is obtained as $\tSEk=\log_2(1+\widetilde{\text{SINR}}_{k})$, where 
\begin{align} \label{eq:simplified SINR expression of FFUE}
    &\widetilde{\text{SINR}}_{k}\!\! =\!\!  \frac{\vert\tDSk\vert^2}{ \Ex\{ \vert\tBUk\vert^2\}\ \!+\!\! \! \sum\limits_{j \in \Kf \setminus k} \!\! \!\!\!\!\Ex\{\vert\tUIj\vert^2\} \!\!+\!\! \sum\limits_{i\in \Kn} \!\!\!\! \Ex\{\vert\tbUIi\vert^2\} \! + \! \sigma_k^2}.
\end{align}

Next, we derive the VR-based SE expressions for different precoding schemes, which are $\MRT$, $\CZF$, and $\LZF$ and assume that VR selection for NFUEs and FFUEs has already been performed on the XL-MIMO array.\footnote{In the context of massive MIMO, it is well known that linear precoding designs can achieve fairly close results to optimal precoding. Therefore, we apply linear precoding to avoid high computational complexity \cite{ngo16, emil17}.}  Consequently, each VR-based precoding scheme assigns a specific VR to each NFUE and FFUE. In Subsection~\ref{subsec:VR}, we introduce our proposed VR selection algorithm. From this point on, we will use the superscripts $\MRT$, $\CZF$, and $\LZF$ to denote the corresponding precoding schemes and their associated parameters.

\subsection{VR-Based MRT Precoding Scheme}
As the channels between the $s$th subarray and the $k$th NFUE (FFUE) are denoted as $\bar{\qg}_{sk}^H \in \mathbb{C}^{1 \times M^{\ast}}$ ($\Tilde{\qg}_{sk}^H \in \mathbb{C}^{1\times M^{\ast}}$), the VR-based MRT precoding vectors in the subarray $s$ for the $k$th NFUE and the $k$th FFUE are given by 
\begin{align}~\label{eq:MRT}
   \bar{\qw}_{sk}^{\MRT} = \bar{\qg}_{sk},\quad\Tilde{\qw}_{sk}^{\MRT} = \Tilde{\qg}_{sk}. 
\end{align}

\begin{proposition}~\label{Prop:MRT:NF} 
The $\text{SE}$ of the $k$th NFUE with VR-based MRT precoding scheme can be represented in closed form as 
$\bSEkMRT = \log_2(1+\bSINRkMRT)$, where
\begin{align}   \label{eq:FSINR:NF:VR}
    &\bSINRkMRT= \frac{P_n\Big(  \sum\nolimits_{s \in \SknMRT} \sqrt{\betask^{\MRT}} \betakks\Big)^2}{ \bar{\Delta}^{\MRT}+ \sigma_k^2},  
\end{align}
with $\bar{\Delta}^{\MRT} \triangleq P_n\!\sum\nolimits_{i \in\Kn\setminus k} \!\Big\vert  \sum\nolimits_{s \in \SniMRT}\!\!\sqrt{\betasi^{\MRT}}  \betakis \Big\vert^2 +\! P_f \sum\nolimits_{j\in\Kf}\!\!  {\sum\nolimits_{s \in \SjfMRT}\!\!\sum\nolimits_{s'\in \SjfMRT}\!\!\sqrt{\Tetasj^{\MRT} \Tetaspj^{\MRT}} \Talphakjssp}$,
where $\betakks \triangleq \bar{\qg}_{sk}^{H} \bar{\qg}_{sk}$, $\betakis \triangleq \bar{\qg}_{sk}^{H} \bar{\qg}_{si}$, $\forall i \neq k$; $\SknMRT$, $\SniMRT$ and $\SjfMRT$ are the VRs chosen by NFUEs $k$, $i$ and FFUE $j$ with MRT precoding scheme, respectively; $ \Talphakjssp \triangleq \Re\Big\{\bar{\qg}_{sk}^{H}(\alpha_2^2\bar{\qH}_{2,s}^H \qM_{jj} \bar{\qH}_{2,s'} +  \delta_{ss'}\varsigma_k\beta_2^2 N \qI_{M^{\ast}})\bar{\qg}_{s'k}\Big\}$, where $\qM_{jj} \triangleq \bTeta^H \qh_j\qh_j^H\bTeta$, while $\qH_{2,s'}\in \mathbb{C}^{N \times M^{\ast}} $ denotes the channel between the $s'$th subarray of XL-MIMO and the RIS, which is a submatrix of $\qH_2 \in \mathbb{C}^{N \times M}$, and we have $\qH_{2,s'}=\qH_2(:,r_{s'})$, where $r_{s'}=\{(s'-1) M^{\ast}+1:s' M^{\ast}\}$.
\end{proposition}
\begin{proof}
    See Appendix~\ref{Prop:MRT:NF:proof}. 
\end{proof}
\begin{proposition}~\label{Prop:MRT:FF}
The SE of  FFUE $k$ with VR-based MRT precoding scheme is given by $\tSEkMRT = \log_2(1 + \tSINRkMRT )$, where
\begin{align} \label{eq:FSINR:FF:VR}
    &\tSINRkMRT =
    \frac{P_f\Big(\sum\nolimits_{s \in \SkfMRT} \!\!\sqrt{\Tilde{\eta}_{sk}^{\MRT}} \Tcks \Big)^2}{  \tilde{\Delta}^{\MRT}+ \sigma_k^2},
\end{align}
where $\SkfMRT$ is the VR chosen by  FFUE $k$, $\Tcks\triangleq \qh_k^H\bTeta (\alpha_2^2 \bar{\qH}_{2,s} \bar{\qH}_{2,s}^H + \beta_2^2 M^{\ast} \qI_N) \bTeta^H\qh_k$, and 
\begin{align}
 \tilde{\Delta}^{\MRT} &\triangleq P_f \! \!\sum\nolimits_{s \in \SkfMRT} \! \!\Tilde{\eta}_{sk}^{\MRT} \Taks\nonumber\\
 &+ \!\!P_f\sum\nolimits_{j\in\Kf\setminus k} \!  \sum\nolimits_{s \in \SjfMRT} \!\sum\nolimits_{s' \in \SjfMRT} \!\! \sqrt{\Tetasj^{\MRT}\Tetaspj^{\MRT}} \Tbkjssp \nonumber\\
 &+P_n \sum\nolimits_{i\in\Kn} \!\! \sum\nolimits_{s\in \SniMRT} \!\!\sum\nolimits_{s' \in \SniMRT} \!\! \sqrt{\bar{\eta}_{si}^{\MRT}\bar{\eta}_{s'i}^{\MRT}} \balphakissp,   
\end{align}
where $\balphakissp\triangleq\Re\big\{\bar{\qg}_{s'i}^T \qB_{ki}^{ss'} \bar{\qg}_{si}^{\ast}\big\}$, $ \qB_{ki}^{ss'} = \alpha_2^2\bar{\qH}_{2,s'}^H\qM_{kk}\bar{\qH}_{2,s} + \delta_{ss'}\varsigma_k \beta_2^2 N\qI_{M^{\ast}}$, with $\qM_{kk}\! \triangleq \!\bTeta^H \qh_k\qh_k^H\bTeta$, and 
\begin{align*} 
    \Taks  &\triangleq  
    \Re\Big\{\qh_k^H\bTeta \big(\varsigma_k \alpha_2^2 \beta_2^2 N \bar{\qH}_{2,s}\bar{\qH}_{2,s}^H \!+\! (\alpha_2^2 \beta_2^2 M^{\ast} \nonumber \\  &\hspace{-1em}\times\qb_N^H(\varphi^a,\varphi^e)\qM_{kk} \qb_N(\varphi^a,\varphi^e) \!+\! \varsigma_k \beta_2^4 NM^{\ast} )\qI_N\big) \bTeta^H \qh_k\Big\},\\
    \Tbkjssp&\triangleq \Re\Big\{\qh_k^H\bTeta \big(\delta_{ss'}\qC_{sj} + (1 - \delta_{ss'})\qD_j^{ss'}\big) \bTeta^H\qh_k\Big\},
\end{align*}
where $\qD_j^{ss'} \triangleq (\alpha_2^2 \bar{\qH}_{2,s} \bar{\qH}_{2,s}^H \!+\! \beta_2^2 M^{\ast} \qI_N) \qM_{jj} (\alpha_2^2 \bar{\qH}_{2,s'} \bar{\qH}_{2,s'}^H  \!+\! \beta_2^2 M^{\ast} \qI_N)$ and
\begin{align}  \label{eq:C_sj}
    &\qC_{sj} 
    \triangleq \alpha_2^4 \bar{\qH}_{2,s}\bar{\qH}_{2,s}^H \qM_{jj} \bar{\qH}_{2,s}\bar{\qH}_{2,s}^H + \beta_2^2\Big(\alpha_2^2  M^{\ast} \bar{\qH}_{2,s}\bar{\qH}_{2,s}^H \qM_{jj} \nonumber \\
    &\hspace{0em} + \alpha_2^2  \trace(\qM_{jj}) \bar{\qH}_{2,s}\bar{\qH}_{2,s}^H  + \alpha_2^2  \trace(\bar{\qH}_{2,s}^H\qM_{jj} \bar{\qH}_{2,s}) \nonumber \\ 
    &\hspace{0em} \!+\! \alpha_2^2  M^{\ast} \qM_{jj} \bar{\qH}_{2,s}\bar{\qH}_{2,s}^H \!+\! \beta_2^2(M^{\ast})^2 \qM_{jj} \!+\! \beta_2^2 M^{\ast}\trace(\qM_{jj}) \qI_N\Big).
\end{align}
\end{proposition}
\begin{proof}
    See Appendix~\ref{Prop:MRT:FF:proof}.
\end{proof}
\subsection{VR-Based CZF Precoding Scheme}
In order to mitigate the intra-group interference caused by signals transmitted to other NFUEs (FFUEs), the CZF precoding method is employed.\footnote{While CZF and LZF effectively handle intra-group interference, they overlook inter-group interference. Although optimal precoding can address this issue, it comes with high computational cost and limited gains. To reduce complexity and manage inter-group interference, we propose using protective MRT (PMRT) or protective ZF (PZF) techniques to shield each user group from interference by others.} For NFUE $k$, the precoding vector of the subarray $s$ is given by
\begin{align}~\label{eq:NF:CZF}
\bar{\qw}_{sk}^{\CZF} = \big[\bar{\qH}_1(\bar{\qH}_1^{H}\bar{\qH}_1)^{-1}\big]_{(r_s,k)} \in \mathbb{C}^{M^{\ast} \times 1}. 
\end{align}
For FFUEs, we set the precoding vector of the subarray $s$ towards FFUE $k$ as
\begin{align}~\label{eq:FF:CZF}
    \Tilde{\qw}_{sk}^{\CZF} = \big[\Tilde{\qG} (\Tilde{\qG}^H\Tilde{\qG})^{-1}\big]_{(r_s,k)} \in \mathbb{C}^{M^{\ast} \times 1}.
\end{align}
Since the VR-based CZF scheme shares similarities with centralized cell-free massive MIMO systems, to eliminate intra-group interference, it is sensible to set the power control coefficients as $\betakCZF\triangleq\bar{\eta}_{1k}^{\CZF}=\ldots=\bar{\eta}_{Sk}^{\CZF}$ for any NFUE $k$~\cite{2024:Mohammadi:survey}. For FFUEs, it follows similarly.

\begin{proposition}~\label{Prop:CZF:NF}
The $\text{SE}$ of NFUE $k$ with VR-based CZF precoding scheme is given by $\bSEkCZF = \log_2\big(1 +\bSINRk^{\CZF}\big)$, where
 \normalsize
\begin{align}   \label{eq:FSINR:NF:VR, ZF}
    &\bSINRk^{\CZF}= \frac{P_n\betakCZF \big(\sum_{s\in \SknZF}\bar{\qg}_{sk}^{H}\bar{\qw}_{sk}^{\CZF}\big)^2 }{
    \bar{\Delta}^{\CZF}+ \sigma_k^2}, 
\end{align}
where
\begin{align}
     \bar{\Delta}^{\CZF}&\triangleq {P_n}\sum\nolimits_{i \in\Kn\setminus k}\betaiCZF\Big\vert \sum\nolimits_{s\in\SinZF}  \bar{\qg}_{sk}^{H} \bar{\qw}_{si}^{\CZF} \Big\vert^2 \nonumber \\
     &\hspace{2em}+ P_f \sum\nolimits_{j\in\Kf} \TetajCZF r_{kj}, 
\end{align}
where $r_{kj}\! \triangleq \Re\Big\{\sum_{s \in \SjfZF}\!\!\sum_{s' \in \SjfZF}\!\bar{\qg}_{sk}^{H}\Ex\big\{\Tilde{\qw}_{sj}^{\CZF} (\Tilde{\qw}_{s'j}^{\CZF})^H \big\}\bar{\qg}_{s'k} \Big\}$, $\SkfZF$, $\SjfZF$ and $\SinZF$ are the VRs selected by  FFUEs $k$, $j$ and  NFUE $i$ under CZF precoding scheme, respectively. 
\end{proposition}
\begin{proof}
    The derivation of $r_{kj}$ follows the proof of $\Ex\{\vert\btUIj\vert^2\}$ in \eqref{eq:UI_J:Near} with MRT and is therefore omitted.
\end{proof}
\begin{proposition}
The SE of FFUE $k$ with VR-based CZF design is given by $\tSEkCZF = \log_2\big(1 +\tSINRk^{\CZF}\big)$, where
\begin{align}   \label{eq:final results of far-field SINR expression with VR, ZF}
    &\tSINRk^{\CZF}= \frac{P_f\TetakCZF \big(\sum_{s\in \SkfZF}\Ex\{\Tilde{\qg}_{sk}^H \Tilde{\qw}_{sk}^{\CZF} \}\big)^2 }{
    \tilde{\Delta}^{\CZF}+ \sigma_k^2
    },  
\end{align}
where
\begin{align}
 \tilde{\Delta}^{\CZF} &\triangleq P_f \TetakCZF \varepsilon_k + \!P_n \sum\nolimits_{i\in\Kn} \betaiCZF t_{ki}  \nonumber\\
 &\hspace{-1em}+P_f \sum\nolimits_{j\in\Kf\setminus k}\!\!\TetajCZF \Ex\Big\{ \Big\vert\sum\nolimits_{s\in \SjfZF} \!\Tilde{\qg}_{sk}^H \Tilde{\qw}_{sj}^{\CZF} \Big\vert^2 \Big\}, 
\end{align}
with $\varepsilon_k \!\triangleq \!\sum\nolimits_{s\in \SkfZF}\! \Ex\Big\{ \big\vert \Tilde{\qg}_{sk}^H \Tilde{\qw}_{sk}^{\CZF} - \Ex\{\Tilde{\qg}_{sk}^H \Tilde{\qw}_{sk}^{\CZF} \}\big\vert^2\Big\}$ and $t_{ki} \triangleq \Re\Big\{\sum\nolimits_{s\in \SniZF} \sum\nolimits_{s' \in \SniZF} (\bar{\qw}_{si}^{\CZF})^{H} \qB_{ki}^{ss'} \bar{\qw}_{si}^{\CZF}\Big\}$.
\end{proposition}
\begin{proof}
    The derivation of $t_{ki}$ follows the proof of $\Ex\{\vert\tbUIi\vert^2\}$ in \eqref{eq:UI_I:Far} with MRT and is therefore omitted.    
\end{proof}
\subsection{VR-Based LZF Precoding Scheme}
To reduce the computational complexity of the CZF precoding scheme associated with channel inversion, we propose a low-complexity scheme, referred to as LZF, which performs zero-forcing locally. Specifically, instead of using all VRs, each subarray $s$ independently applies precoding to nullify intra-group interference within its own subarray. However, it cannot eliminate intra-group interference caused by other subarrays. We notice that LZF employs the same VRs as CZF for simplicity.

For NFUE $k$ served by subarray $s$, the precoding vector with the LZF scheme is given by 
\begin{equation}
    \bar{\qw}_{sk}^{\LZF} = \Big[\bar{\qG}_s(\bar{\qG}_s^H \bar{\qG}_s)^{-1}\Big]_{(:,k)} \in \mathbb{C}^{M^{\ast}\times 1},
\end{equation}
where $\bar{\qG}_s^H \!= \![\bar{\qg}_{s1}, \ldots, \bar{\qg}_{sK_n}\!]^H \!\!\in\! \mathbb{C}^{K_N\times M^{\ast}}$.

Moreover, the LZF precoding vector for FFUE $k$ served by subarray $s$ is 
\begin{equation}
    \Tilde{\qw}_{sk}^{\LZF} = \Big[\Tilde{\qG}_s (\Tilde{\qG}_s^H \Tilde{\qG}_s) ^{-1}\Big]_{(:,k)} \in \mathbb{C}^{M^*\times 1},
\end{equation}
where $\Tilde{\qG}_s^H \!\!=\!\! [(\Tilde{\qg}_{s1}^H)^T, \ldots, (\Tilde{\qg}_{s K_f}^H)^T]^H \!\in\! \mathbb{C}^{K_f\times M^*}$.
Note that if the $s$th subarray serves a single user $k$, LZF precoder at the $s$th subarray reduces to MRT, that is,  $\bar{\qw}_{sk}^{\LZF} = \bar{\qw}_{sk}^{\MRT}$ and $\Tilde{\qw}_{sk}^{\LZF} = \Tilde{\qw}_{sk}^{\MRT}$.
\begin{proposition} 
    The $\text{SE}$ of NFUE $k$ with VR-based LZF precoding scheme is given by $\bSEkLZF = \log_2\Big(1 +\bSINRk^{\LZF}\Big)$, where
\begin{align}   \label{eq:FSINR:NF:VR, LZF}
    \bSINRk^{\LZF} =\frac{P_n \big(\sum_{s\in \SknLZF} \sqrt {\bar{\eta}_{sk}^{\LZF} } \big)^2 }{ \bar{\Delta}^{\LZF} + \sigma_k^2}, 
\end{align}
with $\bar{\Delta}^{\LZF}\triangleq P_n\sum\nolimits_{i \in\Kn\setminus k}\Big\vert\sum\nolimits_{s\in\SinLZF} \!\!\sqrt{\bar{\eta}_{si}^{\LZF} } \bar{\qg}_{sk}^{H} \bar{\qw}_{si}^{\LZF} \Big\vert^2 + P_f \sum\nolimits_{j\in\Kf}\sum\nolimits_{s\in\SjfLZF}\!\!\sum\nolimits_{s'\in\SjfLZF}\!\! \sqrt{\Tetasj^{\LZF} \Tetaspj^{\LZF}}  \Ttaukjsspi$,
with $
    \Ttaukjsspi\triangleq \Re\Big\{\bar{\qg}_{sk}^{H}\Ex\big\{\Tilde{\qw}_{sj}^{\LZF} (\Tilde{\qw}_{s'j}^{\LZF})^H \big\}\bar{\qg}_{s'k}\Big\}$.
\end{proposition}
\begin{proof}
    The derivation of $\Ttaukjsspi$ follows the proof of $\Ex\{\vert\btUIj\vert^2\}$ in \eqref{eq:UI_J:Near} with MRT and is therefore omitted.    
\end{proof}
\begin{proposition}
The SE of FFUE $k$ with VR-based LZF design is given by  $\tSEkLZF = \log_2\Big(1 +\tSINRk^{\LZF}\Big)$, where
\begin{align}   \label{eq:final results of far-field SINR expression with VR, LZF}
    & 
    \tSINRk^{\LZF} =\frac{P_f (\sum\nolimits_{s\in \SkfLZF}\sqrt{\Tetask^{\LZF}})^2 }{ \tilde{\Delta}^{\LZF}+ \sigma_k^2}, 
\end{align}
where $\tilde{\Delta}^{\LZF}\triangleq P_f \sum\nolimits_{j\in\Kf\setminus k} \Ex\Big\{ \Big\vert\sum\nolimits_{s\in \SjfLZF\!\!} \sqrt{\Tetasj^{\LZF}}\Tilde{\qg}_{sk}^H \Tilde{\qw}_{sj}^{\LZF} \Big\vert^2 \Big\}+ P_n \sum\nolimits_{i\in\Kn} \!\!
\sum\nolimits_{s\in \SniLZF} \!\!
\sum\nolimits_{s'\in \SniLZF} \!\!\sqrt{\betasi^{\LZF} \betaspi^{\LZF}} \btaukisspi$,
with $\btaukisspi \triangleq \Re\big\{(\bar{\qw}_{s'i}^{\LZF})^{H} \qB_{ki}^{ss'} \bar{\qw}_{si}^{\LZF}\big\}$.
\end{proposition}
\begin{proof}
    The derivation of $\btaukisspi$ follows the proof of $\Ex\{\vert\tbUIi\vert^2\}$ in \eqref{eq:UI_I:Far} with MRT and is therefore omitted.        
\end{proof}

\begin{Remark}
It is worth noting that, since the majority of the received power is concentrated within each user’s VR, our approach avoids involving the entire antenna array in the processing for every user. This design choice ensures that the computational burden does not scale proportionally with the total number of antennas. Moreover, in the proposed LZF precoding scheme, instead of jointly processing all VRs, each subarray $s$ performs precoding independently to suppress intra-group interference within its own region. This localized and decoupled processing structure eliminates the need for costly large-scale matrix inversions. Consequently, the computational complexity of the precoding stage is reduced from $\OO(M\times K_f^2)$ to $\OO(M^{\ast} \times K_f^2)$ in LZF, where $M^{\ast}$ is the number of antennas of a subarray and $M^{\ast} \ll M$.
\end{Remark}

\section{Optimization problem} \label{Sec:Solution_Optimization}
In this section, we optimize the phase shifts at the RIS and the power control coefficients at the XL-MIMO array to maximize the sum of the weighted minimum SE for both NFUEs and FFUEs, while ensuring that individual QoS requirements are met for all users. At the optimized solution, NFUEs and FFUEs achieve the same SE within their respective groups. This fairness-oriented optimization framework guarantees seamless connectivity for all users, aligning with the core principles of next-generation wireless networks. The problem can be formulated as
\begin{subequations}\label{eq:original optimization problem}
\begin{alignat}{2}
(\text{P}1):\hspace{1em}&\underset{\{\bar{\boldsymbol{\eta}}^{\mathrm{i}}, \tilde{\boldsymbol{\eta}}^{\mathrm{i}},\US^{\mathrm{i}}, \IS^{\mathrm{i}}, \boldsymbol{\theta}\}}{\max}       
&~& w_n \min_{ k \in \Kn} \bSEk^{\mathrm{i}} + w_f \min_{ k \in \Kf} \tSEk^{\mathrm{i}}  \label{eq:optProb}\\
&\hspace{2em}\text{s.t.} 
&         &\sum\nolimits_{s\in\Set} \Big(P_n \sum\nolimits_{k \in \US} \betask^{\mathrm{i}} \bar{c}_{sk}^{\mathrm{i}}\nonumber\\
&&&+ P_f \sum\nolimits_{k \in \IS} \Tilde{\eta}_{sk}^{\mathrm{i}}\Tilde{c}_{sk}^{\mathrm{i}}\Big) \leq P,\label{eq:power control requirements}\\ 
&         &      &\bSEk^{\mathrm{i}}\geq \bSEkth, \forall k\in\Kn, \label{eq:P1:QoS:NF}    \\
&         &      &\tSEk^{\mathrm{i}}\geq \tSEkth, \forall k\in\Kf, \label{eq:P1:QoS:FF}    \\
&         &      &\betask^{\mathrm{i}}\geq 0, \forall s, \forall k\in \US, \label{eq:power control:NF}    \\
&         &      &\Tetask^{\mathrm{i}} \geq 0, \forall s, \forall k\in \IS, \label{eq:power control:FF}    \\
&         &      & 0 \leq \theta_n \leq 2\pi, n=1,\ldots,N,
\end{alignat}
\end{subequations}
where 
$\BEtaN^{\mathrm{i}}=[(\BEtaN_1^{\mathrm{i}})^T,\ldots,(\BEtaN_{K_N}^{\mathrm{i}})^T]^T$, with $\BEtaN_k^{\mathrm{i}}=[\bar{\eta}_{1k}^{\mathrm{i}},\ldots,\bar{\eta}_{Sk}^{\mathrm{i}}]^T$, is a vector of power control coefficients for the NFUEs; $\BEtaF^{\mathrm{i}}=[(\BEtaF_1^{\mathrm{i}})^T,\ldots,(\BEtaF_{K_F}^{\mathrm{i}})^T]^T$, with $\BEtaF_k^{\mathrm{i}} =[\Tilde{\eta}_{1k}^{\mathrm{i}},\ldots,\Tilde{\eta}_{Sk}^{\mathrm{i}}]^T$, represents power control coefficients for the FFUEs, where ${\mathrm{i}}\in\{\CZF,\LZF,\MRT\}$, respectively; $\US^{\mathrm{i}}$ and $\IS^{\mathrm{i}}$ denote the sets of NFUEs and FFUEs served by the subarray $s$; $w_n$ and $w_f$ are the positive weighting factors of the NF and FF parts, respectively. Constraints~\eqref{eq:P1:QoS:NF} and~\eqref{eq:P1:QoS:FF} denote the QoS requirements, where $\bSEkth$ and $\tSEkth$ are the minimum SE requirements for NFUEs and FFUEs, respectively. 

The optimization problem $(\text{P}1)$ is highly intractable due to its non-convex objective function with respect to $\boldsymbol{\theta}$, $\bar{\boldsymbol{\eta}}^{\mathrm{i}}, \tilde{\boldsymbol{\eta}}^{\mathrm{i}},\US^{\mathrm{i}}, \IS^{\mathrm{i}}$. Finding the global maximizer of the weighted sum SE, subject to SINR and RIS phase shifts constraints, is extremely challenging. To address this, we decompose the problem into three stages. First, we employ a heuristic algorithm to determine suitable VR assignments for each user, resulting in the sets
$\US^{\mathrm{i}}$ and $\IS^{\mathrm{i}}$, while satisfying individual QoS requirements. Given these sets, we then propose a two-stage algorithm that decouples the original problem into two more tractable subproblems: RIS phase shifts design and power allocation. This decomposition significantly reduces the computational complexity by breaking the coupling between variables. More specifically, we adopt the penalty method to obtain an optimized phase-shift design and SCA to design an optimized power control algorithm.
\vspace{-1em}
\subsection{Heuristic VR Selection Algorithm}~\label{subsec:VR}
We propose a heuristic VR selection algorithm for linear precoding that ensures the QoS of NFUEs and FFUEs with minimum number of subarrays. Without loss of generality, we focus on the CZF scheme. For NFUE $k$, we first initialize a set $\bar{\boldsymbol{\mathcal{D}}}_k^{\CZF} = \diag(\bar{\qD}_{1k}, \ldots, \bar{\qD}_{Sk})$, with entries $\bar{\qD}_{sk} = \qI_{M^\ast}$, indicating that NFUE $k$ initially utilizes all subarrays. Next, for each subarray $s$, we calculate the SINR of NFUE $k$ without employing subarray $s$ as $\overline{\text{SINR}}_{k}^{s}$ to assess whether the SINR constraint $\bar{\gamma}_k$ for that user is met. We define $\bar{\gamma}_k = \delta\overline{\text{SINR}}_{k}$, where the VR selection ratio $\delta \in [0,1]$, and $\overline{\text{SINR}}_{k}$ can be derived from \eqref{eq:simplified SINR expression of NFUE}. If $\overline{\text{SINR}}_{k}^{s}$ is greater than $\bar{\gamma}_k$, it implies that the remaining subarrays are sufficient for NFUE $k$. Then, we update the set $\bar{\boldsymbol{\mathcal{D}}}_k^{\CZF}$ by setting $\bar{\qD}_{sk} = \mathbf{0}_{M^\ast}$. Finally, we obtain the matrix $\bar{\boldsymbol{\mathcal{D}}}_k^{\CZF}$ which represents a set of minimum subarrays that meet the SINR requirement of  NFUE $k$ with CZF.

The proposed heuristic VR selection algorithm for NFUEs is summarized in \textbf{Algorithm 1}. The algorithm for choosing VRs for FFUEs is similar. The difference is that we compare $\widetilde{\text{SINR}}_{k}^{s}$, which denotes the SINR without incorporating the subarray $s$ for FFUE $k$, with the SINR constraint $\Tilde{\gamma}_k = \delta\widetilde{\text{SINR}}_{k}$, where $\widetilde{\text{SINR}}_{k}$ can be obtained from \eqref{eq:simplified SINR expression of FFUE}. Consequently, we obtain $\Tilde{\boldsymbol{\mathcal{D}}}_k^{\CZF}$ with the CZF precoding scheme. Note that we assume that LZF utilizes the same VRs as CZF for tractability, while the corresponding VRs for MRT can be derived similarly.

The main differences between our proposed VR design and the method in~\cite{Zhi:JSAC:2024} are twofold. First, we use SINR as our VR selection criterion instead of SNR, which is more practical as it accounts for interference from signals intended for other NFUEs. Second, the complexity of the method in~\cite{Zhi:JSAC:2024} is 
$\OO(K S \log(S))$, whereas the complexity of our proposed VR design is $\OO(K S)$. This is because, for each user, we sequentially evaluate each subarray by removing it and checking whether the resulting SINR still meets the QoS requirement, resulting in $\OO(S)$ complexity per user.
\vspace{0em}
\subsection{Phase-Shift Design at the RIS}
By analyzing the SE results for different precoding schemes, we observe that these expressions are highly complex functions of the phase shift matrix. Moreover, the objective function and the constraints in~\eqref{eq:power control requirements}-\eqref{eq:P1:QoS:FF} are intricately coupled with respect to this parameter. To address this challenge, we follow the same approach as in~\cite{Wu:TWC:2019}, \cite{Mu:TWC:2022} and maximize the minimum channel gain of the FFUEs to design the optimized phase shift matrix $\boldsymbol{\Theta}$ at the RIS with $\MRT$, and $\CZF$. The optimization problem can be formulated as 
\begin{subequations}
\begin{alignat}{2}
    (\text{P}2) :\hspace{1em} &\underset{\boldsymbol{\Theta}}{\max} &~& \min_{k} f_k(\boldsymbol{\Theta}) \\
    &\text{s.t.} \quad && 
    0\leq\theta_n \leq 2\pi, \forall n,   
\end{alignat}    
\end{subequations}
where $f_k(\boldsymbol{\Theta}) \triangleq \Ex\big\{ \normLt{\Tilde{\qg}_{k}^H}^2\big\}$, which can be expressed as
\begin{align}
    f_k(\boldsymbol{\Theta}) & = \Ex\big\{ \qh_k^H\bTeta \qH_{2} \Tilde{\boldsymbol{\mathcal{D}}}_k^{\mathrm{i}} (\Tilde{\boldsymbol{\mathcal{D}}}_k^{\mathrm{i}})^H \qH_{2}^H  \bTeta^H\qh_k \big\} \nonumber \\
    &=\qv^H \qB_k \qA \qB_k^H \qv, 
\end{align}
where $\qv\! \triangleq\! [e^{j\theta_1}, \ldots, e^{j\theta_N}]^H$,  $\qB_k \!\!\triangleq\!\! \diag(\qh_k^H)$ and $\qA_k \!\triangleq\! \alpha_2^2 \bar{\qH}_{2} \Tilde{\boldsymbol{\mathcal{D}}}_k^{\mathrm{i}} (\Tilde{\boldsymbol{\mathcal{D}}}_k^{\mathrm{i}})^H \bar{\qH}_{2}^H + \beta_2^2 \trace(\Tilde{\boldsymbol{\mathcal{D}}}_k^{\mathrm{i}} (\Tilde{\boldsymbol{\mathcal{D}}}_k^{\mathrm{i}})^H) \qI_N \in  \mathbb{C}^{N\times N}$. We employ SDP to reformulate problem (\text{P}2). To this end, we define $\qV\! \triangleq\! \qv\qv^H\in \mathbb{C}^{N\times N}$, and $\qR_{k}\! \triangleq \!\qB_k \qA_k \qB_k^H \in \mathbb{C}^{N\times N}$. Then, (\text{P}2) can be rewritten as
\begin{subequations} \label{eq:original phase-shift optimization problem}
\begin{alignat}{2} 
    (\text{P}2.1): \hspace{1em}&\underset{\qV, t \geq 0} {\max} &~& t   \\
    &\hspace{0.5em}\text{s.t.} \quad && 
    t \leq \trace( \qR_{k} \qV), k \in\Kf, \label{eq:t:constraint}\\
    &&&\qV_{n,n} = 1, \forall n, \label{eq:V:constraint}\\
    &&&\rank (\qV) =1, \label{eq:V:rankone}\\
    &&& \qV \succeq {\bf 0}. \label{eq:V:SDP}
\end{alignat}    
\end{subequations}
Note that Problem $(\text{P}2.1)$ is non-convex due to the rank-one constraint. To address this issue, we replace the non-convex rank-one constraint with a difference-of-convex (DC) function formulation, expressed as follows \cite{Yang:2020:TWC,Tao:2019:Glob}:
\begin{align}   \label{eq:rank-one condition}
   \rank (\qV) =1 \Leftrightarrow \| \qV\|_{\ast} - \| \qV\|_2 = 0.
\end{align}
In general, $\|\qV\|_{\ast} - \| \qV\|_2 \geq 0$, and equality holds when we have the rank-one solution. Then, we add the DC function in~\eqref{eq:rank-one condition} into
the objective function as a penalty component, yielding
\begin{subequations}  \label{eq:optimization problem for solving Theta}
\begin{alignat}{2}
    (\text{P}2.2):\hspace{1em} &\underset{\qV, t\ge 0}{\max} \quad &~& t - \varrho(\| \qV\|_{\ast} - \| \qV\|_2)  \\
    &\hspace{0.5em}\text{s.t.} \quad && 
    \eqref{eq:t:constraint},    
    \eqref{eq:V:constraint}, \eqref{eq:V:SDP},
    \end{alignat}    
\end{subequations} 
where $\varrho > 0$ is the penalty factor and penalizes the objective function if $\qV$ is not a rank-one matrix. Similarly to \cite{Mu:TWC:2022}, we gradually increase $\varrho$ to obtain a rank-one matrix $\qV$. 

Note that problem $(\text{P}2.2)$ is still non-convex because of the convex term $\varrho \| \qV\|_2$ in the objective function. To overcome this, we use SCA.\footnote{ The SCA is an iterative optimization technique used to tackle non-convex problems by solving a sequence of convex subproblems. The core idea is to transform problem~\eqref{eq:optimization problem for solving Theta} into a series of tractable subproblems by linearizing the convex term $\varrho \| \qV\|_2$ in the objective function. In each iteration, a local convex approximation is solved, improving upon the previous solution. This iterative process continues until convergence to a near-optimal solution \cite{Beck:2010:JGlobOptim}.} To this end, we find the lower bound for $\| \qV \|_2$ as $ \| \qV \|_2 \geq \bar{\qV}^{(n)}$,
where $\bar{\qV}^{(n)} \!=\! ( \| \qV^{(n)} \|_2 + \trace(\qu( \qV^{(n)}) \qu( \qV^{(n)})^H (\qV \!-\! \qV^{(n)})) )$\cite{Mu:TWC:2022,Yang:2020:TWC}. Accordingly, we obtain the following optimization problem:
\begin{subequations}  \label{eq:final optimization problem for solving Theta}
\begin{alignat}{2}
    (\text{P}2.3):\hspace{1em} &\underset{\qV,t\ge 0}{\max} \quad &~& t - \varrho(\| \qV\|_{\ast} - \bar{\qV}^{(n)} )  \\
    &\hspace{0.5em}\text{s.t.} \quad && 
    \eqref{eq:t:constraint},    
    \eqref{eq:V:constraint}, \eqref{eq:V:SDP}.
\end{alignat}    
\end{subequations}
It is clear that problem \eqref{eq:final optimization problem for solving Theta} is a standard SDP which can be solved efficiently using existing solvers such as CVX \cite{cvx}. The penalty-based algorithm for solving \eqref{eq:original phase-shift optimization problem} consists of two nested loops and is executed iteratively until convergence. In the outer loop, we gradually increase the penalty factor $\varrho$. Specifically, we increase $\varrho = l \varrho$ until $\| \qV\|_{\ast} - \| \qV\|_2 \leq \epsilon$, where $l > 1$ denotes the scaling factor, and $\epsilon$ is a predefined threshold. It means that as $\varrho$ increases, the solutions satisfy $\| \qV\|_{\ast} - \| \qV\|_2 \rightarrow 0$, ensuring near rank-one $\qV$. In the inner loop, with the fixed penalty factor, we solve the problem \eqref{eq:final optimization problem for solving Theta} to obtain the optimized solution $\qV^{\ast}$. The proposed algorithm is outlined in \textbf{Algorithm 2}.

\begin{algorithm}[t]
\caption{Heuristic VR Selection Algorithm} \label{alg:choosing VR} 
	\begin{algorithmic}[1]
		\For {$k=1,\ldots,K_n$} 
                \State set $\bar{\boldsymbol{\mathcal{D}}}_k^{\CZF} = \diag(\bar{\qD}_{1k}, \ldots, \bar{\qD}_{Sk})$, where $\bar{\qD}_{sk} = \qI_{M^\ast}$;
		      \For {$s = 1,\ldots,S$} 
                    \State Calculate the SINR without employing subarray \Statex \hspace{3em} $s$ as  $\overline{\text{SINR}}_{k}^{s}$; 
                    \If { $\overline{\text{SINR}}_{k}^{s} \geq \bar{\gamma}_k$}
                        \State Set $\bar{\qD}_{sk} = \mathbf{0}_{M^\ast}$;
                    \EndIf
		      \EndFor 
                \State Get set of subarrays $\bar{\boldsymbol{\mathcal{D}}}_k^{\CZF}$ chosen by NFUE $k$.
		\EndFor
	\end{algorithmic} 
\end{algorithm}
\subsection{Power Control Design at the XL-MIMO Array}
With $\US^{\mathrm{i}}$, $\IS^{\mathrm{i}}$ obtained from the heuristic VR selection algorithm and the optimized phase shifts $\boldsymbol{\theta}^{\ast}$, problem $(\text{P}1)$ is reduced to 
\begin{subequations}    \label{eq:Original power allocation optimization problem}
\begin{alignat}{2}   
    (\text{P}3): &\underset{\{\bar{\boldsymbol{\eta}}^{\mathrm{i}}, \tilde{\boldsymbol{\eta}}^{\mathrm{i}}, \bar{t}^{\mathrm{i}}, \Tilde{t}^{\mathrm{i}}\}}{\max}
    \quad &~& w_n \bar{t}^{\mathrm{i}} + w_f \Tilde{t}^{\mathrm{i}}  \\
    &\hspace{1.5em}\text{s.t.} \quad 
    & & \bSEk^{\mathrm{i}} \geq \bar{t}^{\mathrm{i}}, \forall k \in \Kn, \label{eq:Near-field SE constraint}\\
    &&&\tSEk^{\mathrm{i}} \geq \Tilde{t}^{\mathrm{i}},  \forall k \in \Kf, \label{eq:Far-field SE constraint}\\
    &&& \tilde{t}^{\mathrm{i}}\geq 0, \bar{t}^{\mathrm{i}}\geq 0 \label{eq:t constraints}, \\
    &&& \eqref{eq:power control requirements}, 
    \eqref{eq:power control:NF}, \eqref{eq:power control:FF}, \label{eq:common:constrains}
\end{alignat}    
\end{subequations}
where $\bar{t}^{\mathrm{i}} \!\triangleq\! \min\limits_{k\in\Kn} \bSEk^{\mathrm{i}}$, $\Tilde{t}^{\mathrm{i}} \triangleq \min\limits_{k\in\Kf} \tSEk^{\mathrm{i}}$. 
Since $\bSEk^{\mathrm{i}} = \log_2(1+\bSINRk^{\mathrm{i}})$, we can replace $\eqref{eq:Near-field SE constraint}$ with the following two constraints
\begin{subequations}
\begin{align}
    &\bSINRk^{\mathrm{i}} \geq \bar{T}^{\mathrm{i}}, \forall k \in \Kn, \label{eq:Near-field SINR constraint} \\
    &\bar{T}^{\mathrm{i}} \geq 2^{\bar{t}^{\mathrm{i}}} - 1. \label{eq:Near-field T constraint}   
\end{align}
\end{subequations}
Similarly, $\eqref{eq:Far-field SE constraint}$ can be equivalently replaced by constraints
\begin{subequations}
\begin{align}
    &\tSINRk^{\mathrm{i}} \geq \Tilde{T}^{\mathrm{i}}, \forall k \in \Kf, \label{eq:Far-field SINR constraint} \\
    &\Tilde{T}^{\mathrm{i}} \geq 2^{\Tilde{t}^{\mathrm{i}}} - 1 \label{eq:Far-field T constraint}.   
\end{align}
\end{subequations}
\begin{algorithm}[t]
\caption{Penalty-Based Iterative Algorithm for \eqref{eq:original phase-shift optimization problem}} \label{alg:Penalty-Based Algorithm} 
\begin{algorithmic}[1]
    \State Initialization: set initial point $\qV^{(0)}$, the penalty factor $\varrho$, maximum number of iterations $I_1$, $I_2$ for outer loop and inner loop.
    \Repeat
    \State Set iteration index $n=0$ for inner loop; 
        \Repeat
        \State Solve problem \eqref{eq:final optimization problem for solving Theta}, get the solution $(\qV^{\ast}, t^{\ast})$;  
        \State Set $n=n+1$, update $\qV^{(n)} = \qV^{\ast}$;
        \Until $\| \qV\|_{\ast} - \| \qV\|_2 \leq \epsilon$ or $n=I_2$.
    \State Set $\varrho = l \varrho$;
    \State Update $\qV^{(0)}$ with $\qV^{(n)}$.
    \Until $\| \qV\|_{\ast} - \| \qV\|_2 \leq \epsilon$ or $n=I_1$.    
\end{algorithmic} 
\end{algorithm}
\noindent Then, $(\text{P}3)$  can be recast as 
\begin{subequations} \label{eq:linear objective function:original}
\begin{alignat}{2}
    (\text{P}3.1): &\underset{\{\bar{\boldsymbol{\eta}}^{\mathrm{i}}, \tilde{\boldsymbol{\eta}}^{\mathrm{i}}, \bar{t}^{\mathrm{i}}, \Tilde{t}^{\mathrm{i}}, \bar{T}^{\mathrm{i}}, \Tilde{T}^{\mathrm{i}} \}}{\max} &~& w_n \bar{t}^{\mathrm{i}} + w_f \Tilde{t}^{\mathrm{i}}  \\
    &\hspace{2.7em}\text{s.t.} & & 
    \eqref{eq:t constraints}, 
    \eqref{eq:common:constrains},      
    \eqref{eq:Near-field T constraint},
    \eqref{eq:Far-field T constraint},   
    \label{eq:convex constraints} \\
    &&& \eqref{eq:Near-field SINR constraint}, 
    \eqref{eq:Far-field SINR constraint}.
\end{alignat}   
\end{subequations}

We observe that all the constraints listed in \eqref{eq:convex constraints} are convex and the difficulty in solving $(\text{P}3.1)$ is due to the non-convexity of the constraints in $\eqref{eq:Near-field SINR constraint}$ and $\eqref{eq:Far-field SINR constraint}$. To this end, we propose a SCA-based design that approximates $\eqref{eq:Near-field SINR constraint}$, $\eqref{eq:Far-field SINR constraint}$ by convex approximation functions at each iteration. Note that the optimization problem varies for different precoding schemes due to differences in SINR expressions and sets of power control coefficients.
\subsubsection{Power Control Design with MRT}
Since $\eqref{eq:Near-field SINR constraint}$ and $\eqref{eq:Far-field SINR constraint}$ follow the same structure, we only present the iterative approximation for $\eqref{eq:Near-field SINR constraint}$. Then, following the same approach, we approximate $\eqref{eq:Far-field SINR constraint}$. To this end, we first define $\bxisk \triangleq \sqrt{\betask^{\MRT}}, \forall k\in\Kn$, $\Txisk \triangleq \sqrt{\Tetask^{\MRT}}, \forall k\in\Kf$. Then, we set $\BXiN_k \triangleq [\bar{\xi}_{1k},\ldots,\bar{\xi}_{Sk}]^T$, and $\BXiN \triangleq [\BXiN_{1}^T,\ldots,\BXiN_{K_N}^T]^T$ for NFUEs and $\BXiF_k \triangleq [\Tilde{\xi}_{1k},\ldots,\Tilde{\xi}_{Sk}]^T$, and $\BXiF\triangleq [\BXiF_1^T,\ldots,\BXiF_{K_F}^T]^T$ for FFUEs. Consequently, the optimization problem (\text{P}3.1) with MRT can be reformulated as
\begin{subequations} \label{eq:linear objective function:MRT}
  \begin{alignat}{2}
    (\text{P}3.2):\hspace{1em} &\underset{\qq^\MRT}{\max} &~&  w_n \bar{t}^{\MRT} + w_f \Tilde{t}^{\MRT}  \\
    &\hspace{0.5em}\text{s.t.} 
    &         &\sum\nolimits_{s\in\Set} \Big(P_n \sum\nolimits_{k \in \US} \bar{\xi}_{sk}^2 \bar{c}_{sk}^{\MRT}\nonumber\\
    &&&+ P_f \sum\nolimits_{k \in \IS} \Tilde{\xi}_{sk}^2 \Tilde{c}_{sk}^{\MRT}\Big) \leq P,\label{eq:power:MRT}\\
    &&& \bSINRk^{\MRT}(\BXiN,\BXiF) \geq \bMRT, \forall k \in \Kn \label{eq:Near-field SINR constraint, MRT}, \\
    &&& \tSINRk^{\MRT}(\BXiN,\BXiF) \geq \tMRT, \forall k \in \Kf \label{eq:Far-field SE constraint, MRT}, \\
    &&& \bxisk \geq 0, \Txisk \geq 0, \forall s, \forall k \label{eq:Power allocation coefficients constraints}, \\
    &&&
    \eqref{eq:t constraints},
    \eqref{eq:Near-field T constraint},   \eqref{eq:Far-field T constraint}, 
\end{alignat}   
\end{subequations}
where $\qq^\MRT\!\triangleq\!\{\BXiNMRT,\BXiFMRT,\bar{t}^{\MRT}, \Tilde{t}^{\MRT},\bMRT, \tMRT\}$. Now, the remaining non-convexity of (\text{P}3.2) is the non-convex constraints \eqref{eq:Near-field SINR constraint, MRT} and~\eqref{eq:Far-field SE constraint, MRT}. To address this problem, we apply SCA to approximate \eqref{eq:Near-field SINR constraint, MRT}. Note that \eqref{eq:Near-field SINR constraint, MRT} can be written as 
\begin{align} \label{eq:simplified near-field SINR constraint}
    \frac{P_n\Big(  \sum\nolimits_{s \in \SknMRT} \bxisk \betakks\Big)^2}{\bMRT } \geq  \bar{F}_k^\MRT(\BXiN, \BXiF), \forall k \in \Kn,    
\end{align}
where $ \bar{F}_k^\MRT(\BXiN, \BXiF) \triangleq 
    P_n\!\sum\nolimits_{i \in\Kn\setminus k} \Big(\big\vert  \sum\nolimits_{s \in \SniMRT}\bxisi  \betakis \big\vert^2 \Big)+
    P_f\sum\nolimits_{j\in\Kf}\!\!  {\sum\nolimits_{s \in \SjfMRT}\!\!\sum\nolimits_{s'\in \SjfMRT}\! {\Txisj \Txispj} \Talphakjssp} \!+ \sigma_k^2.$
 The left-hand side (LHS) of $\eqref{eq:simplified near-field SINR constraint}$ is a quadratic-over-linear function that is not jointly concave with respect to $\bxisk$ and $\bMRT$ on the domain $\bMRT > 0$. In light of SCA, the LHS of $\eqref{eq:simplified near-field SINR constraint}$ can be approximated by the following inequality~\cite{Mohammadi:JSAC:2023}
\begin{align} \label{eq:lower bound of x square divides y}
    \frac{x^2}{y} \geq \frac{x^{(n)}}{y^{(n)}} \Big(2x - \frac{x^{(n)}}{y^{(n)}}y \Big).
\end{align}

Moreover, on the right-hand side (RHS) of~\eqref{eq:simplified near-field SINR constraint}, $\bar{F}_k^{\MRT}(\BXiN, \BXiF)$ is non-convex due to the presence of the term $\Txisj\Txispj$. To address this issue, we replace the non-convex term by the first Taylor expansion around a point $(x^{(n)},y^{(n)})$ as follows \cite{Mohammadi:JSAC:2023}
\begin{align}
    \label{eq:xy:ub}
    & 4xy \! \leq  (x\!+\!y)^2\!-\!2(x^{(n)}\!-\!y^{(n)})(x\!-\!y) 
    \!+\! (x^{(n)}\!-\!y^{(n)})^2.
\end{align}
\begin{figure*}
\normalsize 
\vspace*{5pt} 
\begin{align}  
    &\Bar{F}_k^{\MRT}(\BXiN, \BXiF) \leq 
     \Bar{F}_k^{\MRT,ub}(\BXiN, \BXiF) \triangleq
    \frac{\rho_n}{4}\!
    \sum\nolimits_{i \in\Kn\setminus k}\!\! 
    \bigg(\sum\nolimits_{s \in \SniMRT}\!\!
    4\bxisi^2 \vert\betakis \vert^2
    + 2 \sum\nolimits_{s \in \SniMRT}\!\!
    \sum\nolimits_{s' \in\SniMRT,s'>s}\!\! \Big[(\bxisi+\bxispi)^2-2
    \Big(\bxisi^{(n)}-\bxispi^{(n)}\Big)
   \nonumber\\
    &\hspace{6em}
    \times(\bxisi-\bxispi)+\Big(\bxisi^{(n)}-\bxispi^{(n)}\Big)^2\Big]\Re\{\betakis \betakispcj\} \bigg) + 
    \frac{\rho_f}{4}\sum\nolimits_{j\in\Kf} \!\! \sum\nolimits_{s \in \SjfMRT} \!\! \sum\nolimits_{s'\in \SjfMRT}\!\!\Big[(\Txisj +\Txispj)^2  \nonumber\\
    &\hspace{6em}
    -2
    \Big(\Txisj^{(n)} -\Txispj^{(n)}\Big)(\Txisj -\Txispj)+\Big(\Txisj^{(n)} -\Txispj^{(n)}\Big)^2\Big] \Talphakjssp \!+\!1, \label{eq:bFMRT} 
\end{align}
\vspace{-2.5em}
\end{figure*}
Consequently, we obtain the convex upper bound of $\bar{F}_k^\MRT(\BXiN, \BXiF)$ as $\bar{F}_k^{\MRT,ub}(\BXiN, \BXiF)$, shown in~\eqref{eq:bFMRT}, where $\rho_n = \!\frac{P_n}{\sigma_k^2}$ and $\rho_f = \!\frac{P_f}{\sigma_k^2}$.  Regarding \eqref{eq:Far-field SE constraint, MRT}, we have 
\vspace{-0.2em}
\begin{align} \label{eq:FFSINEMRT}
    \frac{P_f\Big(\sum\nolimits_{s \in \SkfMRT}\! {\Txisk} \Tcks \Big)^2}{ \tMRT }\geq \Tilde{F}_k^{\MRT}(\BXiN, \BXiF), \forall k\in \Kf,
\end{align}    
where 
\begin{align} \label{eq:reformulated tF_k}
    &\Tilde{F}_k^{\MRT}(\BXiN, \BXiF) \triangleq P_f \!\sum\nolimits_{s \in \SkfMRT} \! \!\Txisk^2 \Taks\nonumber\\
    &\hspace{2em} + P_f\!\sum\nolimits_{j\in\Kf\setminus k}  \sum\nolimits_{s \in \SjfMRT} \!\sum\nolimits_{s' \in \SjfMRT} \!\! \Txisj\Txispj \Tbkjssp \nonumber\\
    &\hspace{2em} +P_n\! \sum\nolimits_{i\in\Kn} \!\! \sum\nolimits_{s\in \SniMRT} \!\!\sum\nolimits_{s' \in \SniMRT}\!\! \bxisi\bxispi \balphakissp \!+\! \sigma_k^2.
\end{align}
The LHS of \eqref{eq:FFSINEMRT} is handled in a similar way to that of \eqref{eq:simplified near-field SINR constraint}. For the RHS of \eqref{eq:FFSINEMRT}, the upper bound of $\Tilde{F}_k^{\MRT}(\BXiN, \BXiF)$ is given as $\Tilde{F}_k^{\MRT,ub}(\BXiN, \BXiF)$ in \eqref{eq:tFMRT},
\begin{figure*}
\normalsize
\vspace*{5pt} 
\begin{align}
    &\Tilde{F}_k^{\MRT}(\BXiN, \BXiF) \leq 
     \Tilde{F}_k^{\MRT,ub}(\BXiN, \BXiF) \triangleq
    \frac{\rho_f}{4} \!\sum\nolimits_{j\in\Kf}  \!\! \sum\nolimits_{s \in \SjfMRT}  \!\!\sum\nolimits_{s' \in \SjfMRT} \! \Big[({\Txisj\!+\!\Txispj})^2
    \!-\!2(\Txisj^{(n)}\!-\!\Txispj^{(n)})(\Txisj-\Txispj) \!+\! (\Txisj^{(n)}\!-\!\Txispj^{(n)})^2 \Big]\!\Tdkjssp\nonumber\\
    &\hspace{6em}+ 
    \frac{\rho_n}{4}
    \sum\nolimits_{i\in\Kn}  \!\! 
    \sum\nolimits_{s\in \SniMRT}  \!\!
    \sum\nolimits_{s' \in \SniMRT}  \!\!
    \Big[(\bxisi+\bxispi)^2
    -
    2\Big(\bxisi^{(n)}-\bxispi^{(n)}\Big)(\bxisi-\bxispi)+\Big(\bxisi^{(n)}-\bxispi^{(n)}\Big)^2\Big]\balphakissp + 1,\label{eq:tFMRT}   
\end{align}
\hrulefill
\vspace{-1em}
\end{figure*}
where
\begin{align}
    \Tdkjssp = 
    \begin{cases}
        \Taks & \text{$k = j$, $s=s'$}, \\
        0 & \text{$k = j$, $s\neq s'$}, \\
        \Tbkjssp & \text{$k \neq j$}.
    \end{cases}
\end{align}

As a result, (\text{P}3.2) is transformed into the following optimization problem:
\begin{subequations}  \label{eq:final optimization problem for MRT}
\begin{alignat}{2}
    (\text{P}3.3): \hspace{1em}&\underset{\qq^\MRT}{\max}
     &~& w_n \bar{t}^{\MRT} + w_f \Tilde{t}^{\MRT}  \\
    &\hspace{0.5em}\text{s.t.} 
    && v_k^{(n)}{\Big(   {2\sum\nolimits_{s \in \SknMRT}\!\!\sqrt{\rho_n} {\bxisk} } \betakks - v_k^{(n)}\bMRT}\Big) \nonumber \\
    &&& \hspace{2em} \geq \Bar{F}_k^{\MRT,ub}(\BXiN, \BXiF),\quad\forall k\in \Kn, ~\label{eq:P42:vkN}\\    
    &&& u_k^{(n)}\!\Big({2\!\sum\nolimits_{s \in \SkfMRT} \!\!\sqrt{\rho_f}\Txisk\Tcks}\!-u_k^{(n)} \!\tMRT\Big) \nonumber \\
    &&& \hspace{2em}\geq \Tilde{F}_k^{\MRT,ub}(\BXiN, \BXiF),\quad \forall k\in \Kf,~\label{eq:P42:ukF}\\
    &&&
    \eqref{eq:t constraints},
    \eqref{eq:Near-field T constraint}, \eqref{eq:Far-field T constraint}, \eqref{eq:power:MRT},
    \eqref{eq:Power allocation coefficients constraints},
\end{alignat}    
\end{subequations}
where 
\begin{align*}
    & v_k^{(n)} =  \frac{\sqrt{\rho_n} \sum\nolimits_{s \in \SknMRT} {\bxisk^{(n)}} \betakks}{ (\bMRT)^{(n)} }, \\
    & u_k^{(n)} = \frac{ \sqrt{\rho_f}\sum\nolimits_{s \in \SkfMRT} \Txisk^{(n)} \Tcks}{ (\tMRT)^{(n)}}.
\end{align*}
In addition, if $\Tdkjssp$ is negative, a first-order Taylor expansion should be applied to $({\Txisj\!+\!\Txispj})^2$ in \eqref{eq:tFMRT}.
A similar procedure can also be applied to the other terms in \eqref{eq:tFMRT} and \eqref{eq:bFMRT}, where the details are omitted for brevity.
Note that \eqref{eq:final optimization problem for MRT} is a convex optimization problem. Modern convex solvers, such as MOSEK, are capable of efficiently solving \eqref{eq:final optimization problem for MRT} for a relatively large problem size, which is sufficient to characterize the performance of XL-MIMO systems. The SCA-based algorithm proposed to solve the problem \eqref{eq:linear objective function:MRT} can be summarized in \textbf{Algorithm 3}. 

\subsubsection{Power Control Design with CZF}
Recall that to eliminate interference with the CZF design, we set the power control coefficients $\betakCZF\triangleq\bar{\eta}_{1k}^{\CZF}=\ldots=\bar{\eta}_{Sk}^{\CZF}$ ($\TetakCZF\triangleq\Tilde{\eta}_{1k}^{\CZF}=\ldots=\Tilde{\eta}_{Sk}^{\CZF}$) for any NFUE (FFUE) $k$. Then, we have the following optimization problem:
\begin{subequations} 
  \begin{alignat}{2}
    (\text{P}4):\hspace{1em} &\underset{\qq^{\CZF}}{\max} &~&   w_n \bar{t}^{\CZF} + w_f \Tilde{t}^{\CZF}  \\
    &\hspace{0.5em} \text{s.t.}  
    && P_n \sum\nolimits_{k \in \Kn} \betakCZF \sum\nolimits_{s\in \Set} \bar{c}_{sk} \nonumber \label{eq:power control:CZF}\\ 
    &&& +P_f \sum\nolimits_{k \in \Kf} \TetakCZF \sum\nolimits_{s\in \Set} \Tilde{c}_{sk}  \leq P, \\
    &&& \bSINRk^{\CZF}(\BEtaN^{\CZF},\BEtaF^{\CZF}) \geq \bCZF, \forall k \in \Kn \label{eq:Near-field SINR constraint, ZF}, \\
    &&& \tSINRk^{\CZF}(\BEtaN^{\CZF},\BEtaF^{\CZF}) \geq \tCZF, \forall k \in \Kf \label{eq:Far-field SE constraint, ZF}, \\
    &&& \betakCZF \geq 0, \TetakCZF \geq 0, \forall k, \label{eq:eta:CZF} \\
    &&& \eqref{eq:t constraints},
    \eqref{eq:Near-field T constraint}, \eqref{eq:Far-field T constraint},
  \end{alignat}
\end{subequations}
where $\qq^{\CZF}\!\triangleq\!\{\BEtaN^{\CZF},\BEtaF^{\CZF},\bar{t}^{\CZF}, \Tilde{t}^{\CZF},\bCZF, \tCZF\}$. Since $\bCZF > 0$, we can rewrite \eqref{eq:Near-field SINR constraint, ZF}, for $\forall k\in\Kn$ as
\begin{align} \label{eq:reformulated:SINR:CZF}
    P_n\betakCZF \Big(\!\sum\nolimits_{s\in \SknZF}\bar{\qg}_{sk}^{H}\bar{\qw}_{sk}^{\CZF}\Big)^{\!2} \!\geq \!\bar{F}_k^{\CZF}(\BEtaN^{\CZF}, \BEtaF^{\CZF}) \bCZF,      
\end{align}
where 
\begin{align*}
    &\bar{F}_k^{\CZF}(\BEtaN^{\CZF}, \BEtaF^{\CZF}) \!=\!{P_n}\sum\nolimits_{i \in\Kn\setminus k}\!\!\betaiCZF\Big\vert \sum\nolimits_{s\in\SinZF} \!  \bar{\qg}_{sk}^{H} \bar{\qw}_{si}^{\CZF} \Big\vert^2 \nonumber \\
    &\hspace{2em} + P_f \sum\nolimits_{j\in\Kf} \TetajCZF r_{kj} + \sigma_k^2.  
\end{align*}
\begin{algorithm}[t]
\caption{Proposed SCA-based Algorithm for \eqref{eq:linear objective function:MRT}} \label{alg:SCA} 
\begin{algorithmic}[1]
    \State Initialization: set $n=1$, and generate initial points $(\BXiN^{(1)}, \BXiF^{(1)}, (\bMRT)^{(1)}, (\tMRT)^{(1)})$. Define a tolerance $\epsilon_1$ and the maximum number of iterations $I_3$.
    \State Iteration $n$: solve \eqref{eq:final optimization problem for MRT}. Let $(\BXiN^{\ast}, \BXiF^{\ast}, (\bar{t}^{\MRT})^{\ast}, (\Tilde{t}^{\MRT})^{\ast}, (\bMRT)^{\ast}, (\tMRT)^{\ast})$ be the solution.
    \State If the fractional increase of the objective function value is below a threshold  $\epsilon_1 > 0$ or $n=I_3$, stop. Otherwise, go to step 4.
    \State Set $n = n + 1$, update $(\BXiN^{(n)}, \BXiF^{(n)}, (\bMRT)^{(n)}, (\tMRT)^{(n)}) = (\BXiN^{\ast}, \BXiF^{\ast}, (\bMRT)^{\ast}, (\tMRT)^{\ast})$, go to step 2.
\end{algorithmic} 
\end{algorithm}
\setlength{\textfloatsep}{0.25cm}
We note that the LHS of \eqref{eq:reformulated:SINR:CZF} is linear with respect to $\BEtaN^{\CZF}$, while the RHS is not convex, e.g., the existence of the term $\betaiCZF \bCZF$. Using SCA, we derive the convex upper bound of $\bar{F}_k^{\CZF}(\BEtaN^{\CZF}, \BEtaF^{\CZF}) \bCZF$ as $\bar{U}_k^{\CZF,ub}(\BEtaN^{\CZF}, \BEtaF^{\CZF},\bCZF)$ given in~\eqref{eq:bF_k:CZF}.
\begin{figure*}[!b]
\normalsize
\hrulefill
\vspace*{5pt}
\begin{align}
   \bar{F}_k^{\CZF}(\BEtaN^{\CZF}, \BEtaF^{\CZF}) \bCZF  &\!\leq\! 
    \bar{U}_k^{\CZF,ub}(\BEtaN^{\CZF}, \BEtaF^{\CZF},\bCZF) \!\triangleq\!
    \frac{\rho_n}{4}\sum\nolimits_{i \in\Kn\setminus k}\Big[(\betaiCZF \!+\! \bCZF)^2
    \!-\!2\Big((\betaiCZF)^{(n)}\!-\!(\bCZF)^{(n)}\Big)(\betaiCZF\!-\!\bCZF) \nonumber \\
    &\hspace{-3.5em}+ \Big((\betaiCZF)^{(n)}\!-\!(\bCZF)^{(n)}\Big){^2} \Big] 
    \Big\vert\sum\nolimits_{s\in\SinZF}  \bar{\qg}_{sk}^{H} \bar{\qw}_{si}^{\CZF} \Big\vert^{2}\!+\! \frac{\rho_f}{4} \sum\nolimits_{j\in\Kf} \!\! \Big[(\TetajCZF \!+\! \bCZF)^2
    \!-\!2\Big((\TetajCZF)^{(n)} \!-\! (\bCZF)^{(n)}\Big) \nonumber \\
    &\hspace{-3.5em}\times(\TetajCZF\!-\!\bCZF)\!+ \!\Big((\TetajCZF)^{(n)}\!\!-\!\!(\bCZF)^{(n)}\Big){^2} \Big] r_{kj}\! +\! \bCZF,\label{eq:bF_k:CZF} 
\end{align}
\vspace{-2.5em}
\end{figure*}
Similarly, the FF SINR constraint \eqref{eq:Far-field SE constraint, ZF} can be expressed as
\begin{align} \label{eq:tF_constraint:CZF}
    P_f\TetakCZF \Big(\!\sum\nolimits_{s\in \SkfZF}\!\!\Ex\{\Tilde{\qg}_{sk}^H \Tilde{\qw}_{sk}^{\CZF} \}\!\Big)^{\!2} \!\!\geq\! \!\Tilde{F}_k^{\CZF}(\BEtaN^{\CZF}\!,\! \BEtaF^{\CZF}) \tCZF,
\end{align}
where 
\begin{align} \label{eq:tF:CZF}
    &\Tilde{F}_k^{\CZF}(\BEtaN^{\CZF}, \BEtaF^{\CZF}) \!=\!  P_f \!\sum\nolimits_{j\in\Kf\setminus k} \!\!\TetajCZF \Ex\Big\{\! \Big\vert\sum\nolimits_{s\in \SjfZF} \!\Tilde{\qg}_{sk}^H \Tilde{\qw}_{sj}^{\CZF} \Big\vert^2 \!\Big\} \nonumber \\
    &\hspace{2em} + P_f \TetakCZF \varepsilon_k 
    + P_n \sum\nolimits_{i\in\Kn} \!\!\betaiCZF t_{ki} + \sigma_k^2.
\end{align}

The RHS of \eqref{eq:tF_constraint:CZF} is non-convex due to the presence of multiplicative terms involving optimization parameters. To address this, we use SCA and derive a convex upper-bound for the RHS of \eqref{eq:tF_constraint:CZF} as $\Tilde{U}_k^{\CZF,ub}(\BEtaN^{\CZF}, \BEtaF^{\CZF}, \tCZF)$, shown in \eqref{eq:tF_k:CZF}, where $\Tilde{r}_{kj}$ is given by
\begin{figure*}[!b]
\normalsize
\vspace*{5pt}
\begin{align}
    \Tilde{F}_k^{\CZF}(\BEtaN^{\CZF}, \BEtaF^{\CZF}) \tCZF &\!\leq\!
    \Tilde{U}_k^{\CZF,ub}(\BEtaN^{\CZF}, \BEtaF^{\CZF}, \tCZF)\!\triangleq\! \frac{\rho_f}{4} \!\sum\nolimits_{j\in\Kf}\!\!\Big[(\TetajCZF \!+\! \tCZF)^{\!2}
    \!-\!2\Big((\TetajCZF)^{(n)}\!-\!(\!\tCZF\!)^{\!(n)}\Big)(\TetajCZF\!-\!\tCZF) \nonumber \\
    &\hspace{-3.5em}\!+\! \Big((\TetajCZF)^{(n)}\!-\!(\tCZF)^{\!(n)}\Big)^2 \Big] \Tilde{r}_{kj} 
     \!+\! \frac{\rho_n}{4} \sum\nolimits_{i\in\Kn} \Big[(\betaiCZF \!+\! \tCZF)^2
    -2\Big((\betaiCZF)^{(n)}-(\tCZF)^{(n)}\Big)(\betaiCZF\!-\!\bCZF)  \nonumber \\
    &\hspace{-3.5em}+ \Big((\betaiCZF)^{(n)}-(\tCZF)^{(n)}\Big)^2 \Big] t_{ki} \!+\! \tCZF  \label{eq:tF_k:CZF}, 
\end{align}
\end{figure*}
\vspace{-0.1em}
\begin{align}
    \Tilde{r}_{kj} = 
    \begin{cases}
        \varepsilon_k & \text{$k = j$}, \\
        \Ex\Big\{ \Big\vert\sum\nolimits_{s\in \SjfZF} \Tilde{\qg}_{sk}^H \Tilde{\qw}_{sj}^{\CZF} \Big\vert^2 \Big\} & \text{$k \neq j$}. 
    \end{cases}
\end{align}

As a result, the optimization problem (\text{P}4) is recast as the following convex optimization problem:  
\begin{subequations} \label{eq:updated:CVX:CZF} 
  \begin{alignat} {2}
    (\text{P}4.1) : \hspace{0.1em} &\underset{\qq^{\CZF}}{\max} &\hspace{0.7em}&  w_n \bar{t}^{\CZF} + w_f \Tilde{t}^{\CZF}  \\
    &\hspace{0.5em}\text{s.t.} && \rho_n \betakCZF \big(\sum\nolimits_{s\in \SknZF}\bar{\qg}_{sk}^{H}\bar{\qw}_{sk}^{\CZF}\big)^2 \nonumber \\
    &&&\hspace{0em}\geq \bar{U}_k^{\CZF,ub}(\BEtaN^{\CZF}, \BEtaF^{\CZF},\bCZF), \forall k \in \Kn,  \\
    &&& \rho_f\TetakCZF \big(\sum\nolimits_{s\in \SkfZF}\Ex\{\Tilde{\qg}_{sk}^H \Tilde{\qw}_{sk}^{\CZF} \}\big)^2 \nonumber \\
    &&&\hspace{0em}\geq \Tilde{U}_k^{\CZF,ub}(\BEtaN^{\CZF}, \BEtaF^{\CZF}, \tCZF), \forall k\in\Kf, \\
    &&&
    \eqref{eq:t constraints},
    \eqref{eq:Near-field T constraint},   \eqref{eq:Far-field T constraint},    
    \eqref{eq:power control:CZF},
    \eqref{eq:eta:CZF}.
\end{alignat}
\end{subequations}

The problem $(\text{P}4.1)$ can be effectively solved by modern convex solvers such as CVX \cite{cvx}.  We can adapt \textbf{Algorithm 3} with the appropriate modifications to solve \eqref{eq:updated:CVX:CZF}. The details are omitted for the sake of brevity.
\subsubsection{Power Control Design with LZF}
With LZF, the optimization problem $(\text{P}3)$ can be expressed as  
 \vspace{-0.4em}
\begin{subequations} ~\label{eq:opt:LZF}
  \begin{alignat}{2}
    (\text{P}5): \hspace{1em}&\underset{\qq^{\LZF}}{\max} &\hspace{1em}& w_n \bar{t}^{\LZF} + w_f \Tilde{t}^{\LZF}  \\
    &\hspace{0.5em}\text{s.t.} 
    &&\sum\nolimits_{s\in\Set} \Big(P_n \sum\nolimits_{k \in \US} \bpsi_{sk}^2 \bar{c}_{sk}^{\LZF}\nonumber\\
    &&&+ P_f \sum\nolimits_{k \in \IS} \tpsi_{sk}^2 \Tilde{c}_{sk}^{\LZF}\Big) \leq P,\label{eq:power:LZF}\\
    &&&\bSINRk^{\LZF}(\Bbpsi,\Btpsi) \geq \bLZF, \forall k \in \Kn \label{eq:Near-field SINR constraint, LZF}, \\
    &&& \tSINRk^{\LZF}(\Bbpsi,\Btpsi) \geq \tLZF, \forall k \in \Kf \label{eq:Far-field SINR constraint, LZF}, \\
    &&& \bpsi_{sk} \geq 0, \tpsi_{sk} \geq 0, \forall s, \forall k, \label{eq:xi constraints:LZF}\\
    &&& 
    \eqref{eq:t constraints},
    \eqref{eq:Near-field T constraint},   \eqref{eq:Far-field T constraint}.
\end{alignat}  
\end{subequations}
where $\qq^{\LZF}\!\triangleq\!\{\Bbpsi,\Btpsi,\bar{t}^{\LZF}, \Tilde{t}^{\LZF},\bLZF, \tLZF\}$, and $\Bbpsi \!\triangleq\! [\Bbpsi_1,\ldots,\Bbpsi_{K_n}]$, where $\Bbpsi_k\!\triangleq\! [\bpsi_{1k},\ldots,\bpsi_{Sk}]^T$ with $\bpsi_{sk} \!\triangleq\! \sqrt{\betask^{\LZF}}, \forall s, k\!\in\!\Kn$; $\Btpsi\!\triangleq\! [\Btpsi_1,\ldots,\Btpsi_{K_f}]$, where $\Btpsi_k \!\triangleq\! [\tpsi_{1k},\ldots,\tpsi_{Sk}]^T$ with $\tpsi_{sk} \!\triangleq\! \sqrt{\Tetask^{\LZF}}, \forall s, k\!\in\!\Kf$.
 Constraints~\eqref{eq:Near-field SINR constraint, LZF} and~\eqref{eq:Far-field SINR constraint, LZF}  can be recast as
 \vspace{-0.4em}
 \begin{subequations}
     \begin{align}
          \frac{P_n \Big(\sum_{s\in \SknLZF} \bpsisk \Big)^2 } {\bLZF} 
          &\ge \bar{F}_k^{\LZF}(\Bbpsi,\Btpsi),\forall k \in \Kn,~\label{eq:SINR:LZF}\\
          \frac{P_f \Big(\sum\nolimits_{s\in \SkfLZF}\tpsisk\Big)^2 } {\tLZF}
          &\ge\tilde{F}_k^{\LZF}(\Bbpsi,\Btpsi),\forall k \in \Kf,~\label{eq:SINR:LZF2}
     \end{align}
 \end{subequations}
where $\bar{F}_k^{\LZF}(\Bbpsi,\Btpsi)$ and $\tilde{F}_k^{\LZF}(\Bbpsi,\Btpsi)$ are given by
 \vspace{-0.4em}
 \begin{subequations}
\begin{align}
    &\bar{F}_k^{\LZF}(\Bbpsi,\Btpsi) \triangleq 
    P_n\sum\nolimits_{i \in\Kn\setminus k}\Big\vert\sum\nolimits_{s\in\SinLZF} \!\!\bpsisi \bar{\qg}_{sk}^{H} \bar{\qw}_{si}^{\LZF} \Big\vert^2 \nonumber \\
    &\hspace{1em}+ P_f \sum\nolimits_{j\in\Kf}\!\!
    \sum\nolimits_{s \in \SjfLZF} \!\! \sum\nolimits_{s \in \SjfLZF} \!\!  \tpsisj\tpsispj\Ttaukjsspi + \sigma_k^2, \\
    &\Tilde{F}_k^{\LZF}(\Bbpsi,\Btpsi) \triangleq P_f \sum\nolimits_{j\in\Kf\setminus k}\!\! \Ex\Big\{ \Big\vert\sum\nolimits_{s\in \SjfLZF}\!\!\tpsisj\Tilde{\qg}_{sk}^H \Tilde{\qw}_{sj}^{\LZF} \Big\vert^2 \Big\} \nonumber \\
    &\hspace{1em}+ P_n \sum\nolimits_{i\in\Kn}\!\!\sum\nolimits_{s \in \SniLZF} \!\!
    \sum\nolimits_{s' \in \SniLZF} \!\!
    \bpsisi\bpsispi\btaukisspi + \sigma_k^2.   
\end{align}
 \end{subequations}
The two constraints~\eqref{eq:SINR:LZF} and~\eqref{eq:SINR:LZF2} are non-convex. To address this issue, we apply SCA and replace $\bar{F}_k^{\LZF}(\Bbpsi,\Btpsi)$ and $\tilde{F}_k^{\LZF}(\Bbpsi,\Btpsi)$ by their convex upper bounds $\bar{F}_k^{\LZF,ub}(\Bbpsi,\Btpsi)$, $\Tilde{F}_k^{\LZF,ub}(\Bbpsi,\Btpsi)$ given in \eqref{eq:bF_k:LZF} and \eqref{eq:tF_k:LZF}, respectively, at the top of the next page.
\begin{figure*}
\normalsize
\begin{subequations} 
\vspace*{5pt} 
\begin{align}
    &\bar{F}_k^{\LZF}(\Bbpsi,\Btpsi) \leq \!\bar{F}_k^{\LZF,ub}(\Bbpsi,\Btpsi) \triangleq 
    \frac{\rho_n}{4}\sum\nolimits_{i \in\Kn\setminus k}\bigg(  \sum\nolimits_{s \in \SniLZF}
    4{\bpsisi^2}  \vert\bar{\qg}_{sk}^{H} \bar{\qw}_{si}^{\LZF} \vert^2
    + 2\sum\nolimits_{s \in \SniLZF}
    \sum\nolimits_{s' \in\SniLZF,s'>s}\!\!\Big[(\bpsisi+\bpsispi)^2 \nonumber \\
    &\hspace{2em}-2\Big(\bpsisi^{(n)}-\bpsispi^{(n)}\Big)(\bpsisi-\bpsispi) + \Big(\bpsisi^{(n)}-\bpsispi^{(n)}\Big)^2\Big] \Re\{\bar{\qg}_{sk}^{H} \bar{\qw}_{si}^{\LZF}(\bar{\qg}_{s'k}^{T} \bar{\qw}_{s'i}^{\LZF})^{\ast}\} \bigg) + \frac{\rho_f}{4}
    \nonumber \\
    &\hspace{2em}    
    \times \sum\nolimits_{j\in\Kf} \!\!
    \sum\nolimits_{s \in \SjfLZF} \!\!
    \sum\nolimits_{s' \in \SjfLZF}  \!\!\Big[ (\tpsisj +\tpsispj)^2 - 2 
    \Big(\tpsisj^{(n)} -\tpsispj^{(n)}\Big)(\tpsisj -\tpsispj)+\Big(\tpsisj^{(n)} -\tpsispj^{(n)}\Big)^2\Big]  \Ttaukjsspi + 1,  \label{eq:bF_k:LZF} \\
    &\Tilde{F}_k^{\LZF}(\Bbpsi,\Btpsi) \leq \!\Tilde{F}_k^{\LZF,ub}(\Bbpsi,\Btpsi) \triangleq \frac{\rho_f}{4} \sum\nolimits_{j\in\Kf\setminus k}\!\!
    \bigg(  \sum\nolimits_{s \in \SjfLZF}
    4{\tpsisj^2} \Ex\{\vert \Tilde{\qg}_{sk}^H \Tilde{\qw}_{sj}^{\LZF} \vert^2\}
    + 2\sum\nolimits_{s \in \SjfLZF}\!\!
    \sum\nolimits_{s' \in\SjfLZF,s'>s}\!\!\Big[ (\tpsisj \!+\!\tpsispj)^2 \nonumber \\
    &\hspace{2em}- 2 
    \Big(\tpsisj^{(n)} -\tpsispj^{(n)}\Big)(\tpsisj -\tpsispj)+\Big(\tpsisj^{(n)} -\tpsispj^{(n)}\Big)^2\Big] \Re\big\{\Ex\{\Tilde{\qg}_{sk}^H \Tilde{\qw}_{sj}^{\LZF}(\Tilde{\qg}_{s'k}^H \Tilde{\qw}_{s'j}^{\LZF})^{\ast}\}\big\} \bigg) + \frac{\rho_n}{4} 
    \nonumber \\
     &\hspace{2em} \times\sum\nolimits_{i\in\Kn}\!\!\sum\nolimits_{s \in \SniLZF}\!\! \sum\nolimits_{s' \in \SniLZF}\!\! \Big[(\bpsisi+\bpsispi)^2
    -
    2\Big(\bpsisi^{(n)}-\bpsispi^{(n)}\Big)(\bpsisi-\bpsispi)+\Big(\bpsisi^{(n)}-\bpsispi^{(n)}\Big)^2\Big]
    \btaukisspi + 1, \label{eq:tF_k:LZF}
\end{align}
\end{subequations}
\hrulefill
\end{figure*}
As a result, the final convex optimization problem for LZF is given by
\begin{subequations} \label{eq:final convex expression:LZF}
  \begin{alignat}{2}
    (\text{P}5.1):\hspace{0.5em} &\underset{\qq^{\LZF}}{\max} &\hspace{1em}& w_n \bar{t}^{\LZF} + w_f \Tilde{t}^{\LZF}  \\
    &\hspace{0.5em}\text{s.t.} 
    &&(u_k^{\LZF})^{(n)}\Big(2\!\sum\nolimits_{s \in \SkfLZF}\!\! \sqrt{\rho_f}\tpsisk \!-\!(u_k^{\LZF})^{(n)}  \nonumber \\
    &&& \times\tLZF\Big)  \geq \Tilde{F}_k^{\LZF,ub}(\Bbpsi, \Btpsi), \forall k\in \Kf, \\
    &&&(v_k^{\LZF})^{(n)}{\Big( 2\sum\nolimits_{s \in \SknLZF}\!\!\sqrt{\rho_n} \bpsisk - (v_k^{\LZF})^{(n)}\bLZF}\Big) \nonumber \\
    &&&\geq \bar{F}_k^{\LZF,ub}(\Bbpsi, \Btpsi),\forall k\in \Kn, \\     
    &&& 
    \eqref{eq:t constraints},
    \eqref{eq:Near-field T constraint}, \eqref{eq:Far-field T constraint}, \eqref{eq:power:LZF},
    \eqref{eq:xi constraints:LZF},
\end{alignat}  
\end{subequations}
where
\begin{subequations} 
  \begin{align*} 
    & (u_k^{\LZF})^{(n)} = \frac{ \sqrt{\rho_f}\sum\nolimits_{s \in \SkfLZF} (\Txisk^{\LZF})^{(n)} }{ (\tLZF)^{(n)}},\\
    & (v_k^{\LZF})^{(n)} = \frac{\sqrt{\rho_n} \sum\nolimits_{s \in \SknLZF} (\bxisk^{\LZF})^{(n)} }{ (\bLZF)^{(n)} }.
\end{align*} 
\end{subequations}
The problem $(\text{P}5.1)$ can be effectively solved by CVX \cite{cvx}. We can adapt \textbf{Algorithm 3} with appropriate modifications to solve \eqref{eq:final convex expression:LZF}. The details are omitted for brevity.
\subsection{Complexity and Convergence Analysis} 
\textbf{Computational complexity and convergence analysis of \textbf{Algorithm 2}: } Since the main complexity comes from solving $(\text{P}2.3)$ in the inner loop, the complexity of solving this problem is $\mathcal{O}(N^{3.5})$ if the solver is based on the interior point method \cite{Luo:MSP:2010}. Then, the overall complexity of \textbf{Algorithm 2} is $\mathcal{O}(I_{out}I_{in}N^{3.5})$, where $I_{out}$, $I_{in}$ represent the number of outer and inner iterations for convergence, respectively. 

We now study the convergence of \textbf{Algorithm 2}, where we adopt SCA to solve problem $(\text{P}2)$ to deal with the non-concave penalty term $\varrho \| \qV\|_2$ in the objective function. Let $f_{org}(t, \qV) = t - \varrho(\| \qV\|_{\ast} - \| \qV\|_2)$ denote the original objective. In iteration $(n)$, we linearize the convex term $\| \qV\|_2$ as $ \bar{\qV}^{(n)}$, and obtain the lower bound of the original objective as $f_{lb}^{(n)}(t, \qV) = t - \varrho(\| \qV\|_{\ast} - \bar{\qV}^{(n)} )$. Since $f_{lb}^{(n)}(t, \qV)$ is a global lower bound, we have 
\begin{align} \label{eq:lower bound}
    f_{org}(t, \qV) \geq f_{lb}^{(n)}(t, \qV), \forall (t, \qV), 
\end{align}
and 
\begin{align} \label{eq:lower bound equality}
    f_{org}\big(t^{(n)},\qV^{(n)}\big) = f_{lb}^{(n)}\big(t^{(n)},\qV^{(n)}\big), 
\end{align}
holds at the current point $(t^{(n)}, \qV^{(n)})$ since $f_{lb}^{(n)}(t, \qV)$ is built based on point $(t^{(n)}, \qV^{(n)})$. 

\noindent In addition, we maximize $f_{lb}^{(n)}(t, \qV)$, which implies that 
\begin{align} \label{eq:optimization}
    f_{lb}^{(n)}\big(t^{(n+1)},\qV^{(n+1)}\big) \geq f_{lb}^{(n)}\big(t^{(n)},\qV^{(n)}\big),     
\end{align}
Combing \eqref{eq:lower bound}, \eqref{eq:lower bound equality} and \eqref{eq:optimization}, we obtain the following inequality chain~\cite{Tam:TWC:2017}:
    \begin{align}\label{eq:inequality chain}
        f_{org}\big(t^{(n+1)},\qV^{(n+1)}\big) &\geq f_{lb}^{(n)}\big(t^{(n+1)},\qV^{(n+1)}\big)  \nonumber \\ 
        & \geq f_{lb}^{(n)}\big(t^{(n)},\qV^{(n)}\big)  \nonumber \\ 
        & = f_{org}\big(t^{(n)},\qV^{(n)}\big).
    \end{align}
The first inequality holds because $f_{lb}^{(n)}(t,\qV)$ is a global lower bound on $f_{org}\big(t,\qV\big)$ for all $\big(t,\qV\big)$, by construction of the linear approximation of the convex term $\|\qV\|_2$. The second inequality follows from the fact that we maximize $f_{lb}^{(n)}\big(t,\qV\big)$ in each iteration, and the final equality holds since the approximation is exact at iteration $(n)$. Since $t \leq \tr(\qR_k \qV)$ is bounded above, so the objective of \textbf{Algorithm 2} is non-decreasing and guarantees to converge to a limit.

\begin{table}[t]
	\centering
	\caption{\label{tabel:complexity} Summary of optimization problems' parameters for computational complexity analysis }	
    \vspace{-0.5em}
	\small
\begin{tabular}{|p{1.6cm}|p{2.8cm}|p{2.9cm}|}
	\hline
        \centering\textbf{Problem} 
        &\centering {$A_v^i$}
        &\centering  $A_l^i$
        \cr            
        \hline

        \hspace{0em}\textbf{P3.3} (MRT)    
        & \centering  \small{$|\bar{\Set}^{\MRT}| \! K_n \!+ \!|\Tilde{\Set}^{\MRT}| \!K_f \!+\! 4$ }\normalsize
        & \centering  \small{$|\bar{\Set}^{\MRT}| K_n \!+ \!|\Tilde{\Set}^{\MRT}| K_f \!+\! 2$ }\normalsize
        \cr        
        \hline

       \hspace{0em}\textbf{P4.1} (CZF)        
        &\centering $K+4$
        &\centering $K\!+\!3$
        \cr      
        \hline
           
        \hspace{0em}\textbf{P5.1} (LZF)       
        & \centering  \small{$|\bar{\Set}^{\LZF}| \!K_n \!+ \!|\Tilde{\Set}^{\LZF}| \!K_f \!+\! 4$  }\normalsize
        & \centering  \small{$|\bar{\Set}^{\LZF}|\! K_n \!+ \!|\Tilde{\Set}^{\LZF}| \!K_f \!+\! 2$ }\normalsize 
        \cr
        \hline   
\end{tabular}
\end{table}

\textbf{Computational complexity and convergence analysis of \textbf{Algorithm 3}:} The relaxed optimization problems (\text{P3.3}), (\text{P4.1}) and (\text{P5.1}) each entail a computational complexity of $\mathcal{O}(\sqrt{A_l^i + A_q^i + 3A_{\text{exp}}^i}(A_v^i + A_l^i + A_q^i + A_{\text{exp}}^i)(A_v^i)^2)$~\cite{Tam:TWC:2017,Mohammadi:JSAC:2023}, where $i\in\{\MRT,\CZF,\LZF\}$. Here, $A_v^i$ represents the number of real-valued scalar variables, while $A_l^i$, $A_q^i$ and $A_{\text{exp}}^i\!=\!2$ denote the number of linear, quadratic and exponential cone constraints, respectively. For tractability, we consider an ideal scenario in which the VR selected by users within each group includes the same number of subarrays, e.g. $|\SkfMRT| \!=\! |\Tilde{\Set}^{\MRT}|$ for FFUEs with MRT. Table~\ref{tabel:complexity} summarizes the values of $A_v^i$ and $A_l^i$ for various optimization problems, while $A_q^i\!=\!K+1$ for MRT and LZF and $A_q^i\!=\!K$ for CZF. 

Let $(t^i)^{(n)} =  w_n (\bar{t} ^{i})^{(n)} + w_f (\Tilde{t} ^{i})^{(n)}$ denote the optimized objective and $(\qq^i)^{(n)}$ be the set of optimized solutions at iteration $(n)$. Due to the linear approximation in \eqref{eq:lower bound of x square divides y} and \eqref{eq:xy:ub}, the updating rule in \textbf{Algorithm 3} ensures that $(\qq^i)^{(n)}$ is feasible for the optimization problem at iteration $(n+1)$. More specifically, at iteration $(n)$, we use $(\qq^i)^{(n)}$ as a feasible point to solve the optimization problem and obtain $(\qq^i)^{(n+1)}$. Since we maximize the sum of weighted SE, the inequality $(t^i)^{(n+1)} \geq (t^i)^{(n)}$ holds, which implies that \textbf{Algorithm 3} yields a non-decreasing sequence. We also note that due to the total power constraint \eqref{eq:power control requirements}, the objective of \eqref{eq:Original power allocation optimization problem} is bounded from above. Thus, \textbf{Algorithm 3} guarantees the convergence of the objective function in $(\text{P}3)$ for different precoding schemes.

\section{Simulation Results} ~\label{Sec:Simulation results}
\vspace{-2em}
\subsection{Simulation Setup}
We assume that the XL-MIMO array is divided into $S$ subarrays, each consisting of $M^{\ast} = 50$ elements. Unless otherwise stated, we set $M_y = 10$ and increase $M_x$ linearly. The RIS is equipped with $N = 10 \times 10$ elements. The distance between the XL-MIMO and the RIS is $d_{\text{MR}} = 100$ m. We consider $K_f=5$ FFUEs located in a semicircle and the radius is $d_{\text{RU}} = 20$ m. There are $K_n=5$ NFUEs randomly distributed in an area with $x\in[0,20]$, $y\in[0,20]$ and $z\in[2,20]$, the locations of NFUEs with respect to XL-MIMO is illustrated in Fig. \ref{fig:location}. We assume that the Ricean factor is $\iota = 2$, while the path-loss coefficients are $\varsigma_k = 10^{-3} d_{\text{RU}}^{-2}$ and $\zeta = 10^{-3} d_{\text{MR}}^{-2.5}$, respectively~\cite{Zhi:2023:TIT}. We set the maximum number of iterations $I_1 = I_2 = 30$ for \textbf{Algorithm 2} and $I_3 = 30$ for \textbf{Algorithm 3}. The thresholds are $\epsilon = 10^{-4}$, $\epsilon_1 = 10^{-3}$ and we set the initial penalty factor $\varrho = 10^{-6}$, the scaling factor $l = 10$ for \textbf{Algorithm 2}. The other parameters are set as follows: the carrier frequency of the system is $5$ GHz, the transmit power is $P=30$ dBm, and the noise power is $-104$ dBm~\cite{Zhi:2023:TIT}.

To evaluate the performance of the proposed design,  namely optimized phase shifts design and power control (\textbf{OPS-OPC}),  we consider a benchmark scheme, random phase shifts and equal power control (\textbf{RPS-EPC}) design, where the transmit power $P$ is equally divided between all users with random phase shift matrix at RIS. We also include the heuristic phase shifts proposed in \cite{Zhi:2022:JSAC} and equal power control (\textbf{HPS-EPC}), as a benchmark for our proposed design. The main idea of the heuristic phase shifts design is to align the phase shifts of the RIS with the phase of the RIS-related channels of a specific FFUE $k$, that is, $\qh_k^H\bTeta\qb_N(\varphi^a,\varphi^e) = N$. To this end, we set the entries of $\bTeta$ as $\theta_n = -\angle\big\{[\qh_k^H]_n [\qb_N(\varphi^a,\varphi^e)]_n\big\}$, $\forall n$ \cite{Zhi:2022:JSAC}. To have a fair comparison between all precoding designs, we set $\bSEkth = \min\{\bSEkCZF, \bSEkLZF, \bSEkMRT\}, \forall k\!\!\in\!\!\Kn, \tSEkth =  \min\{\tSEkCZF, \tSEkLZF, \tSEkMRT\}, \forall k\!\!\in\!\!\Kf$ according to the \textbf{RPS-EPC} design.

\begin{figure}[t]
    \centering
    \includegraphics[width=75mm]{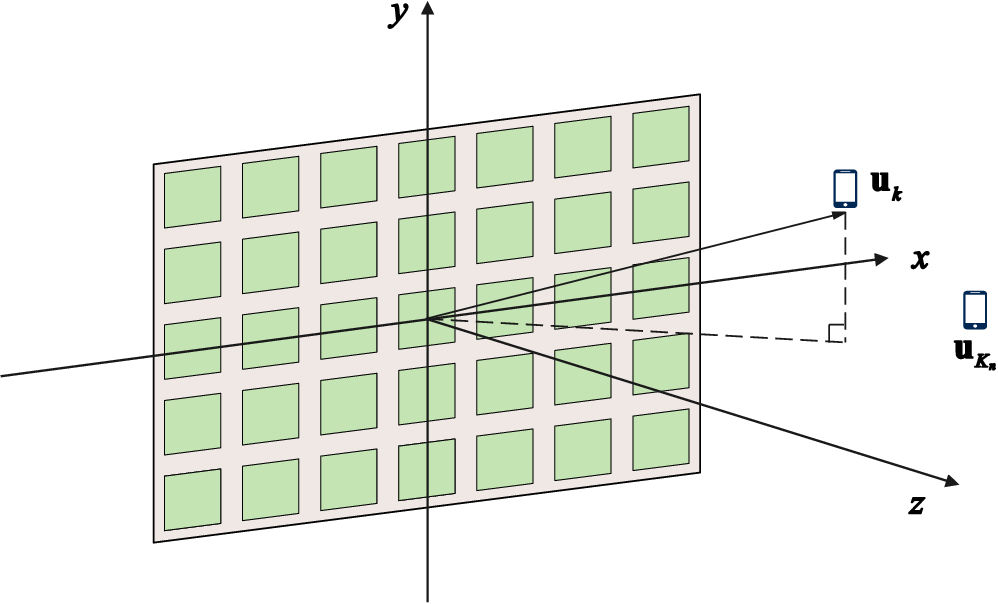}
    \caption{Illustration of the locations of NFUEs with respect to XL-MIMO.}
     \label{fig:location}
\end{figure}

\begin{figure*}[t]
\hspace{-1em}
    \centering
    \begin{subfigure}[t]{0.33\textwidth}
        \centering
        \includegraphics[trim=0 0cm 0cm 0cm,clip,width=\textwidth]{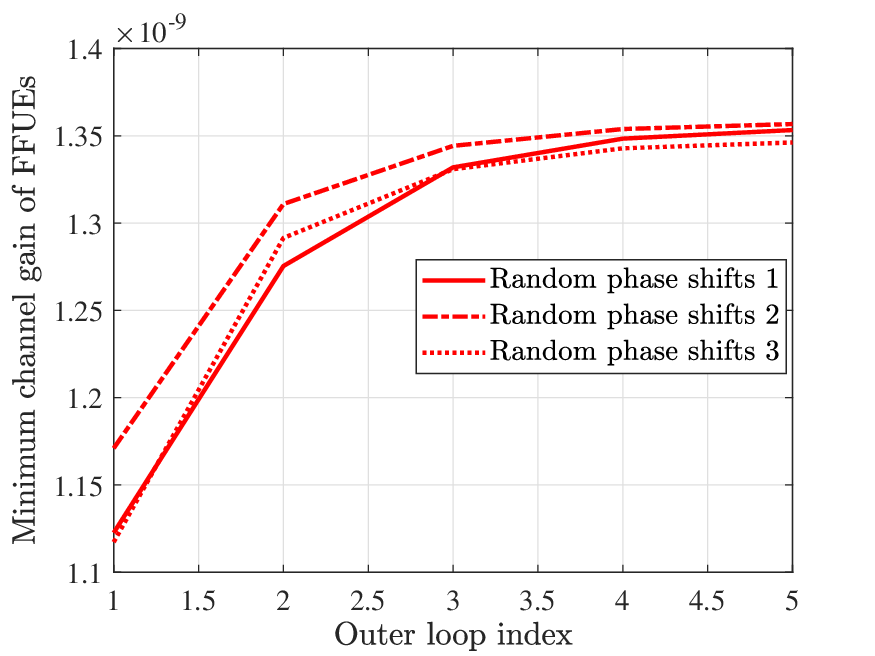}
        \caption{\textbf{Algorithm 2} with MRT.}
        \label{fig:converge:MRT}
    \end{subfigure}
    \hfill
    \begin{subfigure}[t]{0.33\textwidth}
        \centering
        \includegraphics[trim=0 0cm 0cm 0cm,clip,width=\textwidth]{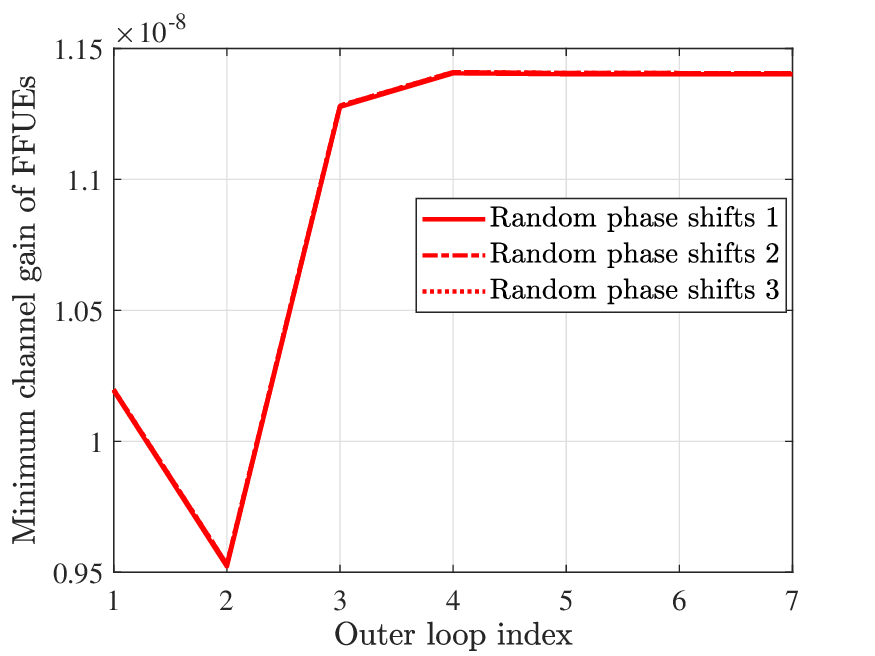}
        \caption{\textbf{Algorithm 2} with CZF.}
        \label{fig:converge:CZF}
    \end{subfigure}
    \hfill
    \begin{subfigure}[t]{0.33\textwidth}
        \centering
        \includegraphics[trim=0 0cm 0cm 0cm,clip,width=\textwidth]{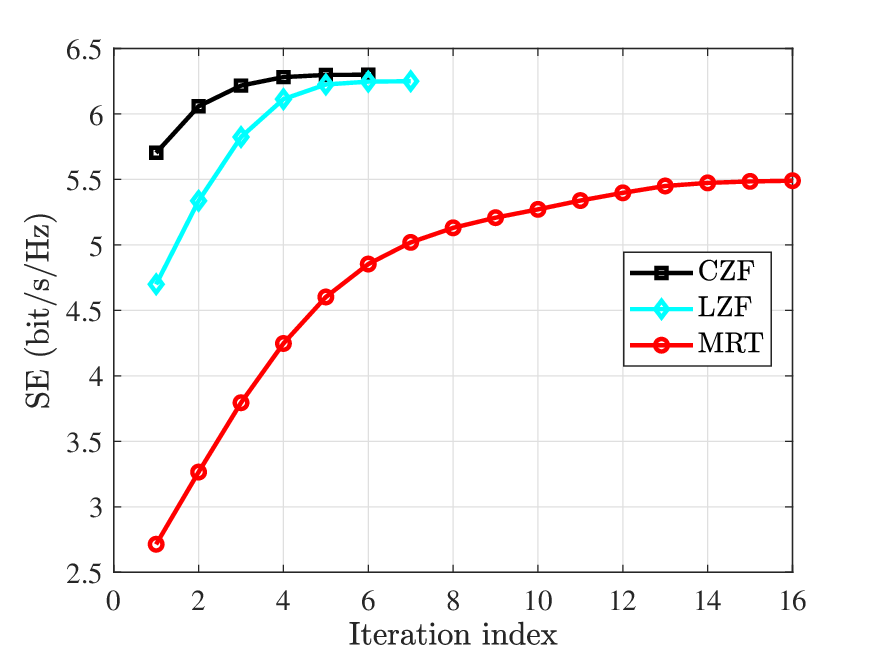}
        \caption{\textbf{Algorithm 3} with different precodings.}
        \label{fig:converge:Power control}
    \end{subfigure}
     \vspace{-0.2em}
    \caption{Convergence behavior of proposed Algorithms.}
    \label{fig:Convergence}
\vspace{-0.2em}
\end{figure*}

\begin{figure*}[t]
\hspace{-1em}
\vspace{-0.5cm}
    \centering
    \begin{minipage}[t]{0.33\textwidth}
        \centering
        \includegraphics[trim=0 0cm 0cm 0cm,clip,width=\textwidth]{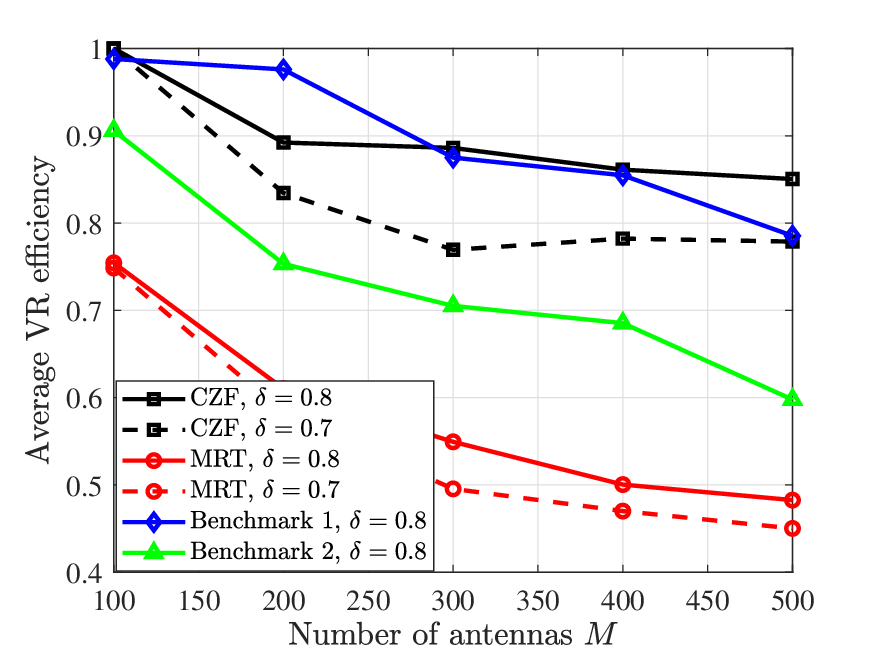}
        \caption{Average VR efficiency versus $M$ ($w_n \!=\!w_f \!=\! 0.5$).}
    \label{fig:VR selection}
    \end{minipage}
    \hfill
    \begin{minipage}[t]{0.33\textwidth}
        \centering
        \includegraphics[trim=0 0cm 0cm 0cm,clip,width=\textwidth]{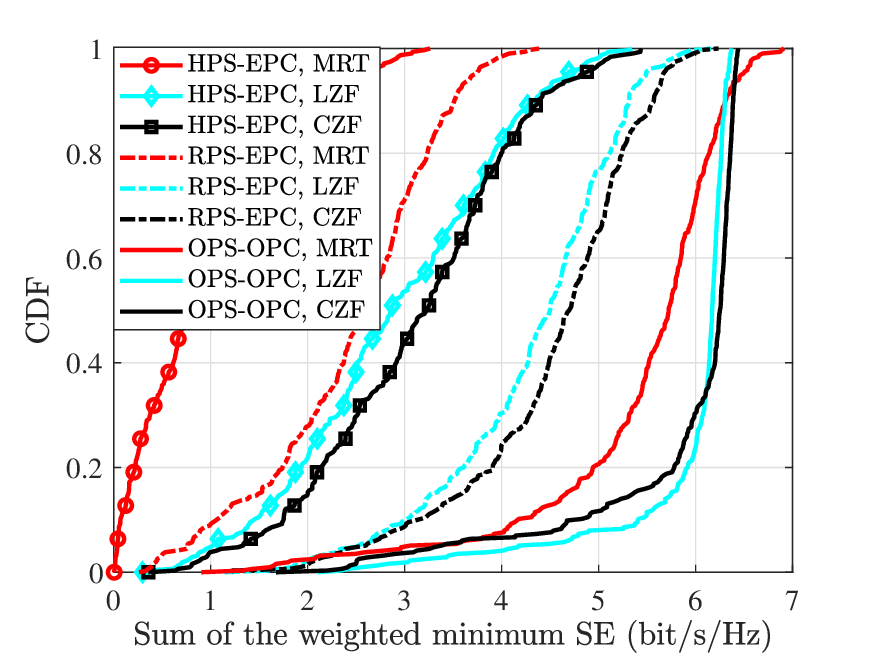}
        \caption{Comparison between precoding schemes ($M\!=\!400$, $w_n\!=\! w_f\!=\!0.5$).}
        \label{fig:Sum:MinSEs:Objective:Equal}
    \end{minipage}
    \hfill
    \begin{minipage}[t]{0.33\textwidth}
        \centering
        \includegraphics[trim=0 0cm 0cm 0cm,clip,width=\textwidth]{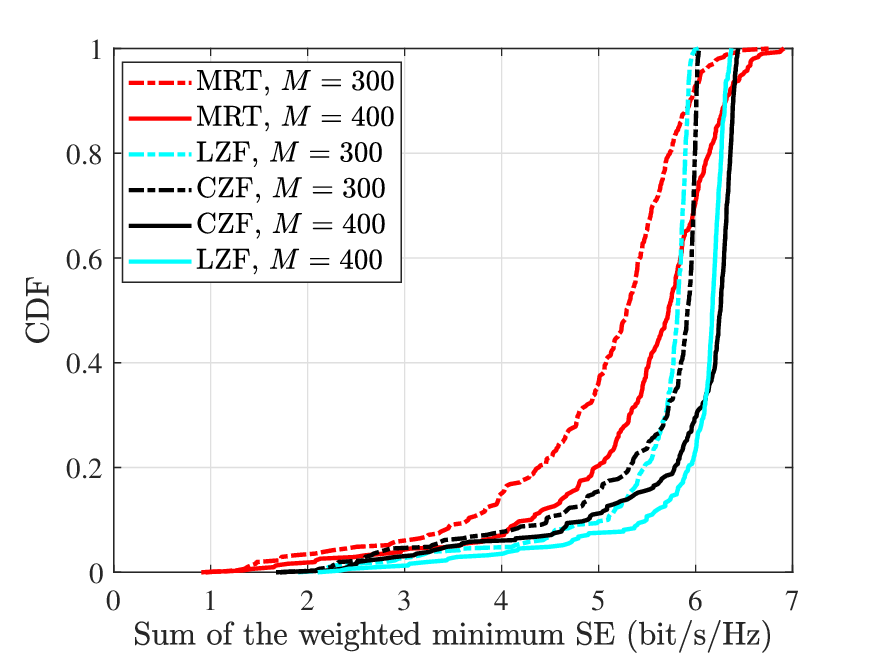}
        \caption{CDF of the objective function for different antenna numbers ($w_n \!=\! w_f \!=\!0.5$). }
        \label{fig:Effect:Antenna Number}
    \end{minipage}
\end{figure*}

\subsection{Results and Discussion}
\subsubsection{Convergence of Algorithms 2 and 3}
We study the convergence behavior of the proposed algorithm with $M=400$ and random initial values for the phase shifts, that is, $\theta_n \in [0, 2\pi), \forall n$. The convergence analysis consists of two parts: the convergence of the phase shifts design, shown in Figs. 3(a) and 3(b), and the convergence of the power control algorithm, illustrated in Fig. 3(c). As observed in Figs. 3(a) and 3(b), \textbf{Algorithm 2} converges to a stationary point after only a few iterations—specifically, after $5$ iterations with MRT and $7$ iterations with CZF. We also observe that with different initial points, the optimized minimum channel gain values are very close, confirming that the algorithm converges irrespective of the initial value. Note that the small transition in Fig. 3(b) with CZF is because at first the penalty factor $\rho$ is relatively small, $\qV$ is not a rank-one matrix. Regarding the convergence of \textbf{Algorithm 3}, we can see that the objective values of increase monotonically, as shown in Fig. 3(c), which aligns with our convergence behavior in Section \ref{Sec:Solution_Optimization}. We also observe that MRT requires more iterations to converge—$16$ iterations—compared to CZF and LZF, which converge in $6$ and $7$ iterations, respectively. This difference arises because CZF and LZF effectively eliminate intra-group interference, enabling faster convergence to a stationary point.
\subsubsection{Performance of Proposed VR Selection}
Figure \ref{fig:VR selection} shows the average of VR efficiency, which is defined as $$\bar{\omega}^{i} = \frac{1}{KS}\Big(\sum\nolimits_{k\in \Kn}\trace(\bar{\boldsymbol{\mathcal{D}}}_k^{i}) + \sum\nolimits_{k\in \Kf}\trace(\Tilde{\boldsymbol{\mathcal{D}}}_k^{i})\Big), $$ where $i\in\{\CZF,\MRT\}$ versus the number of antennas $M$. This parameter represents the ratio of the average number of antennas used in the VR-based design relative to the total array size. We consider two VR selection ratios, $\delta=0.7$ and $\delta=0.8$. To enable a fair comparison, we define two benchmarks: \textbf{Benchmark 1} corresponds to the VR selection scheme in \cite{Zhi:JSAC:2024}, and \textbf{Benchmark 2} refers to our proposed scheme applied only to NFUEs. As the number of antennas increases, the proportion used in VR-based design decreases since the received power is limited to a portion of the array, making VR-based design essential for large arrays. CZF requires more antennas than MRT, e.g., for $M=500$ and $\delta=0.7$, $78\%$ of the antennas are used with CZF, while only $45\%$ with MRT. This is because CZF suppresses interference, requiring more antennas for better performance. Lowering $\delta$ from $0.8$ to $0.7$ with CZF reduces $\bar{\omega}^{\CZF}$ from $0.85$ to $0.78$, as fewer antennas are selected to meet the SINR requirements. Similar trends are observed for MRT. Another observation is that \textbf{Benchmark 2} selects fewer antennas compared to \textbf{Benchmark 1}. This is primarily because we use SINR as the VR selection criterion, which tends to be smaller than the SNR used in \textbf{Benchmark 1}. Note that LZF uses the same VR obtained for CZF, thus, the results for LZF are not shown in this figure.
\subsubsection{Comparison of Different Precoding Schemes}
Figure \ref{fig:Sum:MinSEs:Objective:Equal} presents the cumulative distribution function (CDF) of the weighted minimum SE sum for NFUEs and FFUEs using CZF, LZF, and MRT precoding. The proposed \textbf{OPS-OPC} scheme significantly outperforms \textbf{RPS-EPC} across all precoding methods. With CZF, \textbf{OPS-OPC} achieves a median SE of $6.2$ bit/s/Hz, a $31.9\%$ improvement over \textbf{RPS-EPC} ($4.7$ bit/s/Hz). For LZF, the gain is $37.8\%$, while for MRT, the median SE increases by $119.2\%$ (from $2.6$ to $5.7$ bit/s/Hz). Notably, LZF performs comparably to CZF in both schemes, making it an efficient alternative with lower computational complexity, especially for large arrays. Interestingly, the performance of \textbf{HPS-EPC} is inferior to that of \textbf{RPS-EPC}. The \textbf{HPS-EPC} design enhances the system's sum-SE compared to \textbf{RPS-EPC} under MRT. However, under CZF and LZF precoding, it fails to improve both the sum-SE and the weighted minimum SE across NFUEs and FFUEs. (Sum-SE results are omitted from the paper due to space limitations.) Specifically, under MRT, configuring all RIS elements to coherently enhance the signal for a particular FFUE $k$ maximizes the reflected power gain for that user. However, this RIS phase configuration causes incoherent—or even destructive—combining at the RIS for other users. Therefore, it degrades the minimum SE across the FFUEs compared to the \textbf{RPS-EPC}. In the case of CZF and LZF, customizing the RIS to favor FFUE $k$ results in highly correlated effective channels among the FFUEs. This correlation undermines the ability of CZF and LZF to suppress inter-group interference, leading to a notable degradation in SE—particularly for FFUE $k$, as interference from other FFUEs cannot be effectively mitigated. As a result, the weighted minimum SE across NFUEs and FFUEs under \textbf{HPS-EPC} becomes inferior to that achieved with \textbf{RPS-EPC}.   

\subsubsection{The Effect of Number of Antennas}
In Fig. \ref{fig:Effect:Antenna Number}, we study the CDF of the sum of the weighted minimum SE of the XL-MIMO systems for two different number of antennas, i.e. $M \!=\! 300$, $M \!=\! 400$ and with $w_n \!=\!w_f \!=\! 0.5$. We can see that by increasing the number of antennas from $M\!=\!300$ to $M\!=\!400$, the median performance of CZF, LZF and MRT can be improved by $5.76\%$, $6.1\%$ and $8.1\%$, respectively. 

\vspace{-0.1em}
\subsubsection{Impact of Priority Factors}

\begin{figure*}[t]
\hspace{-1em}
    \centering
    \begin{subfigure}[b]{0.33\textwidth}
        \centering
        \includegraphics[trim=0 0cm 0cm 0cm,clip,width=\textwidth]{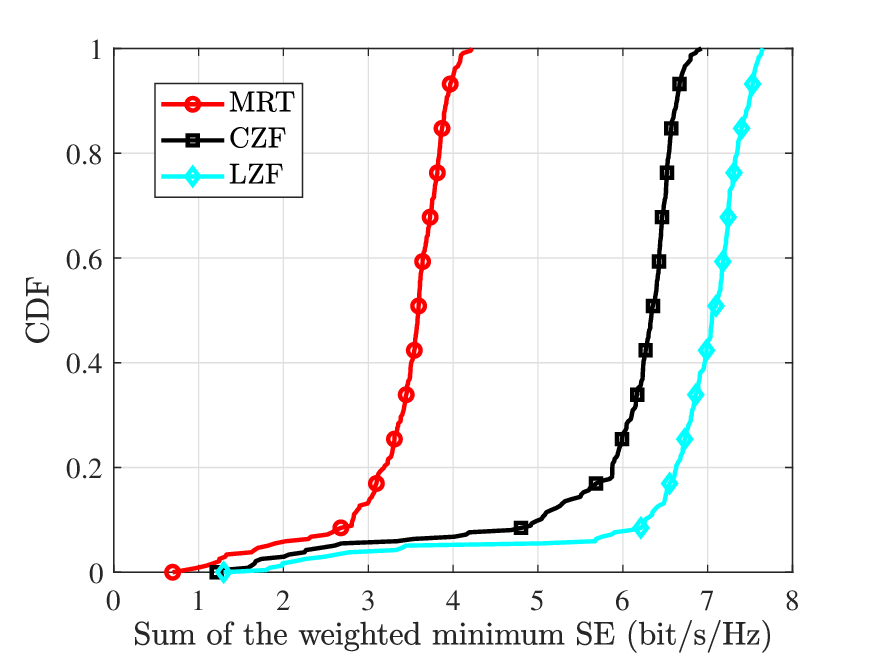}
        \caption{Objective function}
        \label{fig:Sum:MinSEs:Objective:FF}
    \end{subfigure} %
    \hfill
    \begin{subfigure}[b]{0.33\textwidth}
        \centering
        \includegraphics[trim=0 0cm 0cm 0cm,clip,width=\textwidth]{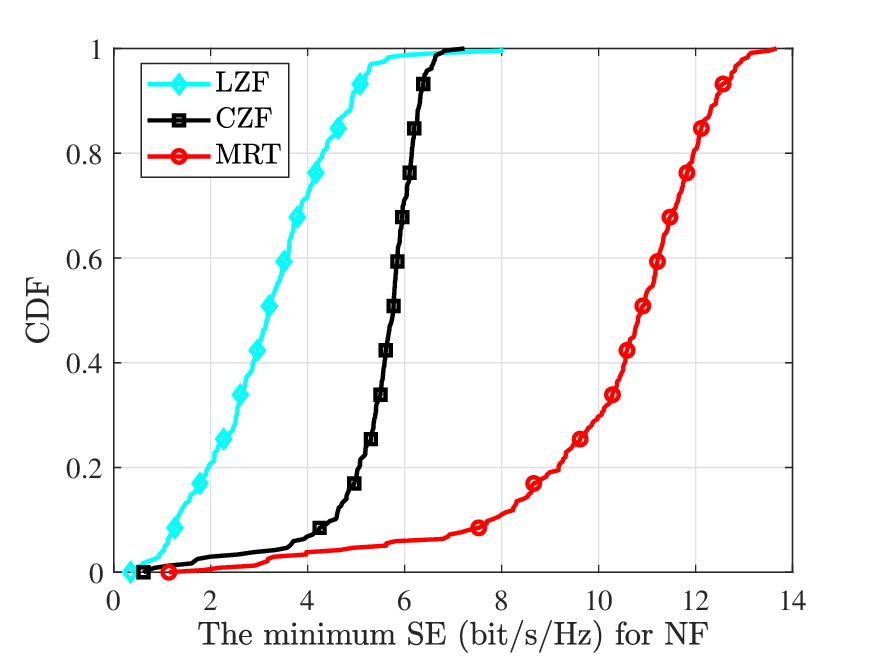}
        \caption{CDF of the minimum SE for NF}
        \label{fig:Sum:MinSE_N:FF}
    \end{subfigure} 
    \hfill
    \begin{subfigure}[b]{0.33\textwidth}
        \centering
        \includegraphics[trim=0 0cm 0cm 0cm,clip,width=\textwidth]{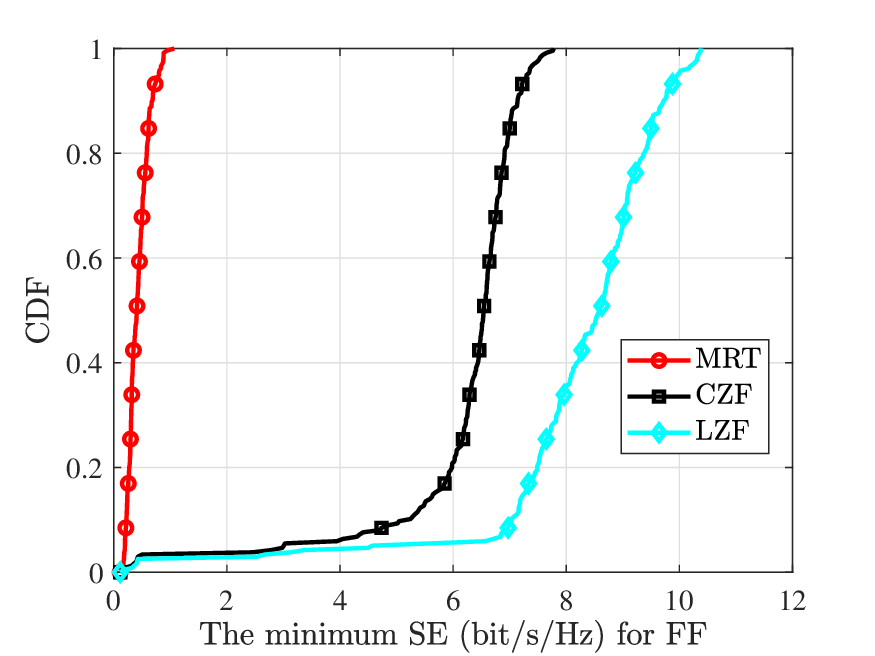}
        \caption{CDF of the minimum SE for FF}
        \label{fig:Sum:MinSE_F:FF}
    \end{subfigure}    
     \vspace{-0.2em}
    \caption{CDFs of the SE for different precoding schemes and for \textbf{OPS-OPC} case ($w_n = 0.3$, $w_f = 0.7$, $M=400$).}
    \label{fig:Different precoding:FF}  
    \vspace{-0.5em}
\end{figure*} 
Next, we compare the CDFs of different precoding schemes under the \textbf{OPS-OPC} design, prioritizing FFUEs with weights $w_n \!=\! 0.3$, $w_f \!=\! 0.7$ and $M\!=\!400$. From Fig. \ref{fig:Different precoding:FF}(\subref{fig:Sum:MinSEs:Objective:FF}), we observe that LZF achieves the highest median sum of the weighted minimum SE among the three precoding schemes. Specifically, LZF improves the median sum of the weighted minimum SE by a factor of $1.11$ compared to CZF and by a factor of $1.97$ compared to MRT. This is expected as LZF can mitigate certain interference components within the same subarrays and offers more flexibility in power control design than CZF, leading to better performance. Furthermore, since MRT cannot suppress interference, the median performance with CZF and LZF is significantly higher than that with MRT. The same insights can be obtained for the CDF of minimum SE for FF, see Fig. \ref{fig:Different precoding:FF}(\subref{fig:Sum:MinSE_F:FF}). In Fig. \ref{fig:Different precoding:FF}(\subref{fig:Sum:MinSE_N:FF}), we observe that MRT achieves the highest median minimum SE for NFUEs, followed by CZF, with LZF performing the worst. This is because, with MRT, nearly all transmit power is allocated to maximize the minimum SE of NFUEs, as the signals transmitted to the FFUEs are weak and MRT lacks interference mitigation capability. As a result, MRT is more effective in enhancing the performance of NFUE. In contrast, CZF provides balanced performance for both NFUEs and FFUEs by effectively eliminating intra-group interference. Meanwhile, with LZF, the priority is given to FFUEs, i.e., less power is allocated to optimize the performance of NFUEs.  
\begin{figure*}[t]
\hspace{-1em}
    \centering
    \begin{subfigure}[b]{0.33\textwidth}
        \centering
        \includegraphics[trim=0 0cm 0cm 0cm,clip,width=\textwidth]{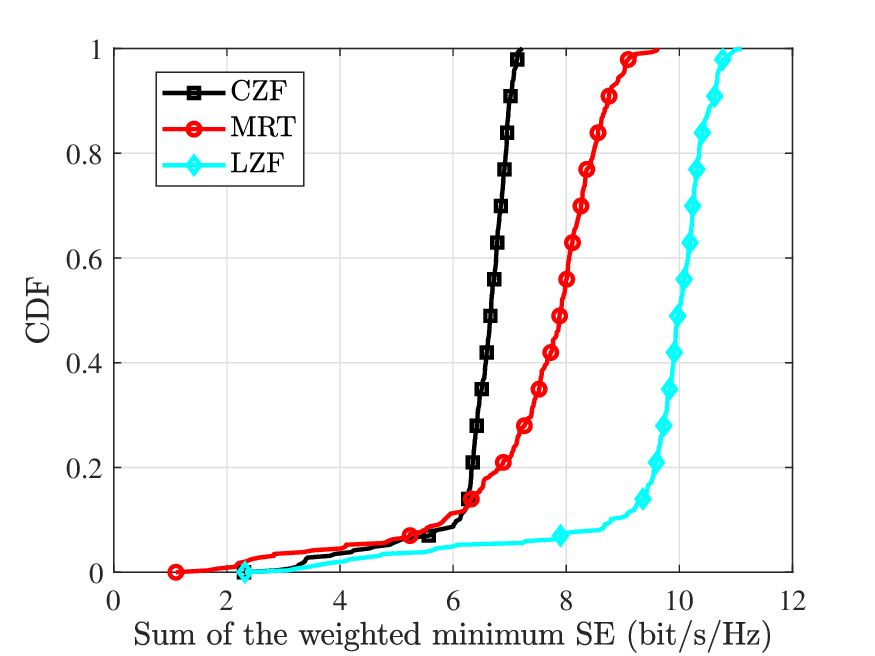}
        \caption{Objective function}
        \label{fig:Sum:MinSEs:Objective:NF}
    \end{subfigure} %
    \hfill
    \begin{subfigure}[b]{0.33\textwidth}
        \centering
        \includegraphics[trim=0 0cm 0cm 0cm,clip,width=\textwidth]{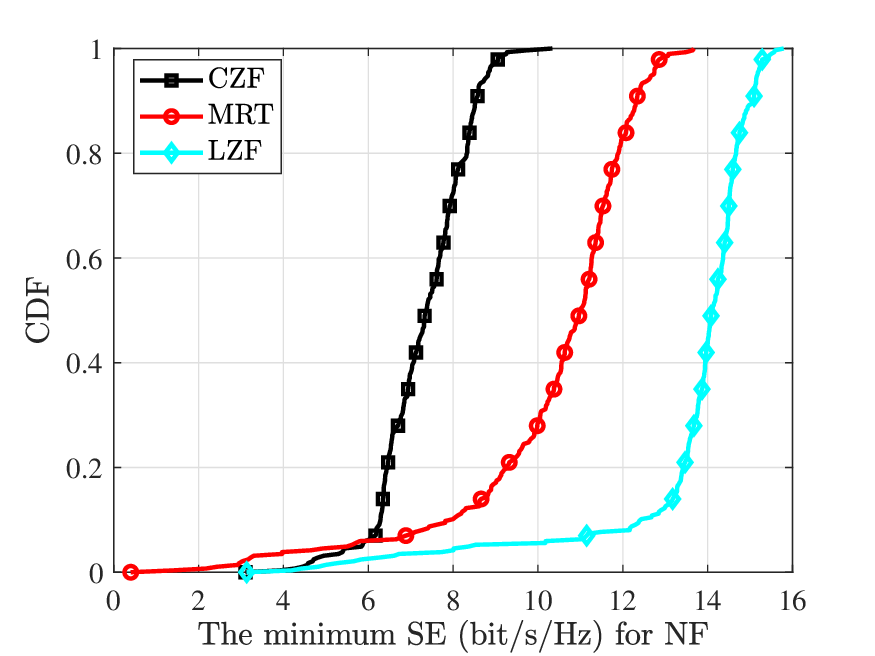}
        \caption{CDF of the minimum SE for NF}
        \label{fig:Sum:MinSE:NF}
    \end{subfigure} 
    \hfill
    \begin{subfigure}[b]{0.32\textwidth}
        \centering
        \includegraphics[trim=0 0cm 0cm 0cm,clip,width=\textwidth]{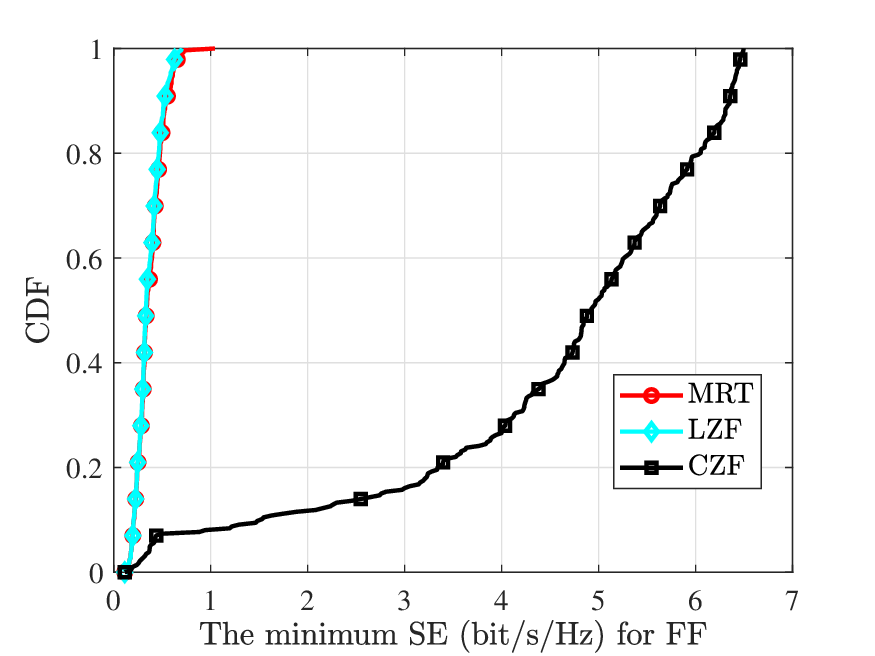}
        \caption{CDF of the minimum SE for FF}
        \label{fig:Sum:MinSE:FF}
    \end{subfigure}
    \vspace{-0.2em}
    \caption{CDFs of the SE for different precoding schemes and for \textbf{OPS-OPC} case ($w_n = 0.7$, $w_f = 0.3$, $M=400$).}
    \label{fig:Different precoding:NF}
    \vspace{-0.7em}
\end{figure*}

Figure \ref{fig:Different precoding:NF} evaluates the performance with CZF, LZF and MRT when prioritizing NFUEs in the \textbf{OPS-OPC} scenario with $w_n\!=\!0.7$, $w_f\!=\!0.3$, and $M\!=\!400$. In Fig.~\ref{fig:Different precoding:NF}(\subref{fig:Sum:MinSE:NF}), LZF demonstrates the best performance among the three precoding schemes, with MRT outperforming CZF. This outcome is primarily attributed to the objective function of the optimization problem, which aims to maximize the weighted sum SE of both NFUEs and FFUEs, rather than prioritizing a specific user group. MRT outperforms CZF since, in aiming to maximize the objective function, it allocates more power to enhance the performance of NFUEs, which consequently results in reduced power allocation for FFUEs. This trade-off is evident in Fig. \ref{fig:Different precoding:NF}(\subref{fig:Sum:MinSE:FF}), where the median of the minimum SE for FFUEs drops to $0.33$ bit/s/Hz under MRT, compared to $4.92$ bit/s/Hz with CZF. On the other hand, LZF achieves the best overall performance due to its flexible power control design—it allows each subarray to independently adjust its transmit power to each user. In contrast, CZF applies a uniform power control coefficient across all subarrays for a given user. Nonetheless, CZF provides a more balanced performance between NFUEs and FFUEs, as reflected in Figs.~\ref{fig:Different precoding:NF}(\subref{fig:Sum:MinSE:NF}) and \ref{fig:Different precoding:NF}(\subref{fig:Sum:MinSE:FF}).

\section{Conclusion} \label{Conclusion}
We investigated the performance of a RIS-assisted downlink XL-MIMO system for both NFUEs and FFUEs with CZF, LZF and MRT precoding schemes, taking into account the effects of VR selection, phase shifts design at the RIS, and power control design at the BS. We proposed a heuristic VR selection design, a two-stage phase shifts and power control design relying on SCA to maximize the sum of the weighted minimum SE of the system under different priority requirements. Our simulation results indicated that phase shifts design and power control can significantly improve the performance of the XL-MIMO assisted by the RIS with three linear precoding schemes. In particular, when giving equal priority to NFUEs and FFUEs, CZF is the best option, while LZF could achieve comparable performance. If we give priority to NFUEs or FFUEs, LZF can obtain the best performance, and CZF strikes a balance between the performance of NFUEs and FFUEs.

CSI acquisition errors or delays can significantly impact the accuracy of beamforming in XL-MIMO systems, leading to degradation in the achievable rate. Our proposed framework relies on SINR expressions that assume perfect channel knowledge; therefore, any inaccuracy in CSI may result in underestimation or overestimation of SINR values. Furthermore, errors in CSI acquisition can lead to incorrect VR selection, which in turn will compromise the subsequent design of RIS phase shifts and power control strategies. Such imperfections may also affect the convergence and effectiveness of the proposed penalty-based iterative algorithm used for optimization of RIS phase shifts. As such, developing efficient CSI acquisition techniques for RIS-assisted XL-MIMO systems remains a critical and promising direction for future research. To support real-time implementation, hardware acceleration using platforms, such as Field-Programmable Gate Arrays (FPGAs) or Graphics Processing Units (GPUs), can be employed to expedite key computations, while machine learning-based methods can also be explored to further reduce latency.
\appendices
\section{Proof of Proposition~\ref{Prop:MRT:NF}}~\label{Prop:MRT:NF:proof}
In order to calculate $\bDSk$ for the MRT precoding scheme, substituting the precoding vector $\bar{\qw}_{k}^{\MRT}$ in ~\eqref{eq:MRT} into~\eqref{eq:FFU:NF}, we have the following results:
\begin{align}   \label{eq:desired signal of NFUE}
   \bDSk &= 
   \sqrt{P_n} \! \sum\nolimits_{s \in \SknMRT}\!\! \sqrt{\betask^{\MRT}} \betakks,
\end{align}
with the fact that $ \bar{\qD}_{sk}\bar{\qw}_{sk}^{\MRT} \!=\! \bar{\qg}_{sk}$ if $s \!\in\! \SknMRT$ and $\mathbf{0}_{M^\ast}$ otherwise.

Similarly, we obtain $\vert\bUIi \vert^2$ as follows:
\begin{align}   \label{eq:simplified near-field interference of near-field scenario}
   \vert\bUIi \vert^2
   &= P_n\Big\vert \sum\nolimits_{s \in \SniMRT}\!\!\sqrt{\betasi^{\MRT}}  \betakis\Big\vert^2.
\end{align}

Now, we derive $\mathbb{E}\big\{\vert\btUIj\vert^2\big\}$, which can be expressed as
\begin{align} \label{eq:UI_J:Near}
   &\mathbb{E}\big\{\vert\btUIj\vert^2\big\}  = P_f \Ex\Big\{ \Big\vert\ {\sum\nolimits_{s \in \Sa}\sqrt{\Tetasj^{\MRT}}  \bar{\qg}_{sk}^{H}\Tilde{\qD}_{sj}\Tilde{\qw}_{sj}^{\MRT}}\Big\vert^2\Big\} 
   \nonumber \\
   &= P_f  \Ex\Big\{\sum\nolimits_{s \in \SjfMRT}\!\!\sum\nolimits_{s' \in \SjfMRT}\!\!\sqrt{\Tetasj^{\MRT} \Tetaspj^{\MRT}}\bar{\qg}_{sk}^{H}\Tilde{\qg}_{sj}
   \Tilde{\qg}_{s'j}^H \bar{\qg}_{s'k} \Big\} \nonumber \\
   &= \!P_f\!  {\sum\nolimits_{s \in \SjfMRT}\!\!\sum\nolimits_{s' \in \SjfMRT}\!\!\sqrt{\Tetasj^{\MRT} \Tetaspj^{\MRT}}\bar{\qg}_{sk}^{H}\Ex\big\{\Tilde{\qg}_{sj} \Tilde{\qg}_{s'j}^H \big\}\bar{\qg}_{s'k} },
\end{align}
where $\Tilde{\qg}_{sj}^H = \qh_j^H\bTeta\qH_{2,s}$, and we have
\begin{align} \label{eq:near-field F_ss}
    &\Ex\big\{\Tilde{\qg}_{sj} \Tilde{\qg}_{s'j}^H \big\} =\Ex\{ (\qh_j^H\bTeta\qH_{2,s})^H (\qh_j^H\bTeta\qH_{2,s'})\} \nonumber \\
    &= \Ex\{(\alpha_2\bar{\qH}_{2,s}^H + \beta_2\Tilde{\qH}_{2,s}^H) \qM_{jj} (\alpha_2\bar{\qH}_{2,s'} + \beta_2\Tilde{\qH}_{2,s'}) \} \nonumber \\
    &= \alpha_2^2 \bar{\qH}_{2,s}^H \qM_{jj} \bar{\qH}_{2,s'} + \beta_2^2 \Ex\{\Tilde{\qH}_{2,s}^H \qM_{jj} \Tilde{\qH}_{2,s'} \}. 
\end{align}
When $s=s'$, we have $\Ex\{\Tilde{\qH}_{2,s}^H \qM_{jj} \Tilde{\qH}_{2,s'} \} = \trace(\qM_{jj}) \qI_{M^{\ast}} = \varsigma_k N\qI_{M^{\ast}}$. Otherwise, $\Ex\{\Tilde{\qH}_{2,s}^H \qM_{jj} \Tilde{\qH}_{2,s'} \} = \boldsymbol{0}_{M^{\ast}\times M^{\ast}}$. Therefore, substituting \eqref{eq:near-field F_ss} into \eqref{eq:UI_J:Near}, we have
\begin{align} \label{eq:final far-field interference expression}
   \Ex\big\{\vert\btUIj\vert^2\big\} 
   = \!P_f \! \sum\nolimits_{s \in \SjfMRT}\!\!\sum\nolimits_{s' \in \SjfMRT}\!\!\sqrt{\Tetasj^{\MRT} \Tetaspj^{\MRT}} \Talphakjssp.     
\end{align}
Substituting~\eqref{eq:desired signal of NFUE}, \eqref{eq:simplified near-field interference of near-field scenario} and \eqref{eq:final far-field interference expression} into \eqref{eq:simplified SINR expression of NFUE}, we obtain the final result given in~\eqref{eq:FSINR:NF:VR}.

\section{Proof of Proposition~\ref{Prop:MRT:FF}}~\label{Prop:MRT:FF:proof}
For FFUE $k$ with the MRT precoding scheme, we substitute the precoding vector $\Tilde{\qw}_{k}^{\MRT}$ in \eqref{eq:MRT} into~\eqref{eq:FFU:UF} to derive $\tDSk$, $\mathbb{E}\{ \vert\tBUk\vert^2\}$, $\mathbb{E}\{ \vert\tUIj\vert^2\}$ and $\mathbb{E}\{\vert\tbUIi\vert^2\}$ as follows: 
\begin{align} \label{eq:simplified desired signal of far-field scenario}
    \tDSk &= \sqrt{P_f}  \mathbb{E}\Big\{\sum\nolimits_{s\in\Set}\!\!\sqrt{\Tilde{\eta}_{sk}^{\MRT}}\Tilde{\qg}_{sk}^H \Tilde{\qD}_{sk}\Tilde{\qw}_{sk}^{\MRT}\Big\} \nonumber \\
    & = \sqrt{P_f}\sum\nolimits_{s \in \SkfMRT} \!\! \sqrt{\Tilde{\eta}_{sk}^{\MRT}} \Tcks,
\end{align}
with the fact that $\Tilde{\qD}_{sk} \Tilde{\qw}_{sk} \!=\! \Tilde{\qg}_{sk} \in \mathbb{C}^{M^{\ast} \times 1} $ if $s \!\in\! \SkfMRT$, and
\begin{align}   \label{eq:E value of far-field scenario}
    \Tcks &\triangleq \Ex\big\{\Tilde{\qg}_{sk}^H \Tilde{\qg}_{sk}\big\}=
    \Ex\{ \qh_k^H\bTeta\qH_{2,s} (\qh_k^H\bTeta\qH_{2,s})^H\} \nonumber \\
    &= \qh_k^H \bTeta \Ex\{  \qH_{2,s} \qH_{2,s}^H \} \bTeta^H \qh_k \nonumber \\
    &= \qh_k^H\bTeta (\alpha_2^2 \bar{\qH}_{2,s} \bar{\qH}_{2,s}^H + \beta_2^2 M^{\ast} \qI_N) \bTeta^H\qh_k.
\end{align}

To derive $\tBUk$, we use the fact that the variance of a sum of independent RVs is equal to the sum of the variances. Therefore, $\Ex\{ \vert\tBUk\vert^2\}$ can be expressed as
\begin{align} \label{eq:beamforming uncertainty of far-field scenario}
    \Ex\{ \vert\tBUk\vert^2\}  
    &\!=\! P_f \! \! \sum\nolimits_{s \in \SkfMRT}\! \! \Tilde{\eta}_{sk}^{\MRT}\Ex\big\{ \vert\Tilde{\qg}_{sk}^H \Tilde{\qg}_{sk}\vert^2\} \!-\! \vert\Ex\{\Tilde{\qg}_{sk}^H \Tilde{\qg}_{sk}\} \vert^2 \big\} \nonumber \\
    &\!=\! P_f \! \! \sum\nolimits_{s \in \SkfMRT}\! \!\Tilde{\eta}_{sk}^{\MRT}\big( T^{s}_k-  (\Tcks)^2 \big),
\end{align}
where $T^{s}_k \triangleq \Ex\big\{ \vert\Tilde{\qg}_{sk}^H \Tilde{\qg}_{sk}\vert^2\}$, which can be obtained as
\begin{align} \label{eq:T_s}
    T^{s}_k 
    &=\Ex\Big\{ \vert\qh_k^H\bTeta\qH_{2,s} (\qh_k^H\bTeta\qH_{2,s})^H\vert^2\Big\} \nonumber \\
    &= \Ex\big\{ (\qh_k^H\bTeta \qH_{2,s} \qH_{2,s}^H \bTeta^H \qh_k) (\qh_k^H\bTeta \qH_{2,s} \qH_{2,s}^H \bTeta^H \qh_k)^H \big \} \nonumber \\ 
    &= \qh_k^H\bTeta \qC_{sk} \bTeta^H \qh_k,  
\end{align}
where $\qC_{sk} \triangleq \Ex\big\{ \qH_{2,s} \qH_{2,s}^H \qM_{kk} \qH_{2,s} \qH_{2,s}^H \big \}$ and $\qC_{sk}$ can be obtained by replacing $\qM_{jj}$ with $\qM_{kk}$ in \eqref{eq:C_sj}.

Substituting \eqref{eq:T_s} and $\Tcks$ into \eqref{eq:beamforming uncertainty of far-field scenario}, we have
\begin{align} \label{eq:final beamforming uncertainty of far-field scenario}
\Ex\{ \vert\tBUk\vert^2\} = P_f \! \sum\nolimits_{s \in \SkfMRT}\! \Tilde{\eta}_{sk}^{\MRT} \Taks.   
\end{align}

Now, we focus on $\Ex\big\{\big\vert\tUIj\big\vert^2\big\}$, which can be expressed as
\begin{align} \label{eq:updated far-field interference of far-field scenario}
    \Ex\big\{\big\vert\tUIj\big\vert^2\big\}  &= P_f  \Ex\Big\{ \Big\vert\sum\nolimits_{s\in \Set}\sqrt{\Tetasj^{\MRT}} \Tilde{\qg}_{sk}^H \Tilde{\qD}_{sj}\Tilde{\qw}_{sj}^{\MRT} \Big\vert^2 \Big\} \nonumber \\ 
    &\hspace{-5em}= P_f \! \sum\nolimits_{s \in \SjfMRT}\!\!\sum\nolimits_{s' \in \SjfMRT}\!\!\sqrt{\Tetasj^{\MRT}\Tetaspj^{\MRT}} \Ex\big\{ \Tilde{\qg}_{sk}^H \Tilde{\qg}_{sj} \Tilde{\qg}_{s'j}^H \Tilde{\qg}_{s'k} \big\}. 
\end{align}
The expectation term in~\eqref{eq:updated far-field interference of far-field scenario} can be rewritten as
\begin{align} \label{eq:L_kk}
    &\Ex\big\{ \Tilde{\qg}_{sk}^H \Tilde{\qg}_{sj} \Tilde{\qg}_{s'j}^H \Tilde{\qg}_{s'k} \big\}  \nonumber \\
    &= \Ex\big\{ \qh_k^H \bTeta \qH_{2,s} \qH_{2,s}^H \bTeta^H \qh_j \qh_j^H \bTeta \qH_{2,s'} \qH_{2,s'}^H \bTeta^H \qh_k \big\}    \nonumber \\
    &= \qh_k^H\bTeta \Ex\big\{  \qH_{2,s} \qH_{2,s}^H \qM_{jj} \qH_{2,s'} \qH_{2,s'}^H  \big\} \bTeta^H \qh_k. 
\end{align}
Note that we have $\Ex\big\{ \Tilde{\qg}_{sk}^H \Tilde{\qg}_{sj} \Tilde{\qg}_{sj}^H \Tilde{\qg}_{sk} \big\} =\qh_k^H\bTeta \qC_{sj}  \bTeta^H \qh_k$
when $s=s'$ based on \eqref{eq:T_s}.  Otherwise, the expectation in~\eqref{eq:L_kk} can be derived as
\begin{align*} 
    &\qD_j^{ss'} \triangleq 
    \Ex\big\{  \qH_{2,s} \qH_{2,s}^H\} \qM_{jj} \Ex\big\{\qH_{2,s'} \qH_{2,s'}^H  \big\}  \nonumber \\
    &= (\alpha_2^2 \bar{\qH}_{2,s} \bar{\qH}_{2,s}^H \!+\! \beta_2^2 M^{\ast} \qI_N) \qM_{jj} (\alpha_2^2 \bar{\qH}_{2,s'} \bar{\qH}_{2,s'}^H  \!+\! \beta_2^2 M^{\ast} \qI_N).
\end{align*}
Substituting \eqref{eq:L_kk} into \eqref{eq:updated far-field interference of far-field scenario}, we obtain
\begin{align} \label{eq:final intra-group interference in far-field scenario}
\Ex\{\big\vert\tUIj\big\vert^2\big\}  = P_f \sum\nolimits_{s \in \SjfMRT}\!\!\sum\nolimits_{s' \in \SjfMRT }\!\!\sqrt{\Tetasj^{\MRT}\Tetaspj^{\MRT}} \Tbkjssp.
\end{align}

Finally, we compute $\Ex\big\{\big\vert\tbUIi\big\vert^2\big\}$, which can be written as
\begin{align} \label{eq:UI_I:Far}
    &\Ex\Big\{\big\vert\tbUIi\big\vert^2\Big\}  = P_n \Ex\Big\{ \big\vert\sum\nolimits_{s \in \SniMRT}\!\!\sqrt{\betasi^{\MRT}}\Tilde{\qg}_{sk}^H \bwsiMRT \big\vert^2 \Big\} \nonumber \\
    &= P_n \Ex\Big\{\sum\nolimits_{s \in \SniMRT}\!\! \sum\nolimits_{s' \in \SniMRT}\!\! \sqrt{\betasi^{\MRT} \betaspi^{\MRT}} \; \Tilde{\qg}_{sk}^H \bar{\qg}_{si} \bar{\qg}_{s'i}^H \Tilde{\qg}_{s'k} \Big\} \nonumber \\
    &= \!P_n\! \sum\nolimits_{s\in \SniMRT}\!\! \sum\nolimits_{s' \in \SniMRT}\!\! \sqrt{\betasi^{\MRT} \betaspi^{\MRT}} \bar{\qg}_{s'i}^H \Ex \{\Tilde{\qg}_{s'k} \Tilde{\qg}_{sk}^H \} \bar{\qg}_{si}, 
\end{align}
where
\begin{align}   
    &\Ex\big\{\Tilde{\qg}_{s'k} \Tilde{\qg}_{sk}^H \big\} 
    = \Ex\{(\qh_k^H\bTeta\qH_{2,s'})^H\qh_k^H\bTeta\qH_{2,s}\} \nonumber \\
    &=\Ex\{(\alpha_2\bar{\qH}_{2,s'}^H + \beta_2\Tilde{\qH}_{2,s'}^H) \qM_{kk}(\alpha_2\bar{\qH}_{2,s} + \beta_2\Tilde{\qH}_{2,s})\} \nonumber \\
    &= \alpha_2^2\bar{\qH}_{2,s'}^H\qM_{kk}\bar{\qH}_{2,s} + \beta_2^2 \Ex\{ \Tilde{\qH}_{2,s'}^H \qM_{kk} \Tilde{\qH}_{2,s})\},
\end{align}
when $s=s'$, we have $\Ex\{\Tilde{\qH}_{2,s'}^H \qM_{kk} \Tilde{\qH}_{2,s})\} = \trace(\qM_{kk})\qI_{M^{\ast}} = \varsigma_k N\qI_{M^{\ast}}$. Otherwise, $\Ex\{\Tilde{\qH}_{2,s'}^H \qM_{kk} \Tilde{\qH}_{2,s})\} = \boldsymbol{0}_{M^{\ast}\times M^{\ast}}$.
Therefore, we can obtain the following result
\begin{align}   \label{eq:near-field interference of far-field scenario}
    \Ex\big\{\big\vert\bUIi\big\vert^2\big\}\! =\! P_n \sum\nolimits_{s\in \SniMRT} \!\!\sum\nolimits_{s'\in \SniMRT}\!\! \sqrt{\betasi^{\MRT} \betaspi^{\MRT}} \balphakissp.
\end{align}
Substituting \eqref{eq:simplified desired signal of far-field scenario}, \eqref{eq:final beamforming uncertainty of far-field scenario}, \eqref{eq:final intra-group interference in far-field scenario}, \eqref{eq:near-field interference of far-field scenario} into \eqref{eq:simplified SINR expression of FFUE}, we obtain the final result \eqref{eq:FSINR:FF:VR}.

\bibliographystyle{IEEEtran}
\bibliography{IEEEabrv,references}

\begin{thebibliography}{10}
\providecommand{\url}[1]{#1}
\csname url@samestyle\endcsname
\providecommand{\newblock}{\relax}
\providecommand{\bibinfo}[2]{#2}
\providecommand{\BIBentrySTDinterwordspacing}{\spaceskip=0pt\relax}
\providecommand{\BIBentryALTinterwordstretchfactor}{4}
\providecommand{\BIBentryALTinterwordspacing}{\spaceskip=\fontdimen2\font plus
\BIBentryALTinterwordstretchfactor\fontdimen3\font minus
  \fontdimen4\font\relax}
\providecommand{\BIBforeignlanguage}[2]{{%
\expandafter\ifx\csname l@#1\endcsname\relax
\typeout{** WARNING: IEEEtran.bst: No hyphenation pattern has been}%
\typeout{** loaded for the language `#1'. Using the pattern for}%
\typeout{** the default language instead.}%
\else
\language=\csname l@#1\endcsname
\fi
#2}}
\providecommand{\BIBdecl}{\relax}
\BIBdecl
\renewcommand{\BIBentryALTinterwordstretchfactor}{4}

\bibitem{Wang:Surveys:2024}
Z.~Wang \emph{et~al.}, ``A tutorial on extremely large-scale {MIMO} for {6G}:
  {Fundamentals}, signal processing, and applications,'' \emph{IEEE Commun.
  Surveys Tuts.}, vol.~26, no.~3, pp. 1560--1605, Thirdquarter 2024.

\bibitem{matthaiou2021road}
M.~Matthaiou, O.~Yurduseven, H.~Q. Ngo, D.~Morales-Jimenez, S.~L. Cotton, and
  V.~F. Fusco, ``The road to {6G}: Ten physical layer challenges for
  communications engineers,'' \emph{{IEEE} Commun. Mag.}, vol.~59, no.~1, pp.
  64--69, Jan. 2021.

\bibitem{Deng:JSAC:2023}
R.~Deng \emph{et~al.}, ``Reconfigurable holographic surfaces for ultra-massive
  {MIMO} in 6{G}: Practical design, optimization and implementation,''
  \emph{{IEEE} J. Sel. Areas Commun.}, vol.~41, no.~8, pp. 2367--2379, Aug.
  2023.

\bibitem{Li:TWC:2024}
X.~Li, Z.~Dong, Y.~Zeng, S.~Jin, and R.~Zhang, ``Multi-user modular {XL-MIMO}
  communications: Near-field beam focusing pattern and user grouping,''
  \emph{{IEEE} Trans. Wireless Commun.}, vol.~23, no.~10, pp. 13\,766--13\,781,
  Oct. 2024.

\bibitem{Wu:2021:TCM}
Q.~Wu, S.~Zhang, B.~Zheng, C.~You, and R.~Zhang, ``Intelligent reflecting
  surface-aided wireless communications: A tutorial,'' \emph{{IEEE} Trans.
  Commun.}, vol.~69, no.~5, pp. 3313--3351, May 2021.

\bibitem{Zijian:TCOM:2023}
Z.~Zhang \emph{et~al.}, ``Active {RIS} vs. passive {RIS}: Which will prevail in
  {6G}?'' \emph{{IEEE} Trans. Commun.}, vol.~71, no.~3, pp. 1707--1725, Mar.
  2023.

\bibitem{Mu:TWC:2022}
X.~Mu, Y.~Liu, L.~Guo, J.~Lin, and R.~Schober, ``Simultaneously transmitting
  and reflecting {(STAR) RIS} aided wireless communications,'' \emph{{IEEE}
  Trans. Wireless Commun.}, vol.~21, no.~5, pp. 3083--3098, May 2022.

\bibitem{Nerini:TWC:2024}
M.~Nerini, S.~Shen, H.~Li, and B.~Clerckx, ``Beyond diagonal reconfigurable
  intelligent surfaces utilizing graph theory: Modeling, architecture design,
  and optimization,'' \emph{{IEEE} Trans. Wireless Commun.}, vol.~23, no.~8,
  pp. 9972--9985, Aug. 2024.

\bibitem{Tor:J_STSP:2024}
G.~Torcolacci \emph{et~al.}, ``Holographic imaging with {XL-MIMO} and {RIS}:
  Illumination and reflection design,'' \emph{{IEEE} J. Sel. Topics Signal
  Process.}, vol.~18, no.~4, pp. 587--602, May 2024.

\bibitem{Lee:TCOM:2024}
J.~Lee, H.~Chung, Y.~Cho, S.~Kim, and S.~Hong, ``Near-field channel estimation
  for {XL-RIS} assisted multi-user {XL-MIMO} systems: Hybrid beamforming
  architectures,'' \emph{{IEEE} Trans. Commun.}, vol.~73, no.~3, pp.
  1560--1574, Mar. 2025.

\bibitem{Zhang:letters:2024}
X.~Zhang, H.~Shao, W.~Zhang, Z.~Xie, X.~Yang, and W.~Jing, ``{RIS}-aided
  {XL-MIMO} channel estimation based on expectation-maximization,''
  \emph{{IEEE} Commun. Lett.}, vol.~28, no.~12, pp. 2869--2873, Dec. 2024.

\bibitem{Cui:CMG:2023}
M.~Cui, Z.~Wu, Y.~Lu, X.~Wei, and L.~Dai, ``Near-field {MIMO} communications
  for {6G}: {Fundamentals}, challenges, potentials, and future directions,''
  \emph{{IEEE} Commun. Mag.}, vol.~61, no.~1, pp. 40--46, Jan. 2023.

\bibitem{Lu:TWC:2022}
H.~Lu and Y.~Zeng, ``Communicating with extremely large-scale array/surface:
  Unified modeling and performance analysis,'' \emph{{IEEE} Trans. Wireless
  Commun.}, vol.~21, no.~6, pp. 4039--4053, June 2022.

\bibitem{Björnson:JOP:2020}
E.~Björnson and L.~Sanguinetti, ``Power scaling laws and near-field behaviors
  of massive {MIMO} and intelligent reflecting surfaces,'' \emph{{IEEE} Open J.
  Commun. Society}, vol.~1, pp. 1306--1324, Sept. 2020.

\bibitem{Dardari:2020:JSAC}
D.~Dardari, ``Communicating with large intelligent surfaces: Fundamental limits
  and models,'' \emph{{IEEE} J. Sel. Areas Commun.}, vol.~38, no.~11, pp.
  2526--2537, Nov. 2020.

\bibitem{Carvalho:IWC:2020}
E.~D. Carvalho, A.~Ali, A.~Amiri, M.~Angjelichinoski, and R.~W. Heath~Jr.,
  ``Non-stationarities in extra-large-scale massive {MIMO},'' \emph{IEEE
  Wireless Commun.}, vol.~27, no.~4, pp. 74--80, Aug. 2020.

\bibitem{Zhi:JSAC:2024}
K.~Zhi \emph{et~al.}, ``Performance analysis and low-complexity design for
  {XL-MIMO} with near-field spatial non-stationarities,'' \emph{{IEEE} J. Sel.
  Areas Commun.}, vol.~42, no.~6, pp. 1656--1672, June 2024.

\bibitem{Chen:TWC:2024}
Y.~Chen and L.~Dai, ``Non-stationarity channel estimation for extremely
  large-scale {MIMO},'' \emph{{IEEE} Trans. Wireless Commun.}, vol.~23, no.~7,
  pp. 7683--7697, July 2024.

\bibitem{Chen:PIMRC:2024}
H.~Chen \emph{et~al.}, ``{ELAA} near-field localization and sensing with
  partial blockage detection,'' in \emph{Proc. IEEE PIMRC}, Sept. 2024, pp.
  1--6.

\bibitem{Xu:TSP:2024}
W.~Xu, A.~Liu, M.-J. Zhao, and G.~Caire, ``Joint visibility region detection
  and channel estimation for {XL-MIMO} systems via alternating {MAP},''
  \emph{{IEEE} Trans. Signal Process.}, vol.~72, pp. 4827--4842, Oct. 2024.

\bibitem{Xiao:WCNC:2024}
X.~Cao, M.~Mohammadi, H.~Q. Ngo, and M.~Matthaiou, ``{RIS}-assisted {XL-MIMO}
  for coexistence of near-field and far-field communications,'' in \emph{Proc.
  IEEE WCNC}, Apr. 2024, pp. 1--6.

\bibitem{Amiri:Globecom:2018}
A.~Amiri, M.~Angjelichinoski, E.~de~Carvalho, and R.~W. Heath~Jr., ``Extremely
  large aperture massive {MIMO}: Low complexity receiver architectures,'' in
  \emph{Proc. IEEE GLOBECOM}, Dec. 2018, pp. 1--6.

\bibitem{Mari:TVT:2020}
J.~C. Marinello, T.~Abrão, A.~Amiri, E.~de~Carvalho, and P.~Popovski,
  ``Antenna selection for improving energy efficiency in {XL-MIMO} systems,''
  \emph{{IEEE} Trans. Veh. Technol.}, vol.~69, no.~11, pp. 13\,305--13\,318,
  Nov. 2020.

\bibitem{Ribe:EUSIPCO:2021}
L.~N. Ribeiro, S.~Schwarz, and M.~Haardt, ``Low-complexity zero-forcing
  precoding for {XL-MIMO} transmissions,'' in \emph{Proc. {IEEE} EUSIPCO}, Aug.
  2021, pp. 1621--1625.

\bibitem{Yang:TVT:2020}
X.~Yang, F.~Cao, M.~Matthaiou, and S.~Jin, ``On the uplink transmission of
  extra-large scale massive {MIMO} systems,'' \emph{{IEEE} Trans. Veh.
  Technol.}, vol.~69, no.~12, pp. 15\,229--15\,243, Dec. 2020.

\bibitem{Lu:TWC:2024}
Z.~Lu, Y.~Han, S.~Jin, and M.~Matthaiou, ``Near-field localization and channel
  reconstruction for {ELAA} systems,'' \emph{{IEEE} Trans. Wireless Commun.},
  vol.~23, no.~7, pp. 6938--6953, July 2024.

\bibitem{ngo16}
T.~L. Marzetta, E.~G. Larsson, H.~Yang, and H.~Q. Ngo, \emph{Fundamentals of
  Massive {MIMO}}.\hskip 1em plus 0.5em minus 0.4em\relax Cambridge University
  Press, 2016.

\bibitem{Mohammadali:TCOM:2024}
M.~Mohammadi, Z.~Mobini, H.~Q. Ngo, and M.~Matthaiou, ``Ten years of research
  advances in full-duplex massive {MIMO},'' \emph{{IEEE} Trans. Commun.},
  vol.~73, no.~3, pp. 1756--1786, Mar. 2025.

\bibitem{Souza:TVT:2021}
J.~H.~I. de~Souza, A.~Amiri, T.~Abrão, E.~de~Carvalho, and P.~Popovski,
  ``Quasi-distributed antenna selection for spectral efficiency maximization in
  subarray switching {XL-MIMO} systems,'' \emph{{IEEE} Trans. Veh. Technol.},
  vol.~70, no.~7, pp. 6713--6725, July 2021.

\bibitem{Liu:TWC:2024}
Z.~Liu, J.~Zhang, Z.~Liu, H.~Xiao, and B.~Ai, ``Double-layer power control for
  mobile cell-free {XL-MIMO} with multi-agent reinforcement learning,''
  \emph{{IEEE} Trans. Wireless Commun.}, vol.~23, no.~5, pp. 4658--4674, May
  2024.

\bibitem{Liu:JOP:2023}
Y.~Liu, Z.~Wang, J.~Xu, C.~Ouyang, X.~Mu, and R.~Schober, ``Near-field
  communications: A tutorial review,'' \emph{{IEEE} Open J. Commun. Society},
  vol.~4, pp. 1999--2049, Sept. 2023.

\bibitem{Zhi:2022:JSAC}
K.~Zhi, C.~Pan, G.~Zhou, H.~Ren, M.~Elkashlan, and R.~Schober, ``Is {RIS}-aided
  massive {MIMO} promising with {ZF} detectors and imperfect {CSI?}''
  \emph{{IEEE} J. Sel. Areas Commun.}, vol.~40, no.~10, pp. 3010--3026, Oct.
  2022.

\bibitem{Mohammadi:TC:2024}
M.~Mohammadi, H.~Q. Ngo, and M.~Matthaiou, ``Phase-shift and transmit power
  optimization for {RIS}-aided massive {MIMO SWIPT IoT} networks,''
  \emph{{IEEE} Trans. Commun.}, vol.~73, no.~1, pp. 631--647, Jan. 2025.

\bibitem{emil17}
E.~Bj\"{o}rnson, J.~Hoydis, and L.~Sanguinetti, ``Massive {MIMO} networks:
  Spectral, energy, and hardware efficiency,'' \emph{Found. Trends. Signal
  Process.}, vol.~11, no. 3-4, pp. 154--655, 2017.

\bibitem{2024:Mohammadi:survey}
M.~Mohammadi, Z.~Mobini, H.~Q. Ngo, and M.~Matthaiou, ``Next-generation
  multiple access with cell-free massive {MIMO},'' \emph{Proc. {IEEE}}, vol.
  112, no.~9, pp. 1372--1420, Sept. 2024.

\bibitem{Wu:TWC:2019}
Q.~Wu and R.~Zhang, ``Intelligent reflecting surface enhanced wireless network
  via joint active and passive beamforming,'' \emph{{IEEE} Trans. Wireless
  Commun.}, vol.~18, no.~11, pp. 5394--5409, Nov. 2019.

\bibitem{Yang:2020:TWC}
K.~Yang, T.~Jiang, Y.~Shi, and Z.~Ding, ``Federated learning via over-the-air
  computation,'' \emph{{IEEE} Trans. Wireless Commun.}, vol.~19, no.~3, pp.
  2022--2035, Mar. 2020.

\bibitem{Tao:2019:Glob}
T.~Jiang and Y.~Shi, ``Over-the-air computation via intelligent reflecting
  surfaces,'' in \emph{Proc. IEEE GLOBECOM}, Dec. 2019, pp. 1--6.

\bibitem{Beck:2010:JGlobOptim}
L.~T. A.~Beck, A. Ben-Tal, ``A sequential parametric convex approximation
  method with applications to nonconvex truss topology design problems,''
  \emph{J Glob Optim}, vol.~47, no.~1, p. 29–51, June. 2010.

\bibitem{cvx}
M.~Grant and S.~Boyd, ``{CVX}: Matlab software for disciplined convex
  programming, version 2.1,'' \url{https://cvxr.com/cvx}, Mar. 2014.

\bibitem{Mohammadi:JSAC:2023}
M.~Mohammadi, T.~T. Vu, H.~Q. Ngo, and M.~Matthaiou, ``Network-assisted
  full-duplex cell-free massive {MIMO}: Spectral and energy efficiencies,''
  \emph{{IEEE} J. Sel. Areas Commun.}, vol.~41, no.~9, pp. 2833--2851, Sept.
  2023.

\bibitem{Luo:MSP:2010}
Z.-Q. Luo, W.-K. Ma, A.~M.-C. So, Y.~Ye, and S.~Zhang, ``Semidefinite
  relaxation of quadratic optimization problems,'' \emph{{IEEE} Signal Process.
  Mag.}, vol.~27, no.~3, pp. 20--34, May 2010.

\bibitem{Tam:TWC:2017}
H.~H.~M. Tam, H.~D. Tuan, D.~T. Ngo, T.~Q. Duong, and H.~V. Poor, ``Joint load
  balancing and interference management for small-cell heterogeneous networks
  with limited backhaul capacity,'' \emph{{IEEE} Trans. Wireless Commun.},
  vol.~16, no.~2, pp. 872--884, Feb. 2017.

\bibitem{Zhi:2023:TIT}
K.~Zhi \emph{et~al.}, ``Two-timescale design for reconfigurable intelligent
  surface-aided massive {MIMO} systems with imperfect {CSI},'' \emph{{IEEE}
  Trans. Inf. Theory}, vol.~69, no.~5, pp. 3001--3033, May 2023.

\end{thebibliography}

\end{document}